\documentclass[a4paper,onecolumn,showpacs,notitlepage,floatfix,nofootinbib,superscriptaddress,prd]{revtex4-2}
\pdfoutput=1
\usepackage[utf8]{inputenc}

\usepackage[bookmarks, linktocpage, colorlinks = true, linkcolor = blue, urlcolor  = blue, citecolor = blue, anchorcolor = green, hyperindex = true, hyperfigures]{hyperref}

\usepackage[T1]{fontenc} 
\usepackage{bm} 

\usepackage{amsthm}

\usepackage{mathrsfs} 
\usepackage{subcaption}
\usepackage{mdframed}
\usepackage{enumitem}
\usepackage{xcolor}
\usepackage{subcaption}
\usepackage{mathtools}
\usepackage{tikz-cd}
\usepackage[normalem]{ulem}
\usepackage{amsmath}
\usepackage{amssymb}
\usepackage{amsfonts}

\usepackage{latexsym}

\usepackage{subcaption}
\usepackage{mdframed}
\usepackage{enumitem}
\usepackage{subcaption}
\usepackage{mathtools}
\usepackage{tikz-cd}
\usepackage{pict2e}
\newmdenv[skipabove=6mm]{kotak}   

\DeclareFontEncoding{LS1}{}{}
\DeclareFontSubstitution{LS1}{stix}{m}{n}
\DeclareSymbolFont{symbols2}{LS1}{stixfrak}{m}{n}
\DeclareMathSymbol{\lcangle}{\mathopen}{symbols2}{"E9}
\DeclareMathSymbol{\rcangle}{\mathclose}{symbols2}{"EA}

\makeatletter
\newcommand{\llcangle}[1][]{\savebox{\@brx}{\(\m@th{#1\lcangle}\)}%
  \mathopen{\copy\@brx\kern-0.5\wd\@brx\usebox{\@brx}}}
\newcommand{\rrcangle}[1][]{\savebox{\@brx}{\(\m@th{#1\rcangle}\)}%
  \mathclose{\copy\@brx\kern-0.5\wd\@brx\usebox{\@brx}}}
\makeatother

\newcommand{\subbe}{\begin{subequations}}
\newcommand{\subee}{\end{subequations}}
\newcommand{\be}{\begin{eqnarray}}
\newcommand{\ee}{\end{eqnarray}}
\newcommand{\nn}{\nonumber}
\newcommand{\nl}{\nonumber \\}
\newcommand{\pd}{\partial}

\newcommand{\Tr}{{\rm Tr}}

\newcommand{\End}{{\rm End}}
\newcommand{\rank}{{\rm rank}}

\newcommand{\appropto}{\mathrel{\vcenter{
  \offinterlineskip\halign{\hfil$##$\cr
    \propto\cr\noalign{\kern2pt}\sim\cr\noalign{\kern-2pt}}}}}

\newtheorem{proposition}{Proposition}
\newtheorem{lemma}{Lemma}

\newtheorem{remark}{Remark}

\newtheorem*{fact*}{Fact}

\makeatletter
\newsavebox{\@brx}
\newcommand{\llangle}[1][]{\savebox{\@brx}{\(\m@th{#1\langle}\)}%
  \mathopen{\copy\@brx\mkern2mu\kern-0.9\wd\@brx\usebox{\@brx}}}
\newcommand{\rrangle}[1][]{\savebox{\@brx}{\(\m@th{#1\rangle}\)}%
  \mathclose{\copy\@brx\mkern2mu\kern-0.9\wd\@brx\usebox{\@brx}}}
\makeatother

\allowdisplaybreaks

\begin{document}

\title{Algebraic Tomography of Non-Hermitian Floquet Systems \\
  from Observable Traces}

\author{Syo Kamata}
\email{skamata11phys@gmail.com}
\affiliation{Department of Physics, The University of Tokyo, 7-3-1 Hongo, Bunkyo-ku, Tokyo 113-0033, Japan}

\begin{abstract}
We formulate a framework of Floquet algebraic tomography for finite-dimensional non-Hermitian monodromy matrices from observable trace sequences $\zeta_n^{(O)}=\Tr(OM^n)$. Since these sequences are constrained by the characteristic polynomial of $M$, the inverse problem is a finite-dimensional algebraic reconstruction problem rather than a generic exponential fit. We organize the reconstruction through the observable resolvent, spectral determinant, and Dirichlet spectral data, separating the common spectral skeleton from observable-dependent dressing. Cayley--Hamilton and Hankel methods recover the similarity-invariant spectral data, while multi-observable and Liouville-space extensions connect the construction to realization theory and tomography reconstruction. 
We further clarify the limits of identifiability from restricted observable algebras: the data determine a visible representative, micromotion can enlarge the sampled visible operator space, and exact symmetries impose residual invisible sectors. Two examples, a driven transmon qutrit and a finite non-Hermitian Floquet SSH chain, demonstrate leakage-induced visibility expansion, observable-dependent phase response, EP-accessible branch geometry, and disorder/probe-dependent observable-dimension readouts.
\end{abstract}

\maketitle
\tableofcontents
\flushbottom

\section{Introduction} \label{sec:introduction}

Floquet systems provide a natural framework for describing time-periodic quantum dynamics through one-period evolution operators, or monodromy matrices~\cite{Shirley1965,Sambe1973,Howland1974}.
In recent years, non-Hermitian extensions of Floquet systems have attracted considerable attention because they combine characteristic features of driven dynamics with non-Hermitian spectral phenomena, including point gaps, exceptional points, non-Hermitian skin effects, anomalous Floquet topology, and strong boundary sensitivity~\cite{YaoWang2018,Gong2018,Borgnia2020,OkumaKawabataShiozakiSato2020,AshidaGongUeda2021,BergholtzBudichKunst2021,DingFangMa2022,ZhangGong2020,WuAn2020,Manna2022,ZhouZhang2023}.
These developments appear in condensed-matter, photonic, circuit, and quantum-control settings, where spectral, dynamical, and measurement viewpoints are often inseparable.
In this context, exact Wentzel--Kramers--Brillouin (WKB) analysis and resurgence theory provide rigorous tools to evaluate such monodromy matrices, capturing nonperturbative resonances in periodic systems~\cite{Sueishi_2021, Fujimori_2025} and analytically resolving exceptional points and $\mathcal{PT}$-symmetry breaking in non-Hermitian regimes~\cite{Bender:1998gh,Dorey:2009tc,Kamata2026_TripleWell}.

A basic difficulty is that the central Floquet object, the monodromy matrix $M$, is often not directly accessible.
What is available instead is a time-discrete response measured through a restricted observable channel.
Thus, the natural inverse problem is not simply to diagonalize a known matrix, but to ask what part of the underlying finite-dimensional Floquet dynamics can be reconstructed from observable data, and what part remains invisible to the chosen probe.
This question becomes especially delicate in non-Hermitian and finite-size systems, where near-degeneracies, exceptional-point-type branch structures, boundary sensitivity, and observable-dependent cancellations can all coexist.
The basic data used in this paper are the observable trace sequences
\be
\zeta_n^{(O)}:=\Tr(OM^n), \qquad n\in{\mathbb N}_0, \nonumber
\ee
where $M\in GL(N,{\mathbb C})$ is a finite-dimensional Floquet monodromy matrix and $O$ is a fixed observable.
At first sight, such a sequence resembles a finite exponential signal, and it is natural to compare the problem with classical spectral-estimation and realization methods, including Prony-type methods, the Matrix Pencil method, filter diagonalization, harmonic inversion, Ho--Kalman realization, ERA, ESPRIT, and subspace-identification approaches~\cite{Prony1795,HuaSarkar1990,WallNeuhauser1995FDM,MandelshtamTaylor1997HI,MainMandelshtamWunnerTaylor1997,HoKalman1966,Kailath1980,JuangPappa1985,Paulraj1985,VanOverscheeDeMoor1994,VanOverscheeDeMoor1996}.
The present work is not intended as a replacement for those methods.
Rather, our starting point is that a Floquet observable trace sequence is not a generic signal: it is generated by powers of a finite-dimensional monodromy matrix and is therefore constrained by exact Cayley--Hamilton identities.
This turns the inverse problem into a finite-dimensional algebraic reconstruction problem.

We call the resulting framework \emph{Floquet algebraic tomography}.
Its purpose is to reconstruct and organize the algebraic structure behind observable trace data.
The key separation is between the common spectral skeleton of the monodromy matrix and the observable-dependent dressing that determines how this skeleton becomes visible in a given channel.
This distinction is implemented through the observable resolvent, observable spectral determinant, and observable Dirichlet spectral data.
In particular, the observable resolvent separates a common denominator, determined by the characteristic polynomial of $M$, from an observable-dependent numerator.
This provides the analytic backbone for distinguishing what belongs to the monodromy matrix itself from what is produced by the choice of observable.

The present work is intended as a structural and inverse-problem contribution, rather than a proposal of new physical phenomena or measurement settings.
The OTS data $\zeta_n^{(O)} = \Tr (OM^n)$ are formally accessible through standard stroboscopic and process-tomography reconstruction already employed in non-Hermitian Floquet platforms~\cite{PoyatosCiracZoller1997,ChuangNielsen1997,MerkelGambettaSmolinPolettoCorcolesJohnsonRyanSteffen2013,BlumeKohoutGambleNielsenMizrahiSterkMaunz2013,Roushan_2017,Elben_2023}; what we develop here is the algebraic backbone that organizes such data into reconstruction-theoretic and visibility-theoretic statements.
The goal is therefore to identify, in a finite-dimensional setting where the question is sharply posable, what aspects of the underlying monodromy matrix are reachable from observable trace data and what aspects remain provably invisible, before turning to specific predictions in any concrete platform.

The reconstruction side of the framework uses Cayley--Hamilton recurrences and Hankel or block-Hankel data matrices.
From finite observable trace data, one can recover the effective recurrence order, characteristic coefficients, eigenvalues, and observable weights under suitable visibility conditions.
For sufficiently rich observable families, multi-observable realization produces a matrix $M_{\rm rel}$ similar to the original monodromy matrix, thereby recovering the similarity-invariant spectral skeleton.
We also formulate Liouville-space variants involving boundary or density operators, including fundamental-type sequences $\Tr(OM^n\Omega)$ and adjoint-type sequences $\Tr(OM^n\Omega M^{-n})$.
These variants connect the framework to MIMO realization, quantum process tomography, gate-set tomography, and Hamiltonian learning protocols ~\cite{PoyatosCiracZoller1997,ChuangNielsen1997,MerkelGambettaSmolinPolettoCorcolesJohnsonRyanSteffen2013,BlumeKohoutGambleNielsenMizrahiSterkMaunz2013,ZhangSarovar2014,ZhangSarovar2015,HolzapfelBaumgratzCramerPlenio2015,SchirmerOi2009,PastoriOlsacherKokailZoller2022,LiK2025}.
Another issue is identifiability from the observable algebra itself.
Even if finite trace data determine a reconstructed object, the full monodromy matrix need not be identifiable if the accessible observables generate only a proper subalgebra of $\End({\cal H})$.
This leads to a partial-realization problem naturally described by commutants and bicommutants of observable algebras~\cite{BratteliRobinson,KadisonRingrose1997,Takesaki2002,ZanardiLidarLloyd2004,ViolaKnillLaflamme2001}.
In this setting, the data determine a visible representative together with a trace-invisible remainder.
Micromotion provides a further mechanism for enlarging the observable algebra: shifting the Floquet reference time replaces $O$ by
\be
O(t) = U(t,0)^{-1} O U(t,0). \nn
\ee
While micromotion is familiar as a carrier of information beyond the stroboscopic Floquet operator in topological setups~\cite{RudnerLindnerBergLevin2013,Harper2020}, here it is used as an algebraic observability sweep.
Exact symmetries can nevertheless survive this extension and impose residual invisible sectors, in close relation to local observability and controllability in finite-dimensional quantum systems~\cite{AlbertiniDAlessandro2003,BurgarthBoseBruderGiovannetti2009}.

The framework is illustrated by two examples chosen to probe complementary regimes of the visibility/identifiability question.
The first, a driven transmon qutrit with coherent leakage and dissipation, motivated by weakly anharmonic superconducting qubits and leakage control~\cite{Koch2007Transmon,Motzoi2009DRAG}, realizes a setting in which the full operator algebra can be reached and the common spectral skeleton is reconstructed from a single observable channel to machine precision; this serves as a benchmark case showing that the algebraic backbone is sharp. It also demonstrates how observable-dependent dressing modifies the visible phase response, and how micromotion detects the expansion of the visible operator space when leakage makes the third level dynamically accessible.
The second, a finite non-Hermitian Floquet SSH chain, motivated by the classic Su--Schrieffer--Heeger physics~\cite{SuSchriefferHeeger1979} and its modern extensions to non-Hermitian regimes, non-reciprocal hopping, and non-Hermitian Floquet topology~\cite{YaoWang2018,Gong2018,ZhangGong2020,WuAn2020,ZhouZhang2023}, realizes the opposite regime, in which spatial locality, finite-size structure, and exact symmetries jointly restrict what is visible from natural probes; the same framework then quantifies the resulting invisible sectors in a structural manner rather than treating them as noise.
This example exhibits EP-accessible branch geometry, channel-selective winding-related readouts, and disorder/probe-dependent observable-dimension growth within the same algebraic language.

\subsection{Overview and our main contributions} \label{sec:overview_results}

We now summarize the internal structure of the framework.
The analysis proceeds in two intertwined directions.
The first is the realization-theoretic reconstruction of the finite-dimensional spectral skeleton from observable trace data.
The second is the spectral organization of how that skeleton becomes visible through a chosen observable channel.

The analytic components introduced in Sec.~\ref{sec:CH_GLN} are the observable trace sequence (OTS), observable resolvent (ORS), observable spectral determinant (OSD), and observable Dirichlet spectral data (ODSD).
They admit two complementary organizations:
\begin{widetext}
\be
\text{Formal spectral hierarchy}
&:&
\boxed{\text{OSD } h^{(O)}(z)} \longrightarrow \boxed{\text{ORS } {\cal Z}^{(O)}(z)} \longrightarrow \boxed{\text{ODSD } {\cal F}^{(O)}(p)} \nn \\[0.3em] \text{Inverse-problem order}
&:& \boxed{\text{OTS } \{\zeta_n^{(O)}\}_n} \longrightarrow \boxed{\text{ORS } {\cal Z}^{(O)}(z)} \longrightarrow
\begin{dcases}
\boxed{\text{OSD } h^{(O)}(z)}\\
\boxed{\text{ODSD } {\cal F}^{(O)}(p)}
\end{dcases} \nn
\ee
\end{widetext}
In this paper we mainly use the second organization: the OTS is the input, the ORS is reconstructed first, and the OSD/ODSD are then used as complementary descriptive components.
The central separation appears at the ORS level:
\be
{\cal Z}^{(O)}(z)
=
\frac{{\cal N}^{(O)}(z)}{\Delta(z)},
\qquad
\Delta(z)=\det({\mathbb I}_N-zM).
\nn
\ee
The denominator $\Delta(z)$ is common to all observable channels and carries the characteristic spectral skeleton of $M$.
The numerator ${\cal N}^{(O)}(z)$ carries the observable-dependent dressing.
Thus, the framework separates the observable-independent algebraic skeleton from the channel-dependent visible image.

The reconstruction side, developed in Sec.~\ref{sec:GLNstructure}, uses Cayley--Hamilton recurrences and Hankel or block-Hankel data matrices.
From finite OTS samples, one recovers the effective recurrence order, characteristic coefficients, eigenvalues, and observable weights under suitable visibility conditions.
The same characteristic skeleton also constrains nonconsecutive sampling patterns through the arbitrary-stepsize extension.
With sufficiently rich observable data, the multi-observable Hankel construction yields a realized monodromy matrix
\be
M_{\rm rel}\sim M, \nn
\ee
and hence recovers the similarity-invariant spectral skeleton.
Liouville-space extensions, including $\Tr (OM^n\Omega)$ and $\Tr (OM^n\Omega M^{-n})$, further connect the framework to MIMO realization, process tomography, and ratio-spectrum reconstruction.
When the observable family is sufficiently rich, this reconstruction procedure may be schematically summarized as
\begin{widetext}
\be
\boxed{\substack{\text{\normalsize Finite OTS samples} \\ \{\zeta_n^{(O)}\}_{n}}}
\ \xrightarrow{\ \,
  \substack{\text{CH recurrence} \\ \text{Hankel matrix}}
\ \,}
\begin{dcases}
  \boxed{\substack{\text{\normalsize Common spectral skeleton}\\ \{e_a\}_{a=1}^{N},\ \{\lambda_j\}_{j=1}^{N}, \ \Delta(z)}}
  \ \xrightarrow{\ \, \substack{\text{multi-observable} \\ \text{block-Hankel}} \ \, }\
\boxed{\boxed{M_{\rm rel}\sim M}}
\\
  \boxed{\substack{\text{\normalsize Observable-dependent dressing}\\ \{c_j^{(O)}\}_{j=1}^{N}, \ {\cal N}^{(O)}(z)}}
\end{dcases}  \nn
\ee
\end{widetext}
The Hankel matrix is the algebraic engine of this construction: from finite OTS samples it separates the common spectral skeleton from the observable-dependent dressing, and through its multi-observable extension delivers a realized monodromy matrix $M_{\rm rel}\sim M$.

The second structural aspect, formulated in Sec.~\ref{sec:prel_micro_symm}, concerns identifiability from the observable algebra.
If the accessible observables generate only a proper algebra ${\frak O}\subsetneq \End({\cal H})$, then the full monodromy matrix is not identifiable in general.
Instead, Prop.~\ref{prop:Alg_rel_M} shows that the trace data determine a canonical visible representative
\be
M_{\rm vis} = {\mathbb E}_{{\frak O}^{\prime\prime}}(M) \in {\frak O}^{\prime\prime}, \nn
\ee
together with a trace-invisible remainder.
Micromotion enlarges the observable family through
\be
O(t)=U(t,0)^{-1} O U(t,0), \nn
\ee
while Prop.~\ref{prop:symmetry_indicator} shows that exact commuting symmetries can still impose non-removable algebraic deficiencies.
The calibrated reconstruction of Sec.~\ref{sec:reconstruct_Mex} explains when the similarity ambiguity can be fixed in a known physical basis.

It is useful to keep four objects of recoverability in mind:
\begin{itemize}
\item[i)] \textbf{Exact object:} the full monodromy matrix and its exact spectral geometry.
\item[ii)] \textbf{Reconstructed object:} the algebraic object recovered from finite trace data, such as characteristic coefficients, poles, residues, or a realized monodromy matrix.
\item[iii)] \textbf{Visible object:} the part of the reconstructed structure that remains visible through a chosen observable channel or decoding functional.
\item[iv)] \textbf{Shadow object:} a reduced remnant, such as a branch-exchange pattern, that can remain accessible even when the full object is not reconstructed.
\end{itemize}
This distinction is essential in the examples, where common spectral structure, observable dressing, and finite-sampling visibility need not coincide.

The two examples illustrate this logic in complementary ways.
In the driven transmon qutrit of Sec.~\ref{sec:example_transmon}, a single observable reconstructs the common $GL(3,{\mathbb C})$ spectral skeleton.
The ODSD then shows how the visible phase response is filtered by observable dressing, and the sampled observable dimension
\be
D_{\rm obs}(Q):=\rank {\cal K}_0(Q)
\nonumber
\ee
detects whether coherent leakage expands the visible operator space from the closed qubit sector to the full qutrit space.
In the finite NHFSSH chain of Sec.~\ref{sec:example_ssh}, the EP-accessible extension supplies a branch-sensitive spectral skeleton; the OSD and twist-cycle readouts show how this skeleton becomes visible or suppressed in different observable channels.
The same sampled dimension $D_{\rm obs}(Q)$ is then used to study how symmetry, disorder, finite-size degeneracies, temporal sampling, and observable choice jointly control visibility.

Our main contributions can be grouped into three parts.
\begin{itemize}
\item \textbf{Framework-level contributions:}
we formulate Floquet algebraic tomography from observable trace sequences; separate common spectral skeletons from observable dressing through the ORS/OSD/ODSD organization; and formulate partial realization and micromotion-enlarged observability using the observable algebras ${\frak O}$ and ${\frak O}_{\rm ext}$.

\item \textbf{Technical reconstruction part:}
we use Cayley--Hamilton recurrences and Hankel-type data matrices to reconstruct characteristic data from finite OTS data; extend the recurrence to arbitrary $(N+1)$-point stencils; adapt multi-observable realization to Floquet OTSs; and formulate Liouville-space variants for boundary/density-operator and adjoint-type data.

\item \textbf{Demonstrations:}
we show leakage-induced growth of sampled visibility and observable-dependent phase response in the driven transmon qutrit, and demonstrate EP-accessible branch geometry, channel-selective winding-related readouts, and disorder/probe-dependent observable-dimension growth in the finite non-Hermitian Floquet SSH chain.
\end{itemize}

\paragraph*{Relation to existing approaches.}
The present framework utilizes subspace realization and matrix-pencil methods~\cite{Prony1795,HuaSarkar1990,WallNeuhauser1995FDM,MandelshtamTaylor1997HI,MainMandelshtamWunnerTaylor1997,HoKalman1966,Kailath1980,JuangPappa1985,Paulraj1985,VanOverscheeDeMoor1994,VanOverscheeDeMoor1996}.
While recent Hamiltonian/Liouvillian learning~\cite{ZhangSarovar2014,ZhangSarovar2015,HolzapfelBaumgratzCramerPlenio2015,SchirmerOi2009,PastoriOlsacherKokailZoller2022,LiK2025} and Floquet spectroscopy~\cite{Saurabh2025} share this inverse-problem perspective on observable traces, our objective is fundamentally distinct. 
Unlike approaches targeting generators $H$ or $\mathcal{L}$ under unitary/CPTP assumptions (e.g., Ref.~\cite{ZhangSarovar2014}) or exploiting rank deficiency merely for minimal realization, we target the non-Hermitian monodromy matrix $M$ itself via exact Cayley--Hamilton constraints. 
Building on operator algebras~\cite{BratteliRobinson,KadisonRingrose1997,Takesaki2002,ZanardiLidarLloyd2004,ViolaKnillLaflamme2001,AlbertiniDAlessandro2003,BurgarthBoseBruderGiovannetti2009}, the bicommutant $\mathfrak{D}^{\prime\prime}$ formulates this rank deficiency as a strict physical consequence of exact symmetries, formalizing identifiability through the decomposition $M_{\rm vis}={\mathbb E}_{{\frak O}^{\prime\prime}}(M)$. 
This repurposes algebraic realization to determine quantum observability limits and explicitly track how non-Hermitian winding structures, such as Riemann-sheet exchanges around exceptional points, are structurally extracted or suppressed. 
Broadly, this framework establishes a rigorous algebraic bridge between quantum tomography procedures~\cite{PoyatosCiracZoller1997,ChuangNielsen1997,MerkelGambettaSmolinPolettoCorcolesJohnsonRyanSteffen2013,BlumeKohoutGambleNielsenMizrahiSterkMaunz2013,Roushan_2017,Elben_2023} and the physics of non-Hermitian (Floquet) theory~\cite{AshidaGongUeda2021,BergholtzBudichKunst2021,YaoWang2018,Gong2018,Borgnia2020,OkumaKawabataShiozakiSato2020,DingFangMa2022,ZhangGong2020,WuAn2020,Manna2022,ZhouZhang2023,RudnerLindnerBergLevin2013,Harper2020}.

\section{Setup and analytic components} \label{sec:CH_GLN}

In this section, we introduce the analytic framework of the present work.
After describing the $GL(N,{\mathbb C})$ Floquet setup and the monodromy matrix in Sec.~\ref{sec:GL_floquet_M}, we define the observable trace sequence as the primary input and construct the observable spectral determinant, resolvent, and Dirichlet spectral data in Sec.~\ref{sec:CH_GLN_const}.

We note that the inversion of time-domain trace signals to systematically extract underlying eigenvalues has a rich history in molecular quantum dynamics, particularly under the framework of the filter diagonalization and harmonic inversion methods~\cite{WallNeuhauser1995FDM,MandelshtamTaylor1997HI}. 
In those approaches, linear recurrence relations are exploited to separate similarity-invariant spectral data from observable-dependent weights, providing a mathematical precedent for the exact algebraic constructions detailed below.

\subsection{Setup: $GL(N,{\mathbb C})$ Floquet monodromy} \label{sec:GL_floquet_M}
Throughout this paper, we consider an $N$-level Floquet system defined by
\be
i \pd_t \Psi(t) = \widehat{H}(t) \Psi(t), \qquad
\Psi(t) =
\begin{pmatrix}
\psi_1(t) \\
\psi_2(t) \\
\vdots \\
\psi_N(t)
\end{pmatrix},
\label{eq:FLQ_H}
\ee
where $\Psi(t)$ is the $N$-component state vector.
We expand the time-dependent Hamiltonian $\widehat{H}(t)$ in the basis consisting of ${\mathbb I}_N$ and traceless generators $\sigma_{j \in \{1,\cdots,N^2-1\}}$ of ${\frak su}(N)$ as
\be
\widehat{H}(t) = f_0(t)\,{\mathbb I}_N + \sum_{j=1}^{N^2-1} f_j(t)\sigma_j,
\qquad
\widehat{H}(t)=\widehat{H}(t+T),
\qquad
f_j(t)=f_j(t+T)\in{\mathbb C},
\label{eq:Hamiltonian}
\ee
with the period $T=2\pi/\omega \in {\mathbb R}_{>0}$.
Notice that the Hamiltonian is Hermitian if $f_j(t) \in {\mathbb R}$ for all $j \in \{ 0,\cdots,N^2-1 \}$.
In order to capture the Floquet physics, we define the time evolution operator $U(t, t_0) \in GL(N, \mathbb{C})$ with an initial time $t=t_0$ as
\be
U(t, t_0) := {\cal T} \exp \left[ -i \int_{t_0}^{t} dt^\prime \, \widehat{H}(t^\prime) \right], \label{eq:time_evo_U}
\ee
where ${\cal T}$ denotes the time-ordered product.
The dynamics of the system is governed by the \textit{monodromy matrix} $M \in GL(N, \mathbb{C})$, which is defined as the evolution operator over one full period $T$ as
\be
M := U(T, 0), \label{eq:M_def}
\ee 
and we assume that $\det M \ne 0$. 
The effective Hamiltonian can be obtained from Eq.\eqref{eq:M_def} as
\be
\widehat{H}_{\rm eff} = \frac{i}{T}\log M,
\ee
and, correspondingly, the state after $n$ periods is given by $\Psi(nT) = M^n \Psi(0)$.

Then, we construct the eigenvectors of the monodromy matrix, $M \in GL(N,{\mathbb C})$.
When $M$ is diagonalizable, let $\{\lambda_j\}_{j=1}^N$ be its eigenvalues, which we write as $\lambda_j=e^{i\varepsilon_j}$.
We define the left and right eigenvectors of $M$ as
\be
M | u_j \rangle = \lambda_j | u_j \rangle, \qquad \langle \widetilde{u}_j | M = \langle \widetilde{u}_j | \lambda_j
. \qquad j \in \{ 1, \cdots, N \} \label{eq:lr_states}
\ee
The eigenvectors satisfy
\be
\langle \widetilde{u}_{j_1} | u_{j_2} \rangle = \delta_{j_1,j_2}, \label{eq:states}
\ee
with the Kronecker delta, $\delta_{j_1,j_2}$, and one can define the projectors $ {P}_j$ with $j \in \{1,\cdots,N \}$ as
\be
   {P}_j := | u_j \rangle \langle \widetilde{u}_j|,  \qquad P_j P_k = \delta_{j,k} P_j, \qquad \sum_{j = 1}^N {P}_j = {\mathbb I}_N. \label{eq:proj}
\ee
Thus, $M$ can be expressed using the projectors as
\be
\qquad M = \sum_{j=1}^N \lambda_j {P}_j. \label{eq:MPj}
\ee
Notice that for the $U(N)$ cases, $| u_j \rangle^\dagger = \langle \widetilde{u}_j|$, but in general this condition is not satisfied for $GL(N,{\mathbb C})$ (and also for $SL(N,{\mathbb C})$).

If $M$ is non-diagonalizable, as in the vicinity of exceptional points (EPs)~\cite{Kato1995,DingFangMa2022}, then, one must instead employ the Jordan normal form.
At an EP, two or more eigenvalues and their corresponding eigenvectors degenerate, and thus, the geometric multiplicity is reduced.
In such cases, the biorthogonality condition in Eq.\eqref{eq:states} and the resolution of identity in Eq.\eqref{eq:proj} are no longer valid in the ordinary eigenbasis, because the set of eigenvectors fails to span the full $N$-dimensional space.
To describe the system at an EP, one has to use generalized eigenvectors and Jordan blocks.
Schematically, a Jordan block takes the form
\be
J_{\lambda_j} =
\begin{pmatrix}
\lambda_j & 1 &  \cdots & 0\\
0 & \lambda_j &  \ddots & \vdots\\
\vdots & \ddots & \ddots &  1\\
0 & \cdots & 0 & \lambda_j  
\end{pmatrix},
\ee
where the off-diagonal $1$'s represent the coupling between generalized eigenvectors.
The generalized eigenvectors are defined by
\be
 M | u_{j,1} \rangle = \lambda_j  | u_{j,1} \rangle, &\qquad& M | u_{j,m} \rangle = \lambda_j  | u_{j,m} \rangle + | u_{j,m-1} \rangle \ \ \text{for \ \ $m \in \{ 2,\cdots, d_j\}$},  \\
 \langle  \widetilde{u}_{j,d_j} | M = \langle  \widetilde{u}_{j,d_j} | \lambda_j, &\qquad& \langle \widetilde{u}_{j,m}
 | M = \langle \widetilde{u}_{j,m} | \lambda_j + \langle \widetilde{u}_{j,m+1} | \ \ \text{for \ \ $m \in \{ 1,\cdots, d_j-1\}$},
\ee
with the orthogonality condition as
\be
\langle \widetilde{u}_{j_1,m_1} | u_{j_2,m_2} \rangle = \delta_{j_1,j_2} \delta_{m_1,m_2}, \quad j_{1,2} \in \{ 1,\cdots, K \}, \ \ m_{1,2} \in \{ 1,\cdots,d_j \}
\ee

where $d_j\in{\mathbb N}$ is the algebraic multiplicity associated with $\lambda_j$, satisfying $\sum_{j=1}^{K}d_j=N$.
Here, $\{\lambda_j\}_{j=1}^{K}$ is the set of distinct eigenvalues with $\lambda_{j_1} \ne \lambda_{j_2}$ for any $j_1 \ne j_2$.
By introducing the projection operators onto the generalized eigenspace and the nilpotent matrices associated with $\lambda_j$, $P_j$ and $N_j$, these operators identically satisfy the following algebraic relations as
\be
&& P_j := \sum_{m=1}^{d_j} | u_{j,m} \rangle \langle \widetilde{u}_{j,m}|, \qquad  N_j := \sum_{m=1}^{d_j-1} | u_{j,m} \rangle \langle \widetilde{u}_{j,m+1} | \nl
&& \sum_{j=1}^K P_j = {\mathbb I}_N
, \qquad P_j P_k = \delta_{j,k} P_j,
\label{eq:Jordan_prop_N} \\
&& N_j^{d_j} = 0, \qquad P_j N_k = N_k P_j = \delta_{j,k} N_j, \qquad N_j N_k = 0 \ \ \text{for} \ \  j \neq k, \nn
\ee
and thus, the matrix $M \in GL(N, \mathbb{C})$ can be written by the spectral decomposition (or Jordan decomposition) in terms of the set of distinct eigenvalues $\{\lambda_j\}_{j=1}^K$ as
\be
M = \sum_{j=1}^K (\lambda_j P_j + N_j). \label{eq:M_EP_P_N}
\ee

It is straightforward to pass from $GL(N,{\mathbb C})$ to $SL(N,{\mathbb C})$, and this reduction might be sometimes convenient in numerical computations.
The normalized $SL(N,{\mathbb C})$ matrix is defined by
\begin{equation}
\bar{M} := (\det M)^{-1/N} M \in SL(N,\mathbb{C}), \label{eq:barM_def_SLN}
\qquad \det \bar{M} = 1.
\end{equation}
The spectrum and the quasi-energies, $\{\bar{\lambda}_j \}_{j=1}^N$ and $\{ \bar{\varepsilon}_j \}_{j=1}^N$, are constrained as $\prod_{j=1}^N \bar{\lambda}_j=1$ and $\sum_{j=1}^N \bar{\varepsilon}_j \equiv 0 \ (\mathrm{mod}\ 2\pi)$, respectively.
The original $GL(N,{\mathbb C})$ spectral data can be recovered from the $SL(N,{\mathbb C})$ reduction by rescaling
\be
\lambda_j = g_0 \cdot \bar{\lambda}_j \quad \text{for all} \quad j \in \{ 1, \cdots N\},
\ee
where $g_0$ is the central factor of the $GL(N,{\mathbb C})$ monodromy defined as
\be
g_0 :=  \exp \left[ -  \frac{i}{N} \Tr \int_{0}^T dt \,  \widehat{H} (t) \right] = \exp \left[ -i \int_{0}^T dt \, f_0(t) \right] = (\det M)^{1/N},
\ee
and carries the global determinant phase.
The $SL(N,{\mathbb C})$ reduction removes this central sector and isolates the projective or relative spectral structure.

\subsection{Observable spectral determinant, resolvent, and Dirichlet spectral data} \label{sec:CH_GLN_const}
In this part, we provide definition of the observable trace sequence and construct the observable spectral determinant, resolvent, and Dirichlet spectral data in Sec.~\ref{sec:framework}.
Then, we describe their expression using eigenvalues in Sec.~\ref{sec:eigen_express}.

\subsubsection{Construction} \label{sec:framework}

We begin with the characteristic polynomial of the monodromy matrix $M\in GL(N,{\mathbb C})$, defined by
\be
\Delta(z):=\det({\mathbb I}_N - z M) =\sum_{a=0}^{N} (-1)^a e_a\, z^a, \qquad (z \in {\mathbb C}) \label{eq:Delta_def}
\ee
where $\{ e_a \}_{a=0}^N$ are the characteristic coefficients as
\be
e_0 := 1, \qquad e_a := \Tr \left( \wedge^a M \right) = \sum_{1 \le j_1 < \cdots < j_a \le N} \lambda_{j_1} \cdots \lambda_{j_a} \quad \text{for} \quad  a \in \{ 1,\dots,N \}. \label{eq:ea_wedge}
\ee
In our framework, the $GL(N,{\mathbb C})$ structure is controlled by $\{ e_a\}_{a=1}^{N}$.
Notice that $e_N = \det M$ for $M \in GL(N,{\mathbb C})$, and $e_N = 1$ for $\bar{M} \in SL(N,{\mathbb C})$.
Then, we consider the Cayley--Hamilton (CH) theorem which gives
\be
  \sum_{a=0}^{N} (-1)^a e_a\, M^{N-a}=0,
\ee
or, equivalently, in the shifted form as
\be
  \sum_{a=0}^{N} (-1)^a e_a\, M^{n+N-a}=0  \quad \text{for} \quad n \in {\mathbb N}_0. \label{eq:CH_matrix_recurrence}
\ee
For any non-zero Hermitian observable operator $O \in {\rm End}({\cal H})$ acting on a finite dimensional Hilbert space ${\cal H} \cong {\mathbb C}^N$, we define the \textit{observable trace sequence (OTS)}, $\{ \zeta_n^{(O)} \}_{n \in {\mathbb N}_0}$, as
\be
  \zeta_n^{(O)} := \Tr \!\left(O M^{n}\right) \quad \text{for} \quad n \in {\mathbb N}_0. \label{eq:OTS_def}
\ee  
Applying the linear functional $X\mapsto \Tr (OX)$ to Eq.\eqref{eq:CH_matrix_recurrence} yields the CH recurrence induced linear recurrence as
\be
&& \sum_{a=0}^N  (-1)^a e_{a} \zeta^{(O)}_{n+N-a} = 0 \quad \text{for} \quad n \in {\mathbb N}_0. \label{eq:CH_bn_recurrence}
\ee
Throughout the basic OTS construction, we assume that the observable is non-zero and independent of the discrete period index $n$.
The OTS \eqref{eq:OTS_def} has the role of a primary input in our framework.

An important property of the OTS is that it is \textit{not} gauge-invariant in general, where here ``gauge'' refers to the shift of the reference time in the monodromy, $U(T, 0) \mapsto U(T+t,t)$, in Eq.\eqref{eq:M_def}.
By the definition of the monodromy matrix in Eq.\eqref{eq:M_def}, the time-dependence of the OTS can be written as
\be
&& M_t := U(T+t,t) = U(t,0)  M U^{-1}(t,0), \qquad t \in [0,T) \nl
&\Rightarrow \quad &\zeta^{(O)}_{n,t} := \Tr(O M_t^n) = \Tr(O(t) M^n), \qquad O(t):= U^{-1}(t,0) O U(t,0), \label{eq:def_M(t)}
\ee
and we used the periodicity of the time-evolution as $U(T+t,T) = U(t,0)$.
Thus, the OTS is invariant under the gauge transformation when $O \propto {\mathbb I}_N$.
Since the spectrum, $\{\lambda_j\}_j$, is a gauge independent object as $\Tr(M^n_t) = \Tr(M^n_{t^\prime})$ for $t \ne t^\prime$ , the gauge transformation can be pushed into the observable as $O(t)$ in $\zeta_{n,t}^{(O)}$.
We will mainly deal with $M = M_{t=0}$ in this paper, but this property has an important role for the micromotion and symmetry, which we will discuss in Sec.~\ref{sec:micro_symm}.


For investigating properties of the monodromy matrix $M$ and visibility by the observable $O$ from the OTS, we introduce analytic components used in spectral theory.
We define a generating function $h^{(O)}(z)$, which we call as \textit{observable spectral determinant (OSD)}, as
\be
&& h^{(O)}(z) := \exp \left[  - \Tr \left( O \log ( {\mathbb I}_N-z M ) \right) \right]. \qquad (z \in {\mathbb C}) \label{eq:def_hz}
\ee
From the OSD, we introduce \textit{observable resolvent (ORS)} (or \textit{transfer function}) ${\cal Z}^{(O)}(z)$ and \textit{observable Dirichlet spectral data (ODSD)} ${\cal F}^{(O)}(p)$ as\footnote{If singularities of ${\cal Z}^{(O)}(e^{-\rho})$ lie on the positive real $\rho$-axis, the integral representation requires a deformation of contour in the ODSD.
In this case, one can perform the integration by replacing with the Hankel integration, i.e.,
\be
{\cal F}^{(O)}(p) := \frac{1}{2 \pi i \Gamma(p)} \int_{\gamma_0}  d \rho \, e^{-\rho} \rho^{p-1} {\cal Z}^{(O)}(e^{-\rho}), \qquad (p \in {\mathbb C}) 
\ee
where $\gamma_0$ denotes the Hankel contour going around the origin as avoiding singularities on the real axis.
}
\be
&&  {\cal Z}^{(O)}(z) := \pd_z \log h^{(O)}(z) = - \pd_z \Tr \left[ O \log ( {\mathbb I}_N - z M ) \right], \qquad (z \in {\mathbb C}) \label{eq:ORS_def} \\
&&  {\cal F}^{(O)}(p) := \frac{1}{\Gamma(p)} \int_0^{+\infty}  d \rho \, e^{-\rho} \rho^{p-1} {\cal Z}^{(O)}(e^{-\rho}). \qquad (p \in {\mathbb C}) \label{eq:ODSD_def} 
\ee
The ORS can be recovered from the ODSD as
\be
{\cal Z}^{(O)}(e^{-\rho}) =\frac{e^{+\rho}}{2 \pi i}  \int^{c+i \infty}_{c-i \infty} dp \, \rho^{- p} \Gamma(p) {\cal F}^{(O)}(p), \label{eq:ODSDtoORS}
\ee
where $c \in {\mathbb R}_{>0}$ has to be taken as $c > \max \, \{ {\rm Re} (p^{\rm pole}_n)\}_n $ with a sequence of poles of $\Gamma(p) {\cal F}^{(O)}(p)$.
These can be directly expressed by the OTS \eqref{eq:OTS_def} as
\be
&& h^{(O)}(z) = \exp \left[ \sum_{n \in {\mathbb N}} \frac{\zeta^{(O)}_n}{n} z^{n} \right], \label{eq:h_zeta}
\ee
for the OSD, and 
\be
   {\cal Z}^{(O)}(z) = \sum_{n \in {\mathbb N}} \zeta^{(O)}_n z^{n-1}, \qquad {\cal F}^{(O)}(p) = \frac{1}{\Gamma(p)} \sum_{n \in {\mathbb N}} \int_0^{+\infty}  d \rho \, \rho^{p-1} \zeta^{(O)}_n e^{-\rho n}  = \sum_{n \in {\mathbb N}} \frac{\zeta^{(O)}_n}{n^p}, \label{eq:ZF_zeta}
\ee
for the ORS and the ODSD, respectively.
By denoting the OSD as
\be
h^{(O)}(z) = \sum_{n \in {\mathbb N}_0} h^{(O)}_nz^n, \label{eq:OSD_coeff}
\ee
taking derivative to $h^{(O)}(z)$ in Eq.\eqref{eq:h_zeta} yields a recurrence relation of $h^{(O)}_n$ as
\be
h^{(O)}_n = \frac{1}{n} \sum_{k=1}^n \zeta^{(O)}_k h^{(O)}_{n-k} \quad \text{for} \quad n \in {\mathbb N}, \qquad h^{(O)}_{0} = 1, \qquad h_{n}^{(O)} = 0 \quad \text{for} \quad n \in {\mathbb Z}_{<0}, \label{eq:def_hz_der}
\ee
which is practically useful.

The ORS \eqref{eq:ORS_def} can be decomposed into observable-independent and dependent parts, which is the crucial mathematical substantiation of \textit{skeleton/dressing decomposition} in our framework.
Since the inverse matrix of $A \in {\mathbb C}^{N\times N}$ with $\det A \ne 0$ can be expressed by the adjugate matrix as $A^{-1} = \frac{{\rm adj}\, A}{\det A}$, one can express the ORS in Eq.\eqref{eq:ORS_def} as\footnote{
Algebraically, invisible sectors correspond to common factors between ${\cal N}^{(O)}(z)$ and $\Delta(z)$, namely to pole-zero cancellation in the ORS.
These invisible sectors and shadow objects are related but distinct.
The former refer to observable-dependent loss of spectral components, for example through pole-zero cancellation or vanishing residues, whereas the latter denote reduced but stable remnants of structure that remain accessible even when the full object is not reconstructed.}
\be
&& {\cal Z}^{(O)}(z) = \Tr \left[ O M ( {\mathbb I}_N -z M)^{-1} \right] = \frac{{\cal N}^{(O)}(z)}{\Delta(z)}, \\
&& {\cal N}^{(O)}(z) := \Tr \left[ O M \cdot {\rm adj} ( {\mathbb I}_N -z M) \right] =  \sum_{n=0}^{N-1} {\cal N}^{(O)}_n z^n,
\ee
with the characteristic polynomial $\Delta(z)$ in Eq.\eqref{eq:Delta_def}, and the coefficients, ${\cal N}^{(O)}_n \in {\mathbb C}$, are defined as
\be
{\cal N}^{(O)}_n := \sum_{a=0}^n (-1)^a e_a \zeta^{(O)}_{n-a+1} \quad \text{for} \quad n \in \{0, \cdots, N-1 \}.
\ee
This implies that all ORSs generated from the same monodromy matrix share the common denominator $\Delta(z)$, while the choice of the observable $O$ affects only the numerator data ${\cal N}^{(O)}(z)$.

Although the formal spectral hierarchy may be written as
${\rm OSD} \rightarrow  {\rm ORS} \rightarrow {\rm ODSD}$,  the inverse-problem viewpoint used in this paper reverses the practical order of use, ${\rm OTS} \rightarrow {\rm ORS} \rightarrow \{{\rm OSD}, {\rm ODSD}\}$.
Starting from the OTS, we first reconstruct the ORS as the transfer-function-like object whose denominator gives the common spectral skeleton and whose numerator carries the observable dressing.
The OSD and ODSD are then used as complementary descriptive components in this practical order of use.
Their roles in the present analysis are summarized as follows:
\begin{itemize}
    \item \textbf{ORS ${\cal Z}^{(O)}(z)$}: 
    This is the primary object reconstructed from the OTS. Its denominator $\Delta(z)$ encodes the observable-independent spectral skeleton, while its numerator ${\cal N}^{(O)}(z)$ encodes the observable-dependent dressing. Thus, the ORS is the reconstruction hub of the framework.

    \item \textbf{OSD $h^{(O)}(z)$}: 
    This object packages the same OTS data in multiplicative form. It is useful for describing global spectral structures, such as determinant-phase accumulation and characteristic-polynomial constraints, after the common skeleton has been identified.

    \item \textbf{ODSD ${\cal F}^{(O)}(p)$}: 
    This object provides a Dirichlet/polylogarithmic readout of the reconstructed spectral data. It is used to examine how the common skeleton becomes visible, suppressed, or reshaped in a chosen observable channel.
\end{itemize}

\subsubsection{Eigenvalues expressions} \label{sec:eigen_express}
For the later analysis, it is convenient to explicitly express the spectral theoretical analytic components, OSD, ORS, and ODSD, using the spectrum of $M$, $\{\lambda_j \}_{j=1}^N$.
In generic non-Hermitian theory, the spectral structure becomes more nontrivial than the Hermitian cases.
The situations can be classified by three cases: simple spectra, diabolic points, and exceptional points.

We first consider $M$ diagonalizable with \textit{simple spectra (SS)}, i.e., $\lambda_{j_1} \ne \lambda_{j_2}$ for all $j_1 \ne j_2$.
In this case, the OTS \eqref{eq:OTS_def} is written down as
\be
\zeta_n^{(O)} = \sum_{j=1}^N c^{(O)}_j \lambda_j^n, \qquad c^{(O)}_j = \Tr (O {P}_j) = \langle \widetilde{u}_j| O | u_j \rangle, \label{eq:OTS_lam}
\ee
where ${P}_{j \in \{ 1,\cdots, N\}}$ denote the projectors in Eq.\eqref{eq:proj}.
The coefficients $c^{(O)}_j$ describe the effect of observable-dressing, and $c^{(O)}_j = 1$ for all $j \in \{ 1, \cdots, N\}$ when the observable is trivial, $O={\mathbb I}_N$, since $c^{({\mathbb I}_N)}_j = \Tr (P_j) = \langle \widetilde{u}_j| u_j \rangle = 1$ by the biorthogonality condition in Eq.\eqref{eq:states}.
From Eq.\eqref{eq:OTS_lam}, the ORS in Eq.\eqref{eq:ZF_zeta} takes the form as
\be
{\cal Z}^{(O)}(z) = \sum_{j=1}^N \frac{c^{(O)}_j \lambda_j}{1-\lambda_jz}, \label{eq:ORS_lam}
\ee
and the OSD and the ODSD can be obtained as
\be
&& h^{(O)}(z) = \prod_{j=1}^N ( 1-\lambda_jz)^{-c^{(O)}_j}, \qquad {\cal F}^{(O)}(p) = \sum_{j = 1}^{N} c^{(O)}_j {\rm Li}_p(\lambda_j), \label{eq:hF_lam}
\ee
respectively, where ${\rm Li}_p(x)$ is the polylogarithm function.

For a diagonalizable $M$, degeneracies might appear in the spectrum, known as \textit{diabolic points (DPs)}, depending on details of a system.
In the case of diabolic points, the spectral degeneracy does not break the diagonalizability of the matrix.
Thus, the OTS is slightly modified from the SS case as
\be
\zeta_n^{(O)} = \sum_{j=1}^K \widetilde{c}_{j}^{(O)}
\lambda_j^n, \label{eq:OTS_lam_dps}
\ee
where $\{ \lambda_j\}_{j=1}^K$ with $K < N$ is a maximal subset of the eigenvalues with $\lambda_{j_1} \ne \lambda_{j_2}$ for any $j_1 \ne j_2$.
Here, the effective coefficient $\widetilde{c}^{(O)}_{j}$ is the sum of the residues associated with the degenerate subspace of $\lambda_j$ as
\be
\widetilde{c}^{(O)}_{j} = \sum_{k \in {\cal S}_j} c^{(O)}_k, \label{eq:cOj_ren}
\ee
where ${\cal S}_j$ denotes the set of indices of eigenvectors belonging to the same eigenvalue $\lambda_j$.
Since no Jordan blocks appear, the ORS continues to possess only simple poles, remaining formally consistent with the SS case in Eq.\eqref{eq:ORS_lam}:
\be
{\cal Z}^{(O)}(z) = \sum_{j=1}^K \frac{\widetilde{c}^{(O)}_{j} \lambda_j}{1-\lambda_j z}, \label{eq:ORS_lam_dps}
\ee
where $K < N$ is the number of distinct eigenvalues.
Similarly, the OSD and the ODSD are obtained as
\be
&& h^{(O)}(z) = \prod_{j=1}^K ( 1-\lambda_jz)^{-\widetilde{c}^{(O)}_{j}},  \qquad {\cal F}^{(O)}(p) = \sum_{j = 1}^{K} \widetilde{c}^{(O)}_{j} {\rm Li}_{p}(\lambda_j), \label{eq:hF_lam_dps}
\ee
respectively.

Let us finally consider the case where $M$ is non-diagonalizable, i.e., \textit{exceptional points (EPs)}.
Unlike the SS and DP cases where a complete basis of eigenvectors always exists, at an EP, the eigenvectors themselves collapse alongside the eigenvalues.
In consequence, the matrix becomes non-diagonalizable, and the use of generalized eigenvectors is required.
According to Eqs.\eqref{eq:Jordan_prop_N}\eqref{eq:M_EP_P_N}, the OTS \eqref{eq:OTS_lam} changes the form as
\be
&& \zeta_n^{(O)} = \sum_{j=1}^K \sum_{m=1}^{d_j} c^{(O)}_{j,m}(\lambda_j)
\begin{pmatrix}
  n+m-2 \\
  m - 1
\end{pmatrix} \lambda_j^{n}, \label{eq:OTS_lam_eps}
\ee
with the coefficients defined by
\be
&& c^{(O)}_{j,m}(\lambda_j) :=
\begin{dcases}
  \Tr(O P_j) & \quad \text{for} \quad m=1 \\
\sum_{k=m-1}^{d_j-1} (-1)^{k-m+1}
\begin{pmatrix}
  k-1 \\
  m-2
\end{pmatrix}
\Tr(O N_j^{k}) \lambda_j^{-k}  & \quad \text{for} \quad m \ge 2
\end{dcases}, \label{eq:c_jm_def} 
\ee
where $d_j$ for $d_j \ge 2$ represents the algebraic multiplicity (specifically, the size of the Jordan block) associated with the coalesced eigenvalue $\lambda_j$.

Notice that in the EP case, polynomial growth terms ($n^{r-1}\lambda^n$) appear in the OTS.
The ORS in Eq.\eqref{eq:ORS_lam} contains $d_j$-ple poles as
\be
   {\cal Z}^{(O)}(z) = \sum_{j=1}^K \sum_{m=1}^{d_j} \frac{c^{(O)}_{j,m}(\lambda_j) \lambda_j}{(1-\lambda_j z)^m},    \label{eq:ORS_lam_eps}
\ee
and the OSD and the ODSD are modified as
\be 
&& h^{(O)}(z) = \prod_{j=1}^{K} (1-\lambda_j z)^{-c_{j,1}^{(O)}(\lambda_j)} \prod_{m=2}^{d_j} \exp\left[ \frac{c_{j,m}^{(O)}(\lambda_j)}{m-1} \left\{ \frac{1}{(1-\lambda_j z)^{m-1}} - 1 \right\} \right], \\
&& {\cal F}^{(O)}(p) = \sum_{j = 1}^{K} \sum_{m=1}^{d_j} \sum_{r=1}^m \frac{c^{(O)}_{j,m}(\lambda_j)}{\Gamma(m)}
\begin{bmatrix}
  m-1 \\
  r-1
\end{bmatrix}
    {\rm Li}_{p-r+1}(\lambda_j), \label{eq:hF_lam_eps}    
\ee
respectively, where $\begin{bmatrix} n \\ m \end{bmatrix}$ denotes the Stirling numbers of the first kind.
For the EP cases, the ODSD regularizes the polynomial growth $n^{r-1}$ via the index shift of the polylogarithm, rendering the sequence well-defined even for large $n$.

It is worth noting that the coefficients, $c^{(O)}_j$, $\widetilde{c}^{(O)}_{j}$, and $c^{(O)}_{j,m}$ in Eqs.\eqref{eq:OTS_lam}\eqref{eq:OTS_lam_dps}\eqref{eq:OTS_lam_eps}, 
are constrained by the initial term of the OTS, $\zeta^{(O)}_0 = \Tr \, O$. Specifically, they satisfy 
\be
\zeta^{(O)}_0  =
\begin{cases}
  \sum_{j=1}^N c_j^{(O)} & \quad \text{for SS} \\
  \sum_{j=1}^K \widetilde{c}_{j}^{(O)} & \quad \text{for DPs} \\
  \sum_{j=1}^K {c}_{j,1}^{(O)}(\lambda_j) & \quad \text{for EPs}
\end{cases}.
\ee
In the trivial case where the observable is the identity, $O = {\mathbb I}_N$, these coefficients directly reflect the algebraic multiplicity of the eigenvalues. For a diagonalizable matrix, $c^{({\mathbb I}_N)}_j = 1$ for all $j \in \{ 1,\cdots, N \}$. For a non-diagonalizable matrix, the higher-order contributions strictly vanish, i.e., $c^{({\mathbb I}_N)}_{j,m} = 0$ for $m \ge 2$, while the simple pole coefficients are determined by the algebraic multiplicity of each eigenvalue, $c^{({\mathbb I}_N)}_{j,1} = d_j$. 
These nontrivial values of the observable-dependent coefficients fundamentally affect the resulting spectral objects.
In particular, while the trivial choice $O= {\mathbb I}_N$ reduces the OSD to the standard characteristic polynomial governing the intrinsic spectral roots, a general probe-observable $O$ dresses these coefficients, inducing a much richer geometric structure of roots and poles on the complex $z$-plane.

\subsection{Remark: T-system-like bilinear structure and observable-induced deformation}

We briefly comment on a T-system-like structure naturally associated with the OSD coefficients.
Although the standard finite-type $A_{N-1}$ T-system is formulated for $SL(N,{\mathbb C})$ without an observable insertion (e.g. see Ref.~\cite{KunibaNakanishiSuzuki2011}), it is useful in the present framework to keep both the $GL(N,{\mathbb C})$ determinant sector and the observable dependence explicit.

From the OSD coefficients $h_n^{(O)}$ in Eq.\eqref{eq:OSD_coeff}, we define
\be
H_{a,n}^{(O)} := \det_{1\le \ell,k\le a}\!\left(h^{(O)}_{\,n+\ell-k}\right), \qquad a \in \{0, \cdots, N \}, \quad n\in{\mathbb N}_0,
\label{eq:def_Han_SLN}
\ee
with the conventions $H_{a,n}^{(O)}=0$ for $n<0$ and $H_{0,n}^{(O)}=1$.
Then these variables satisfy the bilinear relation defined as 
\be
(H^{(O)}_{a,n})^2 - H^{(O)}_{a,n+1}H^{(O)}_{a,n-1} = H^{(O)}_{a+1,n}H^{(O)}_{a-1,n}, \qquad a \in \{1,\cdots,N-1\}, \quad n \in {\mathbb N}_0, \label{eq:HOn_bilinear_rec}
\ee
with the boundary conditions $H^{(O)}_{0,n}=1$ and $H^{(O)}_{a,-1}=0$.
The upper boundary, however, is deformed by both the determinant sector and the observable insertion as
\be
H^{(O)}_{N,n} = (\det M)^n \bigl(1+\xi_n^{(O)}\bigr), \qquad n \in {\mathbb N}_0, \label{eq:HOn_top_boundary}
\ee
where $\xi_n^{(O)}$ measures the observable-induced deviation and satisfies $\xi_n^{({\mathbb I}_N)}=0$ for all $n \in{\mathbb N}_0$.
Thus, after reduction to $SL(N,{\mathbb C})$ and for the trivial observable $O={\mathbb I}_N$, one recovers the standard source-free $A_{N-1}$ closure at the top boundary.
Furthermore, these determinants satisfy a linear recurrence relation governed by the characteristic data.
Defining the rank-$a$ characteristic polynomial $\Delta_a(z)$ and its coefficients $E_r^{(a)}$ by
\be
\Delta_a(z) := \prod_{|J|=a}(1-\Lambda_J z) = \sum_{r=0}^{d_a}(-1)^r E_r^{(a)} z^r, \qquad \Lambda_J:=\prod_{j\in J}\lambda_j,
\ee
where $J\subset\{1,\cdots,N\}$, $|J|=a$, and $d_a := \binom{N}{a}$, the determinants obey
\be
\sum_{r=0}^{d_a}(-1)^r E_r^{(a)} H_{a,n+d_a-r}^{(O)} = S_{a,n}^{(O)}, \qquad a\in\{1,\cdots,N-1\}, \quad n\in{\mathbb N}_0. \label{eq:HOn_linear_rec}
\ee
Here, the source term $S_{a,n}^{(O)}$ encodes the observable-induced deviation from the free closure and identically vanishes for the trivial observable, $S_{a,n}^{({\mathbb I}_N)} = 0$.

In this sense, the present framework naturally produces a T-system-like bilinear structure accompanied by an observable-induced deformation. The standard $A_{N-1}$ closure is recovered only in the determinant-trivial and observable-trivial sector. For general $GL(N,{\mathbb C})$ and nontrivial $O$, both the top-boundary deviation $\xi_n^{(O)}$ and the source term $S_{a,n}^{(O)}$ characterize the observable-dependent dressing. In practical realization algorithms, the vanishing of these terms in the trivial sector ($O={\mathbb I}_N$) provides a rigorous consistency check for the reconstructed spectral skeleton.

\section{Realization-theoretic reconstruction from OTS} \label{sec:GLNstructure}

The ORS defined in Sec.~\ref{sec:CH_GLN_const} provides the theoretical hub that separates the skeleton from the dressing.
However, in practice, it is the Hankel or block-Hankel data matrix that plays the role of reconstructing (realizing) this ORS from OTS data.

While classical methods, such as the Prony method \cite{Prony1795}, the Matrix Pencil method \cite{HuaSarkar1990}, and realization-theoretic subspace methods including the Ho--Kalman algorithm, ERA, and ESPRIT \cite{HoKalman1966, JuangPappa1985, Paulraj1985}, as well as modern subspace state-space identification (4SID) methods \cite{VanOverscheeDeMoor1996}, reconstruct parameters from discrete sequences, our approach emphasizes the exact Floquet algebraic identities behind these finite-dimensional sequences.
In this section, we reconstruct the $GL(N, \mathbb{C})$ data from the OTS by systematically separating the universal spectral skeleton ($\{e_a\}_a, \{\lambda_j\}_j$) from the observable-dependent dressing ($\{c_{j,m}^{(O)}\}_{j,m}$). 

Our basic requirement at this stage is only to be given the OTS as a discrete numerical time-series. 
To reconstruct the similarity-invariant spectral skeleton, we do not require a priori knowledge of (i) the explicit matrix representation of the observable $O$ in a specific physical basis, or (ii) the exact dimension $N$ of the underlying $GL(N,{\mathbb C})$ structure, as the latter can be determined data-drivenly from the rank of the Hankel matrix. 
Furthermore, from a purely algebraic perspective, the exact characteristic polynomial is securely fingerprinted from a minimal stroboscopic window (minimally $2N+1$ points), without requiring an asymptotically long time-series.
However, we assume throughout this section that the observable $O$ is independent of the discrete period index $n$ in the OTS, $\{ \zeta^{(O)}_{n} \}_n$. 
This assumption guarantees the shift-invariance of the Hankel matrix, which is the necessary condition for the realization algorithms described below.

In Sec.~\ref{sec:MGF_HM}, we utilize the CH recurrence to identify the effective rank and skeleton from a single OTS.
In Sec.~\ref{sec:multi_obs}, we extend this via Hankel matrices to recover an equivalent realized monodromy matrix $M_{\rm rel}$.
In Sec.~\ref{sec:Liouville-space}, we discuss mapping to Liouville space for practical treatments for fundamental-type and adjoint-type OTSs.

\subsection{Cayley--Hamilton recurrence and Hankel matrix for the OTS} \label{sec:MGF_HM}
For identification of the $GL(N,{\mathbb C})$ structure, the CH recurrence for the OTS through the characteristic coefficients $\{e_a\}_{a=1}^N$ in Eq.\eqref{eq:CH_bn_recurrence} is useful.
(Recall that $e_0=1$.)
Since the coefficients are determined by the spectrum $\{\lambda_j\}_j$, our problem reduces to finding the appropriate $N$ and $\{e_a\}_{a=1}^N$ from the OTS.
Once the coefficients $\{e_a\}_{a=1}^N$ are determined, the spectrum $\{\lambda_j\}_j$ is obtained from the characteristic polynomial.
Subsequently, the observable-dependent coefficients, either $\{c_j^{(O)}\}_j$
for diagonalizable $M$ or $\{c_{j,m}^{(O)}\}_{j,m}$ for non-diagonalizable $M$, 
can be  recovered accordingly by solving associated linear system.
Below, we assume that all $\{ c_j^{(O)} \}_m$ or $\{c_{j,m}^{(O)}\}_{j,m}$ are non-zero in Eqs.\eqref{eq:ORS_lam}-\eqref{eq:ORS_lam_dps}.

For the SS cases, the visible recurrence order is determined from Eq.\eqref{eq:CH_bn_recurrence} by finding $N\in{\mathbb N}$ such that
\be
N = \min \{ N^\prime \in {\mathbb N} \, | \, {\bf H}^{(O)}_{{\bf n},N^\prime} \cdot {\bf E}_{N^\prime} = {\bf 0} \}
= \max_{N^\prime \in {\mathbb N}} {\rm rank} ({\bf H}^{(O)}_{{\bf n},N^\prime}),
\label{eq:cond_GLN}
\ee
where ${\bf H}^{(O)}_{{\bf n},N}$ and ${\bf E}_N$ are the Hankel matrix and the basis consisting of the characteristic coefficients defined as\footnote{
While we use a square $(N+1) \times (N+1)$ matrix to define the rank condition in Eq.\eqref{eq:cond_GLN} for simplicity, in practical numerical implementations one may construct a larger and/or generally rectangular Hankel matrix $M_{\rm row} \times M_{\rm col}$ to robustly determine the rank via singular value decomposition.
}
\be
{\bf H}^{(O)}_{{\bf n},N} :=
\begin{pmatrix}
\zeta^{(O)}_{n_0+N} & \zeta^{(O)}_{n_0+N-1} &\cdots&  \zeta^{(O)}_{n_0} \\
\zeta^{(O)}_{n_1+N} & \zeta^{(O)}_{n_1+N-1} &\cdots&  \zeta^{(O)}_{n_1} \\
\vdots & \vdots &\ddots& \vdots \\
\zeta^{(O)}_{n_N+N} & \zeta^{(O)}_{n_N+N-1} &\cdots&  \zeta^{(O)}_{n_N}
\end{pmatrix}
\in {\mathbb C}^{(N+1)\times (N+1)},
\qquad
{\bf E}_N :=
\begin{pmatrix}
1 \\
-e_1 \\
\vdots \\
(-1)^N e_N
\end{pmatrix}
\in {\mathbb C}^{N+1},
\label{eq:def_HNEN}
\ee
where ${\bf n} = (n_0,\cdots,n_N) \in {\mathbb N}_0^{N+1}$, and $n_{j_1}\neq n_{j_2}\in{\mathbb N}_0$ for all $j_1\neq j_2$.
This implies that, for identification of the $GL(N,{\mathbb C})$ structure, the minimal number of data points in the OTS is $(2N+1)$ by taking $n_{j+1}=n_j+1$, including $\zeta^{(O)}_{0}=\Tr \, O$.
If one further assumes that $M\in SL(N,{\mathbb C})$, then this minimum is reduced to $(2N-1)$.
This minimum requirement highlights the efficiency of the structural tomography, namely that the underlying finite-dimensional skeleton can already be fingerprinted from a short observable window.
In the EP cases, although the monodromy matrix becomes non-diagonalizable, the condition \eqref{eq:cond_GLN} still works.
This is because the degeneracy is reflected in the polynomially dressed behavior $n^{r-1}\lambda_j^n$ appearing in Eq.\eqref{eq:OTS_lam_eps}.
In the DP cases, however, the rank drops to the number $K$ of distinct eigenvalues because the effect of degeneracy can be absorbed into the effective coefficient $\widetilde{c}_{j}^{(O)}=\sum_{k\in{\cal S}_j} c_k^{(O)}$.
In this case, the OTS cannot distinguish a genuine $GL(K,{\mathbb C})$ system with simple spectrum from a $GL(N,{\mathbb C})$ system at a DP.
This situation stems from an intrinsic limitation of the observable dynamics, and one way to overcome this problem is to resolve the degeneracy by adjusting experimental parameters.
Notice that, strictly speaking, the recurrence extracted from a given OTS identifies the minimal effective annihilating polynomial visible in that channel, and it coincides with the full characteristic polynomial only under the corresponding nondegeneracy and visibility conditions.

Once the effective rank and the characteristic coefficients are determined, the distinct eigenvalues are obtained from the characteristic polynomial.
Subsequently, the observable-dependent coefficients can be recovered by the Prony method; for the DP case, this yields the effective coefficients, $\{\widetilde{c}^{(O)}_{j}\}_{j=1}^K$.
See App.~\ref{app:prony} for the Prony reconstruction.

The CH recurrence also admits a generalization to
an OTS measured by arbitrary $(N+1)$-points.
More precisely, instead of using $\{\zeta^{(O)}_{n+N},\zeta^{(O)}_{n+N-1},\cdots,\zeta^{(O)}_n\}$, one can consider the OTS $\{\zeta^{(O)}_{n^{(0)}},\zeta^{(O)}_{n^{(1)}},\cdots,\zeta^{(O)}_{n^{(N)}}\}$ with $n^{(0)}>n^{(1)}>\cdots>n^{(N)}\in{\mathbb N}_0$ with fixed $(n^{(0)} - n^{(r)}) \in {\mathbb N}$ for $r \in \{1,\cdots,N\}$.
In this case, the same finite-dimensional spectral skeleton constrains these data through Schur polynomial coefficients $S^{(r)}_{\mu(\widetilde{\bf n})}(\{e_a\}_a)$ as
\be
\zeta^{(O)}_{n^{(0)}} + \sum_{r=1}^{N} S^{(r)}_{\mu(\widetilde{\bf n})}(\{e_a\}_a) \zeta^{(O)}_{n^{(r)}} = 0, \qquad \widetilde{\bf n}=(n^{(0)},\cdots,n^{(N)}). \label{eq:GLN_Schur_interp}
\ee
This gives an arbitrary-stepsize extension of the finite-dimensional reconstruction.
Its structural significance is that the reconstruction is not constrained to consecutive sequence, such as $\zeta^{(O)}_n, \zeta^{(O)}_{n+1}, \cdots$, in the OTS, i.e.,
once the characteristic polynomial is fixed, the same common spectral skeleton constrains arbitrary $(N+1)$-point patterns as well.
In this sense, the arbitrary-stepsize extension is not merely a numerical convenience, but a direct algebraic consequence of the finite-dimensional monodromy structure.
The detailed formulas and derivation are given in App.~\ref{app:arb_stepsize}.
Another subtlety arises for decimated or fixed-step subsequences.
For a $k$-step sampling pattern, the effective poles are $\{\lambda_j^k\}_j$ rather than $\{\lambda_j\}_j$ themselves.
Hence, even if $\lambda_{j_1} \neq\lambda_{j_2}$, the condition $\lambda_{j_1}^k=\lambda_{j_2}^k$ leads to aliasing and reduces the rank visible in the sampled subsequence.
This is distinct from observable-dependent rank reduction and reflects a sampling-induced loss of distinguishability.

Finally, we comment on the case of null observable coefficients, $c^{(O)}_j = 0$.
If $c_j^{(O)}=0$ for some $j\in\{1,\cdots,N\}$, then the rank of the Hankel matrix is further reduced regardless of diagonalizability as
\be
\max_{N^\prime\in{\mathbb N}} \mathrm{rank}({\bf H}^{(O)}_{{\bf n},N^\prime}) =
\begin{cases}
  N-(\text{ zeros of }\{c_j^{(O)}\}_j/\{c_{j,m}^{(O)}\}_{j,m}) & \text{for SS and EPs} \\
K-(\text{ zeros of }\{\widetilde{c}_{j}^{(O)}\}_j) & \text{for DPs}
\end{cases}.
\ee
Therefore, the effective reconstructibility depends not only on the underlying spectrum, but also on the observable channel itself.

\subsection{Multi-observable tomography and realized monodromy matrix} \label{sec:multi_obs}
To overcome the informational limitations of a single operational channel and fully reconstruct the coordinate-dependent matrix structure of the monodromy matrix, we extend the single-observable framework to a multi-observable setting $\{O_\ell\}_{\ell=1}^{L}$.
This multi-channel construction is intimately connected to classical realization theory in control engineering, specifically the Ho--Kalman algorithm~\cite{HoKalman1966} and the multi-input-multi-output eigensystem realization algorithm (MIMO-ERA)~\cite{JuangPappa1985}. 
By assembling a block-Hankel data matrix from the intertwined trace sequences, these algorithms allow us to recover not only the universal similarity-invariant spectral skeleton but also an equivalent matrix realization $M_{\rm rel}$ of the underlying monodromy matrix.
In what follows, we formulate this realization-theoretic extension for both the SS and EP cases in detail.

We define the vector-valued OTS measured by $\{O_\ell\}_{\ell=1}^{L}$ as
\be
\bm{\zeta}_n^{({\bf O})}
:=
\begin{pmatrix}
\zeta_n^{(O_1)} \\
\zeta_n^{(O_2)} \\
\vdots \\
\zeta_n^{(O_L)}
\end{pmatrix}
\in{\mathbb C}^{L}, \qquad {\bf O}=\{O_1,\cdots,O_L\}, \qquad (L \in {\mathbb N})
\label{eq:def_zetaOn}
\ee
and generalize the Hankel matrix as
\be
\widetilde{\bf H}^{({\bf O})}_{{\bf n},N} := (\bm{\zeta}^{({\bf O})}_{{\bf n}+N},\bm{\zeta}^{({\bf O})}_{{\bf n}+N-1},\cdots,\bm{\zeta}^{({\bf O})}_{{\bf n}}) \in{\mathbb C}^{L(N+1)\times(N+1)}, \qquad \bm{\zeta}^{({\bf O})}_{{\bf n}+k} :=
\begin{pmatrix}
\bm{\zeta}^{({\bf O})}_{n_0+k} \\
\bm{\zeta}^{({\bf O})}_{n_1+k} \\
\vdots \\
\bm{\zeta}^{({\bf O})}_{n_N+k}
\end{pmatrix} \in {\mathbb C}^{L(N+1)}. \label{eq:def_HNEN_multiO}
\ee
From this generalized Hankel matrix, one can construct an equivalent realized monodromy matrix through the singular value decomposition (SVD) as
\be
\widetilde{\bf H}^{({\bf O})}_{{\bf n},N} = U\Sigma V^\dagger, \label{eq:H_USV}
\ee
where $U (V)$ are the left (right) singular vectors, and $\Sigma$ is, in general, a rectangular diagonal matrix of singular values.
By extracting the leading $N$ singular directions, one obtains a shift matrix ${\cal M}\in{\mathbb C}^{N\times N}$ and hence the realized monodromy matrix as
\be
M_{\rm rel} := \Sigma_N^{-1/2}{\cal M}\Sigma_N^{1/2} = \Sigma_N^{-1/2}U_N^\dagger \widetilde{\bf H}^{({\bf O})}_{{\bf n}+1,N} V_N\Sigma_N^{-1/2}.
\label{eq:def_Meff}
\ee
Under the full-rank assumption, this matrix is equivalent to the original monodromy matrix $M$ up to similarity transformation generated by a $GL(N,{\mathbb C})$ matrix $S$ as
\be
M_{\rm rel} = S M S^{-1}, \qquad S \in GL(N,{\mathbb C}). \label{eq:Mrel_SMSinv}
\ee
The detailed derivation is given in App.~\ref{app:Mrel_proof}.

As shown in App.~\ref{app:Mrel_proof}, there exists a similarity transformation $S$ that connects the realized matrix $M_{\rm rel}$ to the exact monodromy matrix $M$.
Although $S$ itself cannot be numerically determined solely from the OTSs (as it requires prior knowledge of the absolute coordinates of the system), its mathematical existence guarantees that $M_{\rm rel}$ inherits the similarity-invariant spectral skeleton of $M$, such as the spectrum and Jordan structure, under the full-rank assumption.
The remaining question is a condition for ${\bf O}$ to reconstruct the maximal skeleton of $M_{\rm rel}$, i.e., what the suitable set is. We will discuss this issue in Sec.~\ref{sec:rel_of_M}.

The realization-theoretic reconstruction is not merely a supplementary method, but rather the foundation that completes the ORS/OSD/ODSD framework in Sec.~\ref{sec:CH_GLN_const}.
By extracting a full-rank realized monodromy matrix $M_{\rm rel}$ in the multi-observable setting, one obtains a similarity-invariant representative of the common spectral skeleton.
Spectral quantities depending only on the similarity class of $M$ can then be evaluated directly from $M_{\rm rel}$, without relying on a particular single observable channel.
Consequently, applying the OSD and ODSD to this similarity-invariant representative allows one to analyze the common spectral skeleton without relying on a particular single observable channel.

In the DP cases, the rank of the Hankel matrix drops to $K<N$, and the realization procedure directly yields only a $K$-dimensional representation corresponding to the visible operator subspace, but the realized spectrum $\{\lambda_j\}_{j=1}^K$ is identical to the exact one up to the DP degeneracy.
To determine the invisible remainder associated with the unresolved degeneracy, extra physical information is necessary to constrain the internal components, such as tridiagonal or sparse structures in $M$.
Alternatively, when the boundary state can be varied as an additional input channel, this rank reduction problem can be systematically resolved by expanding the framework to a multiple-input multiple-output (MIMO) setting, as we will discuss in Sec.~\ref{sec:Liouville-space}.

\subsection{Liouville-space mapping} \label{sec:Liouville-space}
To connect our algebraic framework with laboratory quantum platforms, we map the reconstruction onto the Liouville space by explicitly introducing an initial state or boundary operator $\Omega \in {\mathbb C}^{N \times N}$.
This Liouville-space formulation allows us to treat two operationally different time-series data in a common linear-algebraic form.
The first is the fundamental-type OTS, $\zeta_n^{(O;\Omega)} = \Tr (O M^n \Omega)$, which captures the fundamental $GL(N,{\mathbb C})$ skeleton.
The second, and more physically prevalent, is the adjoint-type OTS, $\zeta_{n}^{{\rm adj}(O;\Omega)} = \Tr(O M^n \Omega M^{-n})$, representing stroboscopic expectation values~\cite{Breuer2000,Kohler2005}.
This adjoint sequence constitutes the native data format in standard quantum process tomography (QPT)~\cite{PoyatosCiracZoller1997,ChuangNielsen1997}, gate set tomography (GST)~\cite{MerkelGambettaSmolinPolettoCorcolesJohnsonRyanSteffen2013,BlumeKohoutGambleNielsenMizrahiSterkMaunz2013}, Floquet spectroscopy on superconducting processors~\cite{Roushan_2017}, and randomized measurement methods~\cite{Elben_2023}.
Lifting this sequence into the Liouville space translates the physical expectation values back into a linear transition amplitude governed by the adjoint superoperator $M_{\rm adj} := M \otimes M^{-\top}$, enabling exact skeleton extraction for the ratio spectrum.

Crucially, the structural introduction of multiple boundary operators $\Omega_p$ resolves the rank-deficiency problem encountered under eigenvalue degeneracies (degenerate poles, DPs).
While varying only the readout observables $O_\ell$ (as in Sec.~\ref{sec:multi_obs}) shifts the projection on the left-eigenvector sector, it fails to lift algebraic ambiguities within a degenerate subspace. 
Conversely, preparing independent initial boundaries $\Omega_p$ actively injects weights across the right-eigenvector sector. 
This dual scanning of both the left and right sectors can satisfy the full-rank condition of the corresponding block-Hankel matrix, fully probing the internal degenerate structure within the Liouville space.

\subsubsection{Fundamental-type OTS: $\zeta^{(O;\Omega)}_n = \Tr (O M^n \Omega)$} \label{sec:ampli_type_data}
We first define the Liouville space.
The Liouville-(dual) space $| A \rrangle$ ($\llangle A |$) from the bi-orthogonal basis in Eq.\eqref{eq:lr_states} is defined as
\be
| A \rrangle := \sum_{j,k=1}^N A_{jk} | u_j \rangle \otimes | \widetilde{u}_k \rangle^{*}, \qquad \llangle A | := \sum_{j,k=1}^N A_{kj} \langle \widetilde{u}_j | \otimes \langle u_k |^{*}, \qquad A_{jk} \in {\mathbb C}, \label{eq:A_Liouville}
\ee
with the complex conjugate $*$, where the inner-product is given by\footnote{
Notice that Eq.\eqref{eq:A_Liouville} is a bilinear pairing adapted to the biorthogonal convention, not the Hermitian Hilbert--Schmidt inner product.
}
\be
\llangle A | B \rrangle = \Tr(A B).
\ee
The action of an operator $D = B \otimes C$ on $| A \rrangle$ is given by
\be
D | A \rrangle = | B A C^{\top}  \rrangle, \qquad B, C \in {\mathbb C}^{N \times N}.
\ee
By this convention, the fundamental-type OTS with a boundary operator $\Omega$ can be expressed by\footnote{
This identity can be found by
\be
\llangle O | M^n \otimes {\mathbb I}_N | \Omega \rrangle &=&  \sum_{j,k,\ell,m=1}^{N} O_{kj} \Omega_{\ell m} (\langle  \widetilde{u}_j | \otimes  \langle  u_k |^* )   (M^n \otimes {\mathbb I}_N)  (| u_\ell \rangle \otimes | \widetilde{u}_m \rangle^{*} ) \nl
&=&   \sum_{j,k,\ell = 1}^{N} O_{kj} \Omega_{\ell k} \langle  \widetilde{u}_j | M^n| u_\ell \rangle \, =  \, \sum_{j,k,\ell = 1}^{N} \langle \widetilde{u}_k | O | u_j \rangle 
\langle  \widetilde{u}_j |  M^n | u_\ell \rangle \langle \widetilde{u}_\ell | \Omega | u_k \rangle \nl
&=& \Tr (O M^n \Omega).
\ee
}
\be
\zeta_n^{(O;\Omega)} := \Tr (O M^n \Omega) = \llangle O | M^n \otimes {\mathbb I}_N | \Omega \rrangle, \qquad (n \in \mathbb N_0) \label{eq:amp_OTS_def}
\ee
and the ordinary OTS in Eq.\eqref{eq:OTS_def} is reproduced by taking $\Omega = {\mathbb I}_N$.

In order to see how the MIMO setting resolves the rank reduction problem of DPs, let us formulate a matrix-valued OTS using $L$ observables ${\bf O} = \{O_1, \cdots, O_L\}$ and $R$ boundary operators $\bm{\Omega} = \{\Omega_1, \cdots, \Omega_R\}$.
The sequence constructs an $L \times R$ matrix as
\be
\bm{\zeta}_n^{({\bf O};\bm{\Omega})} :=
\begin{pmatrix}
  \zeta_{n}^{(O_1;\Omega_1)} & \cdots & \zeta_{n}^{(O_1;\Omega_R)} \\
 \vdots & \ddots & \vdots \\
\zeta_{n}^{(O_L;\Omega_1)} & \cdots & \zeta_{n}^{(O_L;\Omega_R)} \\
\end{pmatrix} \in {\mathbb C}^{L \times R}. \label{eq:zetan_O_Ome}
\ee
Let us suppose $M$ has a DP with an eigenvalue $\lambda_D$ of algebraic multiplicity $d$. In the single-channel setting ($L=R=1$), the contribution from this degenerate subspace is traced out into a single scalar coefficient, yielding only rank $1$ in the Hankel matrix. In contrast, the MIMO sequence factorizes the contribution into a matrix product as
\be
&& \bm{\zeta}_n^{({\bf O};\bm{\Omega})} = \cdots + \left( {\bf C}_D {\bf B}_D \right) \lambda_D^n + \cdots, \qquad {\bf C}_{D} \in {\mathbb C}^{L \times d^2}, \ \ {\bf B}_{D} \in {\mathbb C}^{d^2 \times R},
\ee
where ${\bf C}_D$ and ${\bf B}_D$ are defined by vectorizing the internal degrees of freedom using a multi-index $\alpha = (a,b)$ for the degenerate basis indices $a,b \in \{1, \cdots, d \}$ as
\be
&& [{\bf C}_D]_{\ell, \alpha} := \llangle O_\ell | \alpha \rrangle = \langle \widetilde{u}_a | O_\ell | u_b \rangle, \qquad [{\bf B}_D]_{\alpha, p} := \llangle \alpha | \Omega_p \rrangle = \langle \widetilde{u}_b | \Omega_p | u_a \rangle.
\ee
Here, $| \alpha \rrangle:= | u_b \rangle \otimes | \widetilde{u}_a \rangle^*$ and $\llangle \alpha | := \langle \widetilde{u}_b | \otimes \langle u_a |^*$ represent the Liouville-space basis states associated with the degenerate block.
As long as the chosen observables and states are sufficiently rich ($L, R \ge d^2$) and linearly independent within the subspace, the coefficient matrix ${\bf C}_D {\bf B}_D$ inherently possesses rank $d^2$ rather than $d$.

As we discussed in Sec.~\ref{sec:multi_obs} (and App.~\ref{app:Mrel_proof}), one can employ $\bm{\zeta}_n^{({\bf O};\bm{\Omega})}$ in Eq.\eqref{eq:zetan_O_Ome} instead of $\bm{\zeta}^{({\bf O})}$ in Eq.\eqref{eq:def_zetaOn} and construct a Hankel matrix to find a realized monodromy matrix.
The resulting realized matrix ${\cal M}_{\rm rel}$ is not an $N \times N$ matrix but an $N^2 \times N^2$ matrix, which is given by
\be
{\cal M}_{\rm rel} = S (M \otimes {\mathbb I}_N) S^{-1}, \qquad S \in GL(N^2,{\mathbb C}),
\ee
with the similarity transformation $S$.
Although the realized matrix ${\cal M}_{\rm rel}$ has $N$-fold degenerate eigenvalues, i.e.,
\be
\{ \underbrace{\lambda_1,\cdots,\lambda_1}_{\times N}, \underbrace{\lambda_2,\cdots,\lambda_2}_{\times N},
\cdots, \underbrace{\lambda_N,\cdots,\lambda_N}_{\times N} \},
\ee
the rank reduction caused by DPs is resolved in ${\cal M}_{\rm rel}$, unlike the ordinary trace-type OTS, $\zeta^{(O)}_n = \Tr(O M^n)$.

\subsubsection{Adjoint-type OTS: $\zeta^{{\rm adj}(O;\Omega)}_n = \Tr (O M^n \Omega M^{-n})$} \label{sec:adjoint_type_data}

In experimentally accessible Floquet systems, one often measures stroboscopic expectation values with a density operator $\Omega$, taking the form as
\be
\zeta_n^{{\rm adj}(O;\Omega)} &:=& \Tr \left( O\, M^n \Omega M^{-n} \right) = \llangle O | M^n_{\rm adj} | \Omega\rrangle, \qquad n \in \mathbb N_0, \label{eq:adjoint_OTS_def}
\ee
where $M_{\rm adj}$ denotes
\be
M_{\rm adj} := M \otimes M^{-\top}, 
\ee
with the inverse-transpose matrix of $M$,  $M^{-\top}$.
The adjoint OTS can be directly translated to the expectation value of $O$ under the periodic-time evolution $M = U(T,0)$ in non-Hermitian quantum systems:
\be
&& \langle O \rangle_n := \langle \widetilde{\psi}_n | O | \psi_n \rangle  = \Tr (O \rho_n), \qquad  \rho_n := | \psi_n \rangle \langle \widetilde{\psi}_n | = M^n \rho_0 M^{-n},
\ee
where $\rho_0 = | \psi_0 \rangle \langle \widetilde{\psi}_0 | =: \Omega$ is the initial density operator.
Physically, the adjoint-type sequence $\zeta_n^{{\rm adj}(O;\Omega)}$ naturally encapsulates the operational structure of quantum process tomography (QPT)~\cite{PoyatosCiracZoller1997,ChuangNielsen1997} and gate set tomography (GST)~\cite{MerkelGambettaSmolinPolettoCorcolesJohnsonRyanSteffen2013,BlumeKohoutGambleNielsenMizrahiSterkMaunz2013}.
In the context of Floquet theory, this stroboscopic map acts as a Floquet superoperator in Liouville space, whose spectrum is governed by the quasi-energy differences, i.e., the ratio spectrum $\Lambda_{jk} = \lambda_j / \lambda_k$~\cite{Breuer2000,Kohler2005}.
Therefore, QPT and GST settings inherently reconstruct the adjoint-type Floquet superoperator $M_{\rm adj}$ and its ratio spectrum, rather than the absolute spectrum $\lambda_j$.

Since it is generated not by powers of $M$ alone but by the adjoint action of $M$ on operators, this type of OTS differs from the ordinary OTS in Eq.\eqref{eq:OTS_def}.
Actually, Eq.\eqref{eq:adjoint_OTS_def} can be expanded depending on the SS, DP, or EP cases as
\be
\zeta_n^{{\rm adj}(O;\Omega)} =
\begin{cases}
  \sum_{j,k=1}^N C_{jk}^{(O;\Omega)}   \Lambda_{jk}^n & \quad \text{for SS}  \\
  \sum_{j,k=1}^K \widetilde{C}_{jk}^{(O;\Omega)} \Lambda_{jk}^n & \quad \text{for DPs}  \\
  \sum_{j,k=1}^K {C}_{jk,n}^{(O;\Omega)}(\lambda_{j},\lambda_k) \Lambda_{jk}^n & \quad \text{for EPs} 
\end{cases}, \label{eq:adjoint_OTS_spectral}
\ee
where $d_j$ is the algebraic multiplicity for the $j$-th distinct eigenvalue $\lambda_j$, and the ratio spectrum is defined as
\be
\Lambda_{jk}:= \frac{\lambda_j}{\lambda_k}, \label{eq:ratio_spectrum}
\ee
and the coefficients are given by
\be
&& C_{jk}^{(O;\Omega)} := \llangle O | P_{jk} | \Omega \rrangle, \qquad P_{jk}:= P_{j} \otimes P^{\top}_{k}, \\
&& \widetilde{C}_{jk}^{(O;\Omega)} := \llangle O | {\cal P}_{jk} | \Omega \rrangle, \qquad {\cal P}_{jk} :=  \sum_{a \in {\cal S}_{j}} \sum_{b \in {\cal S}_{k}} P_a \otimes P^{\top}_b, \\
&& {C}_{jk,n}^{(O;\Omega)}(\lambda_{j},\lambda_k) := \sum_{m=0}^{d_j + d_k -2}
\begin{pmatrix}
  n \\
  m
\end{pmatrix} \Lambda_{jk}^{-m}
\llangle O | N^m_{jk} | \Omega \rrangle, \nl
&& \qquad \qquad \qquad \qquad
N_{jk}:= \lambda_k^{-1} N_{j} \otimes P_k^{\top} + \sum_{s=1}^{d_k-1} (-\lambda_k)^{-s-1}  (\lambda_j P_j + N_j)  \otimes (N_k^s)^{\top}.
\ee

Correspondingly, one may define Liouville-space analogues of the OSD, ORS, and ODSD by
\be
h^{{\rm adj}(O;\Omega)}(z) &:=& \exp \left[ -
\llangle O \,|\, \log ( {\mathbb I}_{N^2}-z M_{\rm adj} ) \,|\,\Omega \rrangle \right] = \exp \left[ \sum_{n \in {\mathbb N}}\frac{\zeta_n^{{\rm adj}(O;\Omega)}}{n}z^n \right], \label{eq:adjoint_OSD} \\
{\cal Z}^{{\rm adj}(O;\Omega)}(z) &:=& \pd_z \log h^{{\rm adj}(O;\Omega)}(z) = \llangle O \,|\, M_{\rm adj} ({\mathbb I}_{N^2}-z M_{\rm adj})^{-1} \,|\,\Omega \rrangle = \sum_{n \in {\mathbb N}}\zeta_n^{{\rm adj}(O;\Omega)} z^{n-1}, \label{eq:adjoint_ORS} \\
{\cal F}^{{\rm adj}(O;\Omega)}(p) &:=& \frac{1}{\Gamma(p)} \int_0^{+\infty} d\rho\, e^{-\rho}\rho^{p-1} {\cal Z}^{{\rm adj}(O;\Omega)}(e^{-\rho}) \, = \, \sum_{n \in {\mathbb N}} \frac{\zeta_n^{{\rm adj}(O;\Omega)}}{n^p}, \label{eq:adjoint_ODSD}
\ee
where ${\mathbb I}_{N^2}={\mathbb I}_N\otimes{\mathbb I}_N$ is the identity operator on the Liouville space.
Here, the Liouville-space OSD should be understood as a generating object built from the superoperator $M_{\rm adj}$ and the pair $(O,\Omega)$, rather than as a trace over the original Hilbert space.
As in the ordinary framework, the adjoint ORS takes the form
\be
{\cal Z}^{{\rm adj}(O;\Omega)}(z) = \frac{{\cal N}_{\rm adj}^{(O;\Omega)}(z)} {\Delta_{\rm adj}(z)}, \qquad \Delta_{\rm adj}(z):=\det({\mathbb I}_{N^2} - z M_{\rm adj}),
\label{eq:adjoint_resolvent_fraction}
\ee
hence, the separation between a common skeleton and observable-dependent dressing persists in Liouville space as well.

The ratio spectrum \eqref{eq:ratio_spectrum} is invariant under the overall rescaling
\be
M \longmapsto g_0 \cdot M, \qquad g_0\in {\mathbb C}^\times.
\ee
Hence, the adjoint-type OTS is blind to the central $GL(N,\mathbb C)$ factor, or equivalently to the determinant sector of $M$.
For this reason, Eq.\eqref{eq:adjoint_OTS_def} is naturally suited to the reconstruction of the reduced or projective part of the monodromy matrix, while the full $GL(N,\mathbb C)$ data requires an additional determinant-sensitive channel such as the ordinary OTS in Eq.\eqref{eq:OTS_def}.
In other words, adjoint-type tomography is naturally compatible with the $SL(N,\mathbb C)$-reduced skeleton, whereas the recovery of the full $GL(N,\mathbb C)$ structure requires supplementary scalar information.

We do not take Eq.\eqref{eq:adjoint_OTS_def} as the primary definition of the OTS in the present paper.
Our main framework is built from the ordinary trace-type OTS $\zeta_n^{(O)}=\Tr(O M^n)$, which directly observes the spectrum of $M$ and retains the determinant sector.
Nevertheless, Eq.\eqref{eq:adjoint_OTS_def} should be viewed as an experimentally motivated companion problem, i.e.,
it leads to a Liouville-space tomography problem of essentially the same algebraic type, but for the superoperator $M_{\rm adj}=\mathrm{Ad}_M$ rather than for $M$ itself.

\section{Algebraic realization, micromotion, and symmetry} \label{sec:prel_micro_symm}

In this section, we clarify a fundamental structural issue left open by the realization-theoretic approach in Sec.~\ref{sec:multi_obs}: even when a finite-dimensional matrix realization can be operationalized from OTS data, what sector of the exact monodromy matrix is genuinely identifiable from a restricted set of probes, and what sector remains fundamentally invisible? 
As we demonstrate below, the boundary between algebraic visibility and invisibility is strictly dictated by the operator-algebraic structure generated by the accessible observables, particularly through the lens of their commutant and bicommutant~\cite{BratteliRobinson}.

To formalize this constraint, Sec.~\ref{sec:rel_of_M} introduces the concept of partial realization, explicitly decomposing the full monodromy matrix into a visible representative uniquely determined by the observable algebra and a trace-invisible remainder. 
Sec.~\ref{sec:micro_symm} then examines how continuous micromotion can enlarge this visible operator space, while showing that exact commuting symmetries can impose a non-removable algebraic deficiency within the corresponding symmetry sector.
All proofs for the propositions established in this section are presented in App.~\ref{app:proof_props}.
Finally, in Sec.~\ref{sec:reconstruct_Mex}, we show that micromotion is not only useful for symmetry analysis but also a practical asset for basis-resolved tomography, providing a systematic procedure to reconstruct the exact monodromy matrix under a calibrated micromotion setting for both the SS and EP cases.

\subsection{Partial realization of the monodromy matrix} \label{sec:rel_of_M}
In Sec.~\ref{sec:multi_obs}, we described the realization of an equivalent monodromy matrix assuming a sufficiently rich set of multi-observable OTSs. 
In many physically relevant settings, however, the accessible observables do not generate the full matrix algebra $\mathrm{End}(\mathcal{H})$. 
For instance, experimental probes are often localized within a specific spatial invariant sector or a restricted subsystem of the Hilbert space. 
In such cases, even ideal, noise-free trace data are insufficient to uniquely determine the full monodromy matrix. 
Instead, the inverse problem naturally splits into a visible object, uniquely fixed by the subalgebra $\mathfrak{O}$ generated by the observables, and an invisible remainder that remains fundamentally undetectable through any OTS generated within this algebra. 
This structural limitation motivates two central algebraic questions: what precise conditions on the observable set $\mathbf{O}=\{O_\ell\}_\ell$ guarantee the recovery of the full monodromy matrix, and what is the well-defined notion of reconstruction when $\mathfrak{O}$ is intrinsically restricted?

To formulate this algebra precisely, we define the unital $*$-algebra ${\frak O}$ generated by ${\bf O}$ as
\be
{\frak O} := \left\{ c_0{\mathbb I}_N + \sum_{m=1}^{m_{\rm max}} \sum_{\ell_1,\dots,\ell_m=1}^{L} c_{\ell_1\cdots\ell_m} O_{\ell_1}\cdots O_{\ell_m} \ \middle|\ m_{\rm max} \in{\mathbb N},\ c_0,c_{\ell_1\cdots\ell_m}\in{\mathbb C} \right\}. \label{eq:unital_alg}
\ee
We define its commutant and bicommutant by
\be
{\frak O}^\prime:=\{X \in {\rm End}({\cal H})\,|\,[X,\sigma]=0\ \text{for all }\sigma\in {\frak O}\},
\qquad
{\frak O}^{\prime\prime}:=\{Y\in {\rm End}({\cal H})\,|\,[Y,X]=0\ \text{for all }X\in {\frak O}^\prime\}.
\ee
Intuitively, these algebraic objects can be understood as follows: 
The observable algebra ${\frak O}$ represents the space of all possible operations and measurements accessible via the given set of observables ${\bf O}$. Conversely, it defines the strict boundary of what can be observed. 
The commutant ${\frak O}^\prime$ consists of all operators that commute with ${\frak O}$. Physically, these operators do not interfere with the observables in ${\bf O}$, and therefore constitute an invisible sector (or invisible sector) entirely inaccessible to the chosen measurement channels. 
The bicommutant ${\frak O}^{\prime \prime}$ then describes all operators that are isolated from this invisible sector. 
In finite dimension, ${\frak O}^{\prime\prime}$ identically coincides with ${\frak O}$ by von Neumann's double commutant theorem, representing the canonical maximal algebra visible from the set of observables (see, e.g., Refs.~\cite{BratteliRobinson,ZanardiLidarLloyd2004,ViolaKnillLaflamme2001} for the use of commutant structures in quantum observable algebras).

The following proposition makes the visible/invisible decomposition explicit.
\begin{proposition}[Algebraic partial realization of Floquet monodromy matrix] \label{prop:Alg_rel_M}
  Let $\mathcal H\cong \mathbb C^N$ be a finite-dimensional Hilbert space, $M\in GL(N,{\mathbb C})$ be an unknown Floquet monodromy matrix acting on $\mathcal H$, and ${\bf O}=\{O_1,\cdots,O_L\}$ be a set of known Hermitian observables.
  In addition, suppose ${\frak O}$ be the unital $*$-algebra generated by ${\bf O}$ and that ${\frak O}\subsetneq {\rm End}({\cal H})$ is a proper subalgebra.
  Given the family of OTSs restricted to ${\frak O}$,
\be
\zeta_n^{(\sigma)}:=\Tr(\sigma M^n),\qquad \sigma\in {\frak O},\ \ n\in {\mathbb N}_0,
\ee
the following statements hold:
\begin{enumerate}[label=(P1-\arabic*)]
\item \label{prop:p1-1}
  In general, the full monodromy matrix $M$ is not identifiable from the family of restricted OTSs, $\{\zeta_n^{(\sigma)}\}_{\sigma\in {\frak O},\,n\ge 0}$.

\item \label{prop:p1-2}
There exists a canonical visible representative
\be
M_{\rm vis} = {\mathbb E}_{{\frak O}^{\prime\prime}}(M)\in {\frak O}^{\prime\prime}, \label{eq:Mvis_p1-2}
\ee
where $\mathbb E_{{\frak O}^{\prime\prime}}$ denotes the finite-dimensional conditional expectation onto ${\frak O}^{\prime\prime}$, which is equivalently the orthogonal projection onto ${\frak O}^{\prime\prime}$ with respect to the Frobenius inner product $\langle A, B \rangle_F := \Tr (A^\dagger B)$. This representative is uniquely determined by the family of restricted OTSs.

\item \label{prop:p1-3}
The invisible remainder
\be
\Delta M:=M-M_{\rm vis}
\ee
is Frobenius-orthogonal to ${\frak O}^{\prime\prime}$, and hence, in particular
\be
\Tr(\sigma\,\Delta M)=0, \qquad \forall\sigma\in {\frak O}.
\ee
Thus, $\Delta M$ is trace-invisible relative to the observable algebra. More generally, the same visible/invisible decomposition applies to each power $M^n$ separately.
\end{enumerate}
\end{proposition}

The role of Sec.~\ref{sec:MGF_HM} and Sec.~\ref{sec:multi_obs} was to describe how a full-rank realization can be built from finite OTS samples when the observable information is sufficiently large.
Prop.~\ref{prop:Alg_rel_M} clarifies the complementary algebraic limitation. 
Statement \ref{prop:p1-1} means that if the observable algebra itself is too small, even an ideal and infinite family of restricted OTSs cannot determine the full monodromy matrix.
Statement \ref{prop:p1-2} shows that what remains accessible is strictly the canonical visible representative inside ${\frak O}^{\prime\prime}$. 
Consequently, statement \ref{prop:p1-3} implies that the difference between the true dynamics and this representative leaves a trace-invisible remainder, completely hiding a sector of the dynamics from the chosen observable channels.

We emphasize that the canonical visible representative $M_{\rm vis}$ in Eq.\eqref{eq:Mvis_p1-2} is conceptually distinct from the full-rank similarity realization $M_{\rm rel}$ discussed in Sec.~\ref{sec:multi_obs}. 
As we described, the latter (Hankel realization) relies on the dynamical time-series to preserve the exact spectrum of $M$ up to similarity.
In contrast, the former (algebraic projection) intrinsically traces out the invisible sectors statically, which alters the spectrum unless the observable family generates the full matrix algebra. 
These two complementary perspectives converge, meaning the exact spectral realization identically reconstructs the full algebraic state space without any invisible remainder, when the family of observables is large enough to generate the full matrix algebra (${\frak O} = {\rm End}({\cal H})$). In this maximal limit, $M_{\rm vis}=M$ and the full monodromy matrix becomes completely identifiable.

This algebraic framework naturally extends to the inclusion of varying initial states $\Omega$ in Sec.~\ref{sec:Liouville-space}, which play a dual role to the observables $O$ by defining a symmetric projection structure on the input side.

\subsection{Micromotion and symmetry} \label{sec:micro_symm}
A major technical challenge in the realization problem is that uniquely fixing all $N^2$ components of the similarity transformation $S$ in Eq.\eqref{eq:Mrel_SMSinv} poses an extremely stringent requirement on the static observable set. 
To bypass this limitation, one can exploit continuous micromotion $U(t,0)$ as a dynamical gauge transformation, Eq.\eqref{eq:def_M(t)}, which generates a time-shifted family of observables and effectively expands the reconstructible operator space.
However, if the system possesses an exact commuting symmetry compatible with the initial observable algebra, the micromotion-generated family remains restricted by a non-removable algebraic deficiency within that symmetry sector.
In this subsection, we examine how micromotion systematically enlarges the visible operator algebra, while showing that exact commuting symmetries can leave a residual invisible subspace. This residual subspace is measured by a positive algebraic deficiency $\Delta_{\rm obs}>0$ as long as the exact commuting symmetry and the initial observable algebra are kept fixed.

Let us return to the gauge transformation in Eq.\eqref{eq:def_M(t)}. 
Since the OTS is not gauge-invariant for a general observable, the time-shifted observables $O(t)$ may themselves be used as an enlarged observable family. 
In the Liouville-space notation introduced in Sec.~\ref{sec:Liouville-space}, the time-shifted OTS is expressed as
\be
\zeta^{(O)}_{n,t} = \llangle O(t) | M^n \otimes {\mathbb I}_N| {\mathbb I}_N \rrangle = \llangle O(t)  M^n | {\mathbb I}_N \rrangle, \label{eq:zetaOnt_Liouville}
\ee
where we assumed $\Omega = {\mathbb I}_N$.
If $O$ is an algebraically generative observable that does not commute with the Floquet dynamics, the time-shifted operator $O(t)$ traverses the $N^2$-dimensional dual operator space through the gauge transformation $U(t,0)$. 
This dynamical trajectory allows a single observable channel to effectively scan the $N^2$ dimensions of the Liouville space. 
In consequence, the micromotion provides sufficient algebraic constraints to lock the similarity transformation $S$ relative to the visible operator space, thereby enabling the reconstruction of the full algebraic skeleton. 
In the context of Floquet theory, it is well recognized that micromotion is not merely a gauge artifact but carries essential physical information, such as anomalous topological signatures that cannot be captured by the stroboscopic monodromy matrix alone~\cite{RudnerLindnerBergLevin2013, Harper2020}. 
In our algebraic framework, micromotion plays an analogous role for observability: shifting the Floquet reference time dynamically enlarges the accessible observable algebra ${\frak O}_{\rm ext}$, thereby increasing the observable dimension $D_{\rm obs}$ toward the full $N^2$ limit, up to the bound imposed by the system's exact symmetries and spectral degeneracies.

However, micromotion does not remove all obstruction.
From the perspective of quantum control and algebraic system identification, it is a fundamental fact that exact symmetries commuting with both the dynamics and the initial observables strictly bound the generated observability algebra, imposing symmetry-protected invisible sectors~\cite{AlbertiniDAlessandro2003,BurgarthBoseBruderGiovannetti2009}.
The following proposition formalizes this interplay, i.e., micromotion can enlarge the visible sector up to its algebraic limit, but exact commuting symmetries impose a strict, non-removable algebraic deficiency even after the full time-shift:
\begin{proposition}[Symmetry indicator via algebraic dimension deficiency] \label{prop:symmetry_indicator}
Let ${\frak O}_{\rm ext}$ be the extended unital $*$-algebra generated by the time-shifted observables:
\be
{\frak O}_{\rm ext} := \left\{ c_0{\mathbb I}_N + \sum_{m=1}^{m_{\rm max}} \sum_{a_1,\dots,a_m} c_{a_1\cdots a_m} X_{a_1}\cdots X_{a_m} \ \middle|\ m_{\rm max}\in{\mathbb N},\ c_0,c_{a_1\cdots a_m}\in{\mathbb C},\ X_{a_r}\in{\cal G}_{\rm ext} \right\}, \label{eq:Oext_alg}  
\ee
where
\be
{\cal G}_{\rm ext} := \left\{ O(t),\,O(t)^\dagger \ \middle|\ O\in{\bf O},\ t\in[0,T) \right\}, \qquad O(t):=U(t,0)^{-1}OU(t,0). \label{eq:Oext_generators}
\ee
Let $D_{\rm obs}$ denote the number of linearly independent matrix elements of $M$, which we call observable dimension, reconstructible from the family of extended OTSs.
Then the following statements hold:
\begin{enumerate}[label=(P2-\arabic*)]
\item \label{prop:p2-1}
  The observable dimension $D_{\rm obs}$ is bounded by the dimension of the extended bicommutant ${\frak O}_{\rm ext}^{\prime\prime}$:
\be
D_{\rm obs}\le \dim({\frak O}_{\rm ext}^{\prime\prime})\le N^2.
\ee

\item \label{prop:p2-2}
Suppose there exists a nontrivial symmetry operator ${\cal Q}$ ($ {\cal Q}\neq c {\mathbb I}_N$) such that
\be
[\mathcal Q,U(t,0)]=0 \ \ \text{for all} \ \  t \in [0,T),
\qquad
[\mathcal Q,O]=0 \ \ \text{for all} \ \ O \in {\bf O}.
\ee
Then, ${\cal Q}$ survives the micromotion extension and belongs to the extended commutant:
\be
{\cal Q}\in {\frak O}_{\rm ext}^{\prime}.
\ee

\item \label{prop:p2-3}
Consequently, such an exact symmetry forces the bicommutant ${\frak O}_{\rm ext}^{\prime\prime}$ to be a proper subspace of ${\rm End}({\cal H})$, so that the invisible dimension $\Delta_{\rm obs}$ is bounded as
\be
\Delta_{\rm obs} = N^2-\dim({\frak O}_{\rm ext}^{\prime\prime})>0.
\ee
Thus, an exact commuting symmetry enforces a nonzero algebraic deficiency even after micromotion extension.
\end{enumerate}
\end{proposition}

Physically, statement \ref{prop:p2-1} means that although micromotion dynamically enlarges the observable algebra, the resulting visible operator space remains strictly bounded by the algebraic structure of the initial observables.
Statement \ref{prop:p2-2} shows that if an exact symmetry commutes with both the system's dynamics and the initial observables, it also commutes with the time-shifted observables, remaining completely invisible to the time shift.
Consequently, Prop.~\ref{prop:symmetry_indicator} implies that such a commuting symmetry leaves a fixed deficiency in the reconstructible operator space, isolating a symmetry-protected sector that remains strictly unobservable regardless of the amount of time-series data collected.

Notice that the logical direction of Prop.~\ref{prop:symmetry_indicator} is one-way.
An exact symmetry guarantees a nonzero algebraic deficiency after micromotion extension.
The converse is not true in general.
A deficiency may also arise from an insufficiently rich set of observables, accidental degeneracies, or other visibility losses unrelated to a genuine symmetry.
For this reason, the algebraic deficiency should be interpreted as a symmetry-protected invisible sector when an exact commuting symmetry is present, rather than as a necessary and sufficient condition for symmetry.

As we will discuss more precisely in Sec.~\ref{sec:reconstruct_Mex}, the observable dimension $D_{\rm obs}$ can be defined by the dimension of the subspace spanned by the extended observable algebra ${\frak O}_{\rm ext}$. In the Liouville space notation introduced in Sec.~\ref{sec:Liouville-space}, this is given by
\be
D_{\rm obs} := {\rm rank} ( \Phi ), \qquad \Phi := (|O_1 \rrangle, |O_2 \rrangle, \dots ), \qquad O_{a} \in {\bf O}_{\rm ext}. \label{eq:Dobs_def}
\ee
This quantity represents the number of independent matrix elements of $M$ that can be uniquely determined from the given set of observables, and it can be evaluated numerically through the rank $\Phi$.

\subsection{Reconstruction of the exact monodromy matrix} \label{sec:reconstruct_Mex}
We now address a stronger class of inverse problems than the similarity-class reconstruction developed in Sec.~\ref{sec:multi_obs}. 
Instead of recovering $M$ up to an arbitrary basis choice, our goal here is to determine the exact matrix elements of the monodromy matrix in a designated, fixed physical basis.
This basis-resolved reconstruction becomes possible when the family of time-shifted observables $\{O_\ell(t_q)\}$ is sufficiently rich to span $\End({\cal H})$ and is explicitly known in the target basis.
The latter condition can be met by a sufficiently rich static observable set $\{O_\ell\}$ alone, or more typically when only a few static probes are accessible, by a continuous micromotion $U(t,0)$ that is well-characterized or independently calibrated.
In the following, we provide the step-by-step reconstruction procedures for both the SS and EP cases.

In order to reconstruct the monodromy $M$ in practical computations, one can incorporate the micromotion by discretizing the continuous time parameter $t \in [0, T)$ into $Q$ slices.
  The Hankel matrix analysis introduced in Eqs.\eqref{eq:def_zetaOn}\eqref{eq:def_HNEN_multiO} is then systematically expanded by upgrading the static observable set ${\bf O}$ to the time-shifted extended set:
\be
&& {\bf O}_{\rm ext} = \{O_\ell(t_q) \}_{\ell=1,q=0}^{L,Q-1}, \qquad O_\ell(t_q)=U(t_q,0)^{-1} O_\ell U(t_q,0), \nl
&&   t_q := q \Delta t, \qquad  \Delta t = \frac{T}{Q} \ \ \text{with} \ \ Q \in {\mathbb N}. \label{eq:def_obs_dt_ext}
\ee
Correspondingly, the OTS vector and the composite Hankel matrix are enlarged from Eqs.\eqref{eq:def_zetaOn}\eqref{eq:def_HNEN_multiO} as
\be
\bm{\zeta}_n^{({\bf O})} \  \rightarrow  \ \bm{\zeta}_n^{({\bf O}_{\rm ext})} \in{\mathbb C}^{L Q}, \qquad \widetilde{\bf H}^{({\bf O})}_{{\bf n},N}  \ \rightarrow  \ \widetilde{\bf H}^{({\bf O}_{\rm ext})}_{{\bf n},N} \in{\mathbb C}^{LQ(N+1)\times(N+1)}. \label{eq:def_HNEN_multiO_ext}
\ee
Performing the SVD yields the decomposition as $\widetilde{\bf H}^{({\bf O}_{\rm ext})}_{{\bf n},N} = U_N \Sigma_N V^{\dagger}_N$ with $\Sigma_N={\rm diag}(\sigma_1,\cdots,\sigma_N)\in{\mathbb R}^{N\times N}$, $U_N\in{\mathbb C}^{QL(N+1)\times N}$, and $V_N\in{\mathbb C}^{(N+1)\times N}$ (see App.~\ref{app:Mrel_proof} for details).
Below, we assume that $LQ\ge N^2$ and $D_{\rm obs}=N^2$ by taking $Q$ large enough and choosing a suitable ${\bf O}$, so that the sampled time-shifted observables resolve the full visible operator space without additional temporal aliasing, namely without distinct Liouville-frequency components becoming indistinguishable on the chosen time grid, or degeneracy.
Under this assumption, the reconstruction of $M$ can be performed from the SVD profile $(\Sigma_N,U_N,V_N)$.

The main strategy is to consider the decomposition of $\widetilde{\bf H}^{({\bf O}_{\rm ext})}_{{\bf n},N}$ in Eq.\eqref{eq:def_HNEN_multiO_ext} as
\be
\widetilde{\bf H}^{({\bf O}_{\rm ext})}_{{\bf n},N} = {\cal A} \cdot {\cal B}, \label{eq:Hankel_AB}
\ee
and ${\cal A}\in {\mathbb C}^{LQ(N+1) \times N}$ and ${\cal B} \in {\mathbb C}^{N \times (N+1)}$ are determined by imposing a condition.
For the SS cases, it is given by
\be
   {\cal A}_{\alpha j} = \llangle X_{\alpha} | P_j \rrangle, \qquad {\cal B}_{jk} = \lambda_{j}^{k}, \qquad X_{\alpha} := O_\ell (t_q) M^{n} \ \ \text{with} \ \ \alpha = (n,\ell,q), \ \ n \in {\bf n}, \label{eq:A_B_cond}
\ee
where $P_j$ denotes the projection operator in Eq.\eqref{eq:proj}.
Recall that ${\bf n} \in {\mathbb N}_0^{N+1}$ with $n_{j_1} \ne n_{j_2}$ for all $j_1 \ne j_2$ in Eq.\eqref{eq:def_HNEN}.
Indeed, using ${\cal A}$ and ${\cal B}$, the time-shifted OTS can be written by
\be
\zeta^{(O_\ell)}_{n+k,t_q} = \llangle X_{(n+k,\ell,q)} | {\mathbb I}_N \rrangle = \sum_{j =1 }^N  \llangle X_{(n,\ell,q)} | P_j \rrangle \lambda_j^k = \sum_{j=1}^N {\cal A}_{\alpha j} {\cal B}_{jk}. 
\ee
From the SVD profile $(\Sigma_N, U_N, V_N)$, ${\cal A}$ and ${\cal B}$ can be written by
\be
   {\cal A} = U_N \Sigma_N {\cal G}^{-1}, \qquad {\cal B} = {\cal G} V^{\dagger}_N, \qquad {\cal G} \in {\mathbb C}^{N \times N}, \label{eq:AGBG}
\ee 
and one can determine ${\cal G}$ to satisfy Eq.\eqref{eq:A_B_cond}.
This can be performed by using Estimation of Signal Parameters via Rotational Invariance Techniques (ESPRIT)~\cite{Paulraj1985}, together with a diagonal eigenvalue-matrix $\Lambda$:
\be
  {\cal V}_1 {\cal V}_0^{-1} =   {\cal G}^{-1} \Lambda {\cal G}  \quad \text{for SS} 
\qquad {\cal V}_{n=0,1} = (V_N^{\dagger})_{1 \le  \ell \le N , n+1 \le  k \le n+N}. \label{eq:V0V1inv}
\ee
By finding ${\cal G}$ by ESPRIT, ${\cal A}$ is obtained, and then, we compute $| P_j \rrangle$ from ${\cal A}$. 
Here, we define
\be
{\cal K}_{0} =
\begin{pmatrix}
  \llangle X_{({0},1,0)} | \\
  \vdots \\
  \llangle X_{({0},L,Q-1)} |
\end{pmatrix} \in {\mathbb C}^{LQ \times N^2}, \qquad {\bf a}^{(n)}_j  := \lambda_j^{n}
\begin{pmatrix}
  \llangle X_{({0},1,0)} | P_j \rrangle\\
  \vdots \\
  \llangle X_{({0},L,Q-1)} | P_j \rrangle
\end{pmatrix} \in {\mathbb C}^{LQ}, \label{eq:K_a}
\ee
where ${\bf a}^{(n)}_j$ is embedded as a ($N^2$-dimensional projected) column vector in ${\cal A}$. 
Since ${\cal K}_0$ is known, $| P_j \rrangle$ is determined as
\be
{\cal K}_0 | P_j \rrangle = \lambda_j^{-n} {\bf a}^{(n)}_j \quad \Rightarrow \quad 
| P_j \rrangle = \lambda_j^{-n} {\cal K}_0^{+} {\bf a}^{(n)}_j \in {\mathbb C}^{N^2}, \label{eq:K0Pj}
\ee
with pseudo-inverse ${\cal K}_0^+=({\cal K}_0^\dagger {\cal K}_0)^{-1} {\cal K}_0^\dagger$.
The label of bases in the Liouville space expression is directly assigned from the index of the matrix representation of $GL(N,{\mathbb C})$, which can be seen from the basis expansion,  i.e. $| P_j \rrangle = \sum_{ k_1, k_2 =1}^N (P_j)_{k_1,k_2} | u_{k_1} \rangle \otimes | \widetilde{u}_{k_2} \rangle^{*}$, and as a result, $M$ is realized from Eq.\eqref{eq:MPj}.


For the EP cases, as one can see from Eq.\eqref{eq:M_EP_P_N}, the monodromy matrix $M$ consists of not only the projection operators but also the nilpotent matrices $N_j$.
The composition of the Hankel matrix in Eqs.\eqref{eq:Hankel_AB}\eqref{eq:AGBG} still works, but the condition for ${\cal A}$ and ${\cal B}$ in Eq.\eqref{eq:A_B_cond} is modified as
\be
  {\cal A}_{\alpha (j,m)} =
   \begin{cases}
     \llangle X_{\alpha} | P_j \rrangle  & \ \ \text{for} \ \ m = 1 \\
     \llangle X_{\alpha} | N_j^{m-1} \rrangle  & \ \ \text{for} \ \ m > 1 
   \end{cases},
   && \qquad {\cal B}_{(j,m) k} = 
   \binom{k}{m-1} \lambda_j^{k-m+1}, \label{eq:A_B_cond_EP}
\ee
with $j \in \{1,\cdots,K\}$ and $m \in \{1,\cdots,d_j\}$.
As a result, the decomposition of ${\cal V}_1 {\cal V}_0^{-1}$ in Eq.\eqref{eq:V0V1inv} generates a Jordan normal form ${\cal J}$ instead of $\Lambda$ as
\be
  {\cal V}_1 {\cal V}_0^{-1} =    {\cal G}^{-1} {\cal J} {\cal G} & \quad \text{for EPs},
\qquad {\cal V}_{n=0,1} = (V_N^{\dagger})_{1 \le  \ell \le N , n+1 \le  k \le n+N}. \label{eq:V0V1inv_EP}
\ee
The column vector in Eq.\eqref{eq:K_a}, ${\bf a}^{(n)}_{j}$ changes as
\be
\qquad {\bf a}^{(0)}_{(j,1)}
:= 
\begin{pmatrix}
  \llangle X_{({0},1,0)} | P_j \rrangle\\
  \vdots \\
  \llangle X_{({0},L,Q-1)} | P_j \rrangle
\end{pmatrix} \in {\mathbb C}^{LQ}, \qquad
{\bf a}^{(0)}_{(j,m>1)}  := 
\begin{pmatrix}
  \llangle X_{({0},1,0)} | N_j^{m-1} \rrangle\\
  \vdots \\
  \llangle X_{({0},L,Q-1)} | N_j^{m-1} \rrangle
\end{pmatrix} 
\in {\mathbb C}^{LQ}, \label{eq:K_a_EP}
\ee
and ${\bf a}^{(n)}_{(j,m)}$ with a higher $n$ can be obtained by acting the Jordan from ${\cal J}$ as
\be
   {\cal A}^{(n)}_{j} =  {\cal A}^{(0)}_{j} {\cal J}^n, \qquad {\cal A}^{(n)}_j =\left( {\bf a}^{(n)}_{(j,1)}, \cdots, {\bf a}^{(n)}_{(j,d_j)}  \right) \in {\mathbb C}^{QL \times d_j}. \label{eq:AnjA0j}
\ee
By extracting ${\cal A}^{(n)}_j$ from ${\cal A}$ and constructing ${\cal A}^{(0)}_j$ from Eq.\eqref{eq:AnjA0j}, one can find
\be
| P_j \rrangle = {\cal K}_0^{+} {\bf a}^{(0)}_{(j,1)}, \qquad | N_j^{m-1} \rrangle = {\cal K}_0^{+} {\bf a}^{(0)}_{(j,m>1)}. \label{eq:PjK_NjM}
\ee
Since the spectrum is found from the Jordan normal form ${\cal J}$, the exact monodromy matrix $M$ is realized using Eq.\eqref{eq:M_EP_P_N}.

The logic based on the micromotion for the full reconstruction is directly applicable to both the fundamental and adjoint-type OTSs incorporating the boundary/density operator $\Omega$, as discussed in Secs.~\ref{sec:ampli_type_data} and \ref{sec:adjoint_type_data}.
In these MIMO settings, the realized superoperator ${\cal M}_{\rm rel}$ belongs to $GL(N^2,{\mathbb C})$ or $SL(N^2,{\mathbb C})$, and the information required to reconstruct the original monodromy matrix is encoded in the realized operator, provided the corresponding observability and calibration conditions are satisfied.
In particular, as discussed in Sec.~\ref{sec:ampli_type_data}, this framework successfully performs the exact reconstruction of the DP cases, which cannot be fully resolved by using solely the ordinary trace-type OTS.


Finally, let us discuss the reconstructibility for $M_{\rm rel}$ in the procedure above.
In the current procedure, the observable dimension $D_{\rm obs}$ in Eq.\eqref{eq:Dobs_def} can be written using ${\cal K}_0$ in Eq.\eqref{eq:K_a} as
\be
D_{\rm obs} = {\rm rank}({\cal K}_0). \label{eq:Dobs_K0}
\ee
Here, Eq.\eqref{eq:Dobs_K0} should be understood as the sampled observable dimension associated with the chosen finite family of time-shifted observables.
In general it is bounded by the algebraic visible dimension,
\be
D_{\rm obs}\leq \dim({\frak O}_{\rm ext}^{\prime\prime}),
\ee
and equality requires that the sampled time-shifted observables span the full visible algebra without additional temporal aliasing or degeneracy.
Thus $D_{\rm obs}$ measures what is actually resolved by the chosen finite sampling setup, whereas $\dim({\frak O}_{\rm ext}^{\prime\prime})$ gives the algebraic upper bound associated with the full extended observable algebra.
This meaning and the difference from $\dim ({\frak O}_{\rm ext}^{\prime \prime})$ in Prop.~\ref{prop:symmetry_indicator} can be interpreted as follows:

Suppose a set of the projectors and nilpotent matrices ${\frak Z}$ as
\be
{\frak Z} = 
\begin{cases}
  \{ {P}_j \}_{j=1}^N & \quad \text{for SS}  \\[1mm]
  \{ {P}_j \}_{j=1}^K \cup \big\{ {N}_j^{m} \big\}_{j=1, m=1}^{K, d_j-1} & \quad \text{for EPs}
\end{cases}, \label{eq:Zk}
\ee
and define a projection operator $\Pi_{\cal K}$ from ${\cal K}_0$ and its pseudo-inverse ${\cal K}_0^+$ as
\be
\Pi_{\cal K} := {\cal K}_{0}^+ {\cal K}_0.
\ee
If ${\cal K}_0$ is full-rank, i.e. $D_{\rm obs} = N^2$, then Eqs.\eqref{eq:K0Pj}\eqref{eq:PjK_NjM} can be expressed by
\be
 | \widetilde{Z} \rrangle = \Pi_{\cal K} | Z \rrangle = |Z\rrangle \quad \text{for all} \ \ Z \in {\frak Z}.
\ee
When $D_{\rm obs}< N^2$, only the projected components $\Pi_{\cal K}|Z\rrangle$ are accessible from the sampled observable family, and the orthogonal complement remains invisible to this reconstruction procedure.
The algebraic origin of this rank deficiency can be understood by expanding the time-shifted observables $O_\ell(t_q)$ in their spectral components.
For the SS cases, the expansion is given by
\be
O_\ell(t_q) = \sum_{j_1,j_2=1}^{N} \sum_{m \in {\mathbb Z}} \Xi_{\ell,j_1 j_2}^{(m)} e^{i \Omega_{j_1 j_2}^{(m)} t_q}, \qquad \Omega_{j_1 j_2}^{(m)} = \varepsilon_{j_1} -\varepsilon_{j_2} + 2 \pi \frac{m}{T},
\ee
where the coefficient operators are defined as
\be
\Xi_{\ell,j_1 j_2}^{(m)} = P_{j_1} O^{(m)}_\ell P_{j_2}.
\ee
For the EP cases, the expansion involves polynomial factors in $t_q$ due to the nilpotent structure of the generalized eigenspaces, taking the formal appearance
\be
O_{\ell}(t_q) = \sum_{j_1,j_2=1}^K \sum_{m \in {\mathbb Z}} \sum_{k=0}^{d_{j_1} + d_{j_2}-2} \Xi_{\ell, j_1 j_2}^{(m,k)} e^{i \Omega_{j_1 j_2}^{(m)} t_q} t_q^k,
\ee
where the continuous-time expansion coefficients are strictly obtained from the Taylor series of the evolution operator as
\be
\Xi_{\ell,j_1 j_2}^{(m,k)} = \sum_{\substack{k_1,k_2 \in {\mathbb N}_0 \\ k_1+k_2=k}} \frac{i^{k_1}(-i)^{k_2}}{k_1! k_2!} P_{j_1} {\cal N}_{j_1}^{k_1} O_{\ell}^{(m)} {\cal N}_{j_2}^{k_2} P_{j_2}.
\ee
Here, ${\cal N}_j$ represents the nilpotent part of the effective Floquet Hamiltonian $H_{\rm eff}$ restricted to the $j$-th generalized subspace. 
Crucially, since ${\cal N}_j$ and the reconstructed nilpotent part of the monodromy matrix, $N_j$, are uniquely mapped to each other via a non-singular polynomial equivalence as
\be
N_j = \lambda_j \sum_{n=1}^{d_j-1} \frac{(-iT)^n}{n!} \mathcal{N}_j^n, \qquad \mathcal{N}_j = \frac{i}{T} \sum_{n=1}^{d_j-1} \frac{(-1)^{n-1}}{n} \lambda_j^{-n} N_j^n,
\ee
they generate the same nilpotent algebra relevant to the reconstruction.
Consequently, the linear operator space spanned by the exact set of coefficient operators $\Xi_{\ell,j_1j_2}^{(m,k)}$ is structurally identical to that generated directly by the discrete algebraic components of $M$, ensuring that the sampled growth of $D_{\rm obs}$ is controlled within the monodromy-centric realization framework.
Notice that, in this exact algebraic representation, the SS case is cleanly recovered as a geometric boundary condition by fixing $k=0$ and setting the maximal nilpotent degree to zero ($d_{j} = 1$). 
The set of these exact coefficient operators $\Xi_{\ell,j_1j_2}^{(m,k)}$ fully spans the linear operator space generated by the micromotion-expanded observable family, providing the linear operator space that bounds the sampled observable dimension $D_{\rm obs}$.

The coefficient operators $\Xi_{\ell,j_1j_2}^{(m,k)}$ span the linear space of operator components generated by the micromotion-expanded observable family.
This linear span should be distinguished from the full unital $*$-algebra generated by the same family.
Thus, in general, one has
\be
D_{\rm obs} = {\rm rank}({\cal K}_0) \leq \dim {\rm Span}\left\{ \llangle \Xi_{\ell, j_1 j_2}^{(m, k)} | \mid \forall \ell, m, k, j_1, j_2 \right\}
\leq \dim \left( {\frak O}_{\rm ext}^{\prime \prime} \right). \label{eq:Dobs_span_bound}
\ee
The first inequality becomes an equality when the chosen time samples resolve all independent temporal modes without aliasing.
The second inequality becomes an equality only when the linear span of the coefficient operators is already closed as the relevant visible $*$-algebra generated by the time-shifted observables.
When $D_{\rm obs}$ is smaller than the dimension of the linear span of the coefficient operators, the finite sampling cannot distinguish all algebraically available temporal modes.
A typical mechanism is temporal aliasing or frequency degeneracy.
Namely, if
\be
\Omega_{j_1 j_2}^{(m)} = \Omega_{j^\prime_1 j^\prime_2}^{(m^\prime)}
\ee
holds for $(j_1,j_2) \ne (j^\prime_1,j^\prime_2)$ and admissible $(m,m^\prime)$, then the corresponding vectorized operators $\llangle \Xi_{\ell, j_1 j_2}^{(m, k)} |$ and $\llangle \Xi_{\ell, j^\prime_1 j^\prime_2}^{(m^\prime, k)} |$ can appear only through their combined temporal mode in the sampled sequence.
This merging reduces the rank of ${\cal K}_0$ and contributes to the algebraic deficiency $\Delta_{\rm obs}$.

The EP case can partially mitigate this loss of temporal distinguishability.
The polynomial factors $t_q^k$ generated by Jordan blocks provide additional temporal signatures beyond the pure oscillatory factors.
When the corresponding observable coefficients are nonzero and sufficiently many samples are available, these polynomially dressed modes can help distinguish contributions that would otherwise be merged by frequency degeneracy.
However, this mechanism is not automatic, i.e., if the relevant observable coefficients vanish, if the set of observables is too small, or if different contributions share both the same frequency and the same polynomial order, the corresponding components can still remain unresolved.

The present subsection has treated the strongest reconstruction scenario, in which the observable algebra is rich enough to span $\End({\cal H})$ and the sampled time-shifted observables exhaust the visible operator space.
The realistic case of restricted observability falls back to the partial-realization viewpoint of Sec.~\ref{sec:rel_of_M}, where only the canonical visible representative $M_{\rm vis}$ is determined.
In the following two sections, we examine concrete finite-dimensional non-Hermitian Floquet systems in which both regimes can be probed within the same algebraic language: the driven transmon qutrit (Sec.~\ref{sec:example_transmon}) realizes a setting where the full operator space can be reached by micromotion, while the finite NHFSSH chain (Sec.~\ref{sec:example_ssh}) realizes a setting in which the sampled observability rank is restricted jointly by symmetry, locality, and finite-size structure.

\section{Example 1: Driven transmon qubit with coherent leakage and dissipation} \label{sec:example_transmon}
To demonstrate the physical validity and operational utility of our algebraic tomography framework, we first apply the $GL(3,\mathbb{C})$ formulation to a periodically driven superconducting transmon system.
The transmon is intrinsically a weakly anharmonic multilevel quantum circuit, where coherent leakage outside the designated computational qubit subspace poses a major practical obstacle to high-fidelity quantum control~\cite{Koch2007Transmon,Motzoi2009DRAG}.
To balance analytical transparency with the essential physics of non-unitary leakage, we truncate the system's Hilbert space to the lowest three energy levels. 
The lowest two states constitute the target qubit subspace, while the third level explicitly models the primary leakage channel, thereby effectively rendering the system a driven transmon qutrit.

\subsection{Model definition and purposes}
We define the non-Hermitian Hamiltonian of the DTQ3 in the laboratory frame as
\be
\widehat{H}(t) = \widehat{H}_{\rm diag} + \widehat{H}_{\rm drive}(t) + \widehat{H}_{\rm leak}(t). \label{eq:ex1_H_total}
\ee
The diagonal part $\widehat{H}_{\rm diag}$ is given by
\be \widehat{H}_{\rm diag} = \sum_{j=0}^{2}(E_j-i\gamma_j)\,|j\rangle\langle j| = {\rm diag}(E_0-i\gamma_0,\ E_1-i\gamma_1,\ E_2-i\gamma_2), \label{eq:ex1_H_diag}
\ee
with the bare energies $E_{j=0,1,2} \in {\mathbb R}$ and the effective decay rates $\gamma_{j=0,1,2} \in {\mathbb R}$.
The anharmonicity is encoded in
\be
\alpha=(E_2-E_1)-(E_1-E_0). \label{eq:ex1_anharmonicity}
\ee
We set the coherent drive and the leakage couplings as
\be
&& \widehat{H}_{\rm drive}(t)  :=  A\cos(\omega t) ( |0\rangle\langle 1|+|1\rangle\langle 0| ) = A \cos (\omega t)
\begin{pmatrix}
  \sigma_x  & \\
   & 0 
\end{pmatrix}, \label{eq:ex1_H_drive} \\
&& \widehat{H}_{\rm leak}(t) := B\cos(\omega t) (|1\rangle\langle 2|+|2\rangle\langle 1| )
= B \cos (\omega t)
\begin{pmatrix}
0 & \\
  & \sigma_x  
\end{pmatrix}, \label{eq:ex1_H_leak}
\ee
with the Pauli matrix $\sigma_x = \begin{pmatrix} 0  & 1 \\1 & 0 \end{pmatrix}$, where $A \in {\mathbb R}$ and $B \in {\mathbb R}$ are the drive and leakage amplitudes, respectively.
Since the Hamiltonian is periodic with a periodic time $T=\frac{2 \pi}{\omega}$, the one-period evolution is described by the monodromy matrix defined in Eq.\eqref{eq:M_def}, which is the natural Floquet object for the present problem \cite{Shirley1965}.

Unless otherwise stated, in the DTQ3, we fix
\be
E_0=0,\qquad E_2=2.05,\qquad (\gamma_0,\gamma_1,\gamma_2)=(0,\,0.02,\,0.40),
\ee
where we set $E_0=0$ as the energy reference.
The other numerical parameters used in the figures are summarized in Tab.~\ref{tab:ex1_params}.
Throughout this example we use the qubit-population observable defined as
\be
O_{\rm qub} := |0\rangle \langle 0|+|1\rangle \langle 1| = {\rm diag}(1,1,0), \label{eq:ex1_Oqub}
\ee
with the corresponding OTS, $\{ \zeta^{(O_{\rm qub})}_n \}_n$, obtained from Eq.\eqref{eq:OTS_def}.
For comparison we also use the difference observable,
\be
O_{\rm diff} := |0\rangle \langle 0| - |1\rangle \langle 1| = \mathrm{diag}(1,-1,0).
\ee

\begin{table}[t]
\centering
\begin{tabular}{|p{0.36\linewidth} | p{0.55\linewidth}|}
\hline
Purpose / figure & Parameters \\
\hline
Reconstruction check (Fig.~\ref{fig:ex1_reconstruction}) & $A=0.94$, $B=A/2$, $E_1\in[0.02,\,2.02]$ (101 points).
\\[0.0em]
Near-degeneracy regime (Fig.~\ref{fig:ex1_funnel}) & $A\in[0.75,\,0.95]$ and $E_1\in[0.05,\,0.14]$ ($81\times81$ grid), $B=A/2$.
\\ [0.0em]
Visibility mismatch (Fig.~\ref{fig:ex1_visibility}) & $A=0.94$, $B=A/2$, $E_1\in[0.02,\,2.02]$ (501 points).
\\ [0.0em]
Branch-wise phase allocation (Fig.~\ref{fig:ex1_phase_mixing}/Left) & $A\in\{0.10,\,0.30,\,0.94,\,2.00\}$, $B=A/2$, $E_1\in[0.02,\,2.02]$ (501 points).
\\[0.0em]
\qquad \qquad  \qquad \qquad \qquad \qquad \ \ \ (Fig.~\ref{fig:ex1_phase_mixing}/Right) & $A\in[0.10,\,2.00]$ (13 points), $B=A/2$, $E_1\in[0.02,\,2.02]$ (501 points).
\\[0.0em]
Observable-dimension growth (Fig.~\ref{fig:dtq3_rank_saturation}) 
& $A=0.94$, $B \in \{0, A/2 \}$, $E_1=1.00$.
\\ [0.0em]
\hline
\end{tabular}
\caption{Numerical parameter choices used in our numerical experiments of the DTQ3.
}
\label{tab:ex1_params}
\end{table}

The purposes of this example are as follows:
\begin{enumerate}
    \item We provide a minimal $GL(3,\mathbb C)$ Floquet setting that explicitly verifies the algebraic reconstruction of the common spectral skeleton from a single observable channel.
    \item We examine the distinction between the common spectral skeleton and observable-dependent dressing by studying how drive-induced mode mixing changes the visible phase response relative to the global determinant-phase accumulation.
    \item We use the micromotion-expanded observable dimension $D_{\rm obs}=\mathrm{rank}({\cal K}_0)$ to test whether the qubit-population observable remains confined to the invariant qubit sector or expands to the full qutrit operator space through coherent leakage.
\end{enumerate}

\subsection{Numerical experiments}

\subsubsection{Algebraic reconstruction from a single observable} \label{sec:alg_rec}
We first check the $GL(N,{\mathbb C})$ structure using the recurrence relation in Eq.\eqref{eq:CH_bn_recurrence} and the Prony reconstruction, which is summarized in App.~\ref{app:prony}.
We use $O_{\rm qub}$ in Eq.\eqref{eq:ex1_Oqub} and numerically reconstruct the common spectral skeleton, namely the characteristic coefficients $\{e_a\}_{a=1}^3$ and the spectrum $\{\lambda_j\}_{j=1}^3$, from the first six points of the OTS $\{\zeta_n^{(O_{\rm qub})}\}_{n=0}^5$.
The corresponding observable-dependent coefficients $\{c_j^{(O_{\rm qub})}\}_{j=1}^3$ are then obtained straightforwardly once the spectral skeleton is fixed, although we do not display them explicitly here.

The result is shown in Fig.~\ref{fig:ex1_reconstruction}, which shows that the reconstruction errors remain at the level of numerical roundoff throughout the scan.
The reconstructed characteristic coefficients and eigenvalues agree with those obtained by direct diagonalization to essentially machine precision.
Thus, even a single observable channel already recovers the full $GL(3,\mathbb C)$ spectral skeleton. 
In the present noiseless $GL(3,\mathbb C)$ example, the reconstructed values of $\{e_a\}_a$, $\{\lambda_j\}_j$, and $\{c_j^{(O)}\}_j$ are numerically indistinguishable from the exact ones at plotting resolution, with typical differences of order $10^{-13}$ or smaller for the spectral quantities.
In the analyses below, we therefore use the exact spectral data obtained from direct diagonalization of the monodromy matrix.
Accordingly, the plots below should be understood as visualizations of the same reconstructed spectral skeleton.

\begin{figure}[t]
\centering
\includegraphics[width=.68\linewidth]{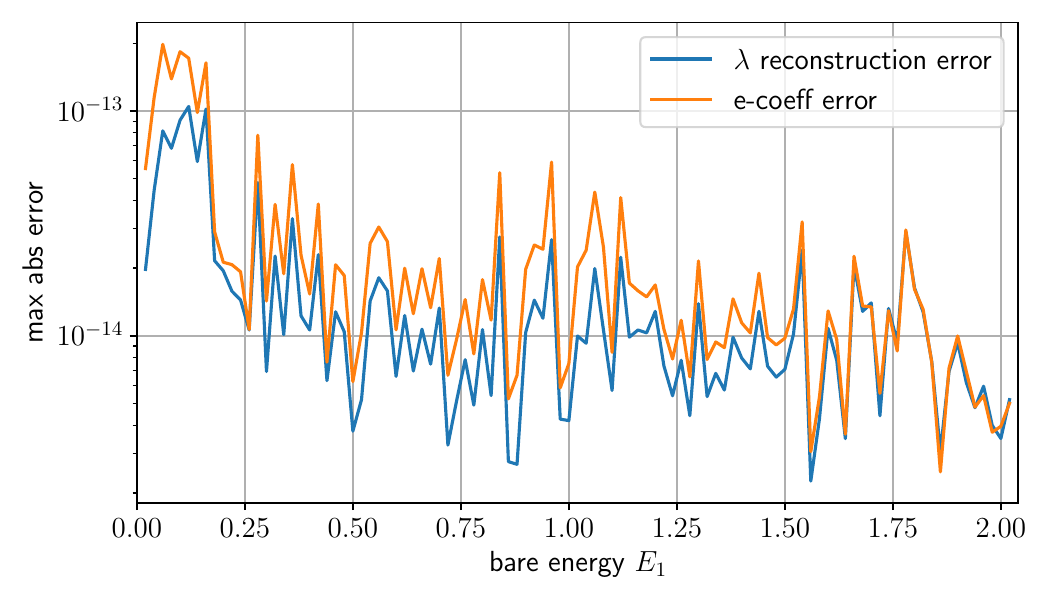}
\caption{Reconstruction errors for the DTQ3 from a single observable $O_{\rm qub}$. The errors in the reconstructed eigenvalues and characteristic coefficients remain close to machine precision over the full scan, showing that a single observable channel already recovers the full $GL(3,\mathbb C)$ spectral skeleton.}
\label{fig:ex1_reconstruction}
\end{figure}

\subsubsection{Observable phase visibility and level allocation near DPs} \label{sec:phase_vis}

\begin{figure}[t]
\centering
\includegraphics[width=.525\linewidth]{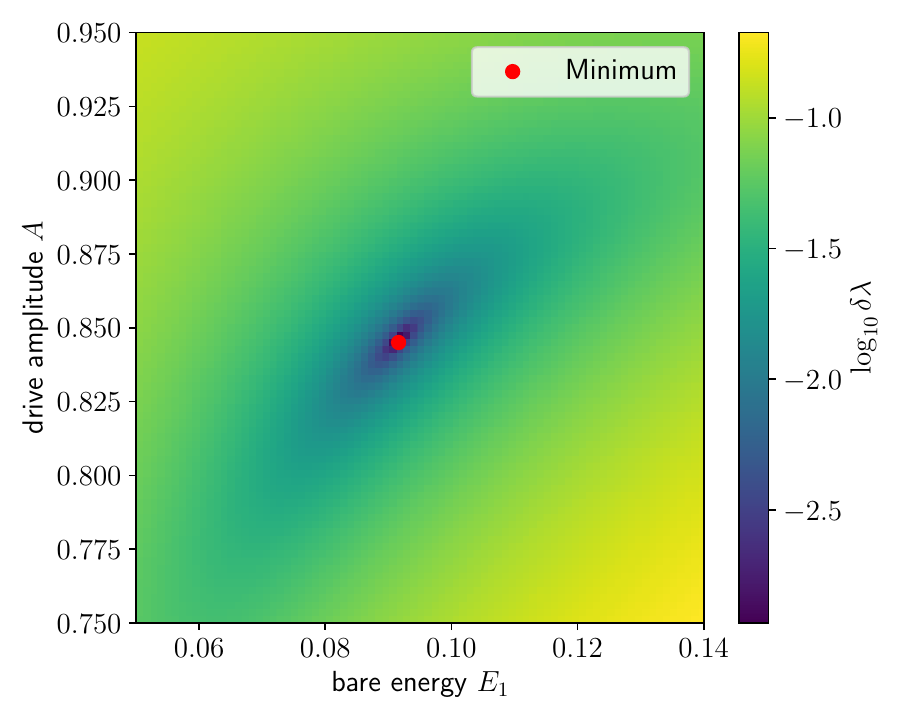}
\includegraphics[width=.425\linewidth]{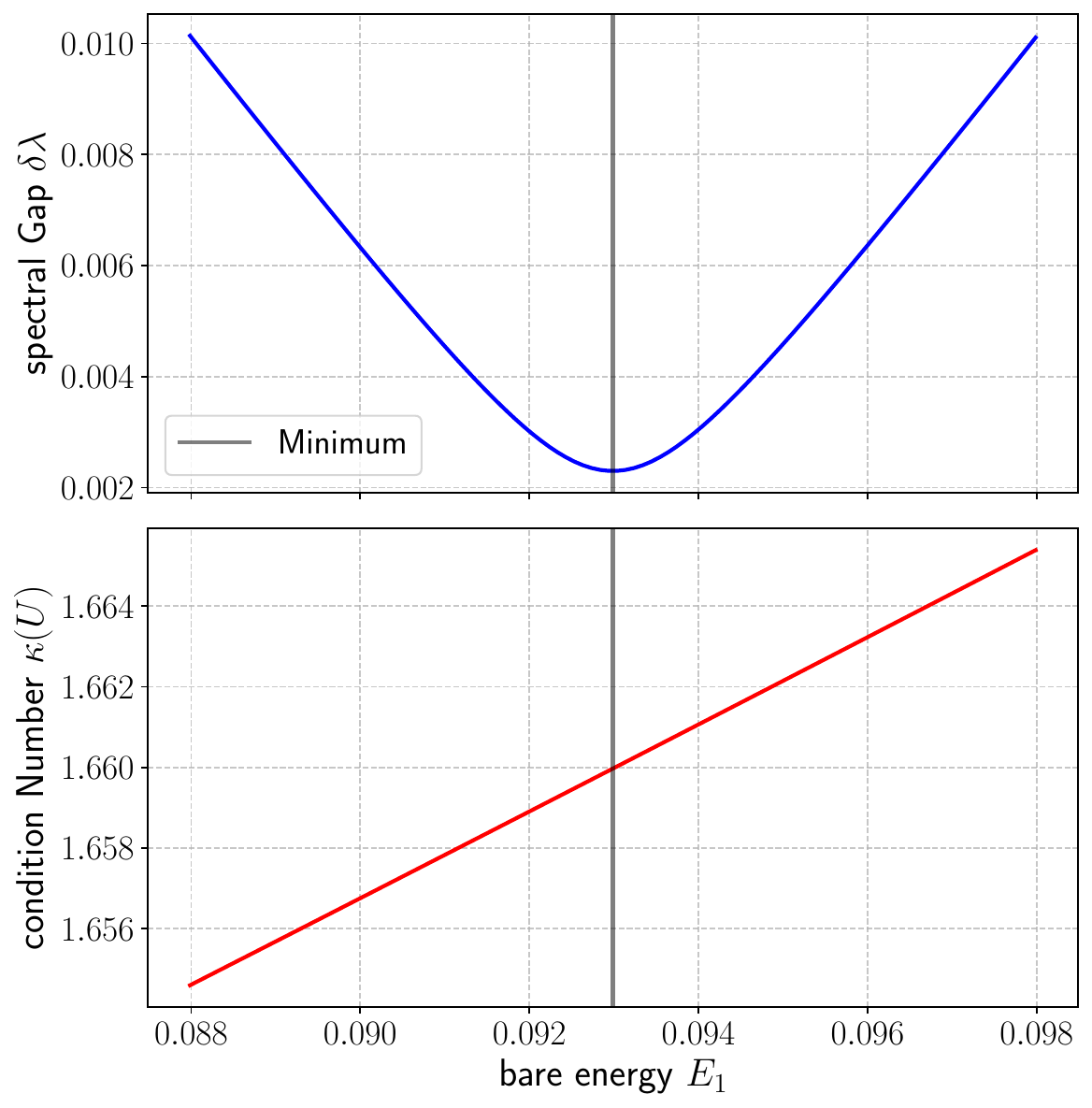}
\caption{Near-degeneracy point characterized in the two-parameter plane of the bare energy $E_1$ and the drive amplitude $A$. (Left) Color plot of the minimal spectral gap $\delta \lambda$ on a $\log_{10}$ scale. (Right) The $E_1$-dependence of the spectral gap $\delta \lambda$ and the matrix condition number $\kappa(U)$ evaluated at a fixed drive amplitude $A = 0.85$, demonstrating the stable, DP nature of the degeneracy.}
\label{fig:ex1_funnel}
\end{figure}

We first identify the parameter regime near degeneracy points.
Fig.~\ref{fig:ex1_funnel} visualizes this near-degeneracy point through the spectral gap and the condition number of $U$, $\delta \lambda$ and $\kappa(U)$, defined as
\be
\delta \lambda := \min_{i<j}|\lambda_i-\lambda_j|, \qquad \kappa(U) := \| U \|_2 \cdot \| U^{-1} \|_2,
\label{eq:ex1_min_gap}
\ee
in the $E_1$-$A$ plane, where $\| U \|_2$ is the matrix-norm of $U = (| u_1 \rangle, | u_2 \rangle, | u_3 \rangle)$ with the right eigenvectors of $M$, $ \{ | u_j \rangle \}_{j=1}^3$.
This result shows that this degeneracy point corresponds to a DP because the condition number $\kappa(U)$ has no sign of rapid growth and is stable near the degeneracy point.
For our purposes, the essential point is not whether the minimum of $\delta \lambda$ contains an exact DP, but rather that the reconstructed spectral skeleton directly pinpoints the region where phase redistribution is highly sensitive.
Below, we examine how this underlying sensitivity manifests in the observable response.

In order to investigate the observable response beyond the OTS, we consider the ODSD in Eq.\eqref{eq:hF_lam}, whose ORS poles are located at $z_j=\lambda_j^{-1}$.
For comparison, we also define the global determinant-phase accumulation along a continuous path of parameters by
\be
   {\cal A}_{\det} := \frac{1}{2\pi} \Delta \arg \det M =  \frac{1}{2 \pi} \int_{\cal C} d (\arg \det M)  = \frac{1}{2\pi} \sum_{j=1}^{N}\Delta \arg \lambda_j, \label{eq:ex1_global_phase}
\ee
where $\Delta\arg$ denotes the net unwrapped change of the argument between the initial and final points of the path, ${\cal C}$.
In the present DTQ3 analysis, we vary $E_1 \in [0.02, 2.02]$ along the path, and consequently the total determinant phase ${\cal A}_{\det}$ remains robust and close to two full $2\pi$-rotations. 
Since the parameters are varied along an open path, $\mathcal A_{\det}$ is generally not a strictly quantized integer topological invariant but an accumulated phase quantity. 
To strictly isolate the effect of the observable dressing, we evaluate the unfiltered response using the trivial observable $O={\mathbb I}_3$. 
In the same way, we define the accumulated observable phase ${\cal A}_p^{(O)}$ and introduce the visibility loss $\Delta_{\rm vis}^{(O)}(p)$ to quantify the mismatch as
\be
   {\cal A}_p^{(O)} := \frac{1}{2\pi} \Delta \arg \mathcal F^{(O)}(p), \qquad
   \Delta_{\rm vis}^{(O)}(p) := {\cal A}_p^{({\mathbb I}_3)} - {\cal A}_p^{(O)}. \label{eq:ex1_visibility_loss}
\ee
Thus, the quantity $\Delta_{\rm vis}^{(O)}(p)$ purely measures how much of the unfiltered spectral phase accumulation fails to remain visible in the observable ODSD channel due to the destructive interference of the observable-specific weights $\{ c_j^{(O)} \}_{j}$.

\begin{figure}[t]
\centering 
\includegraphics[width=.85\linewidth]{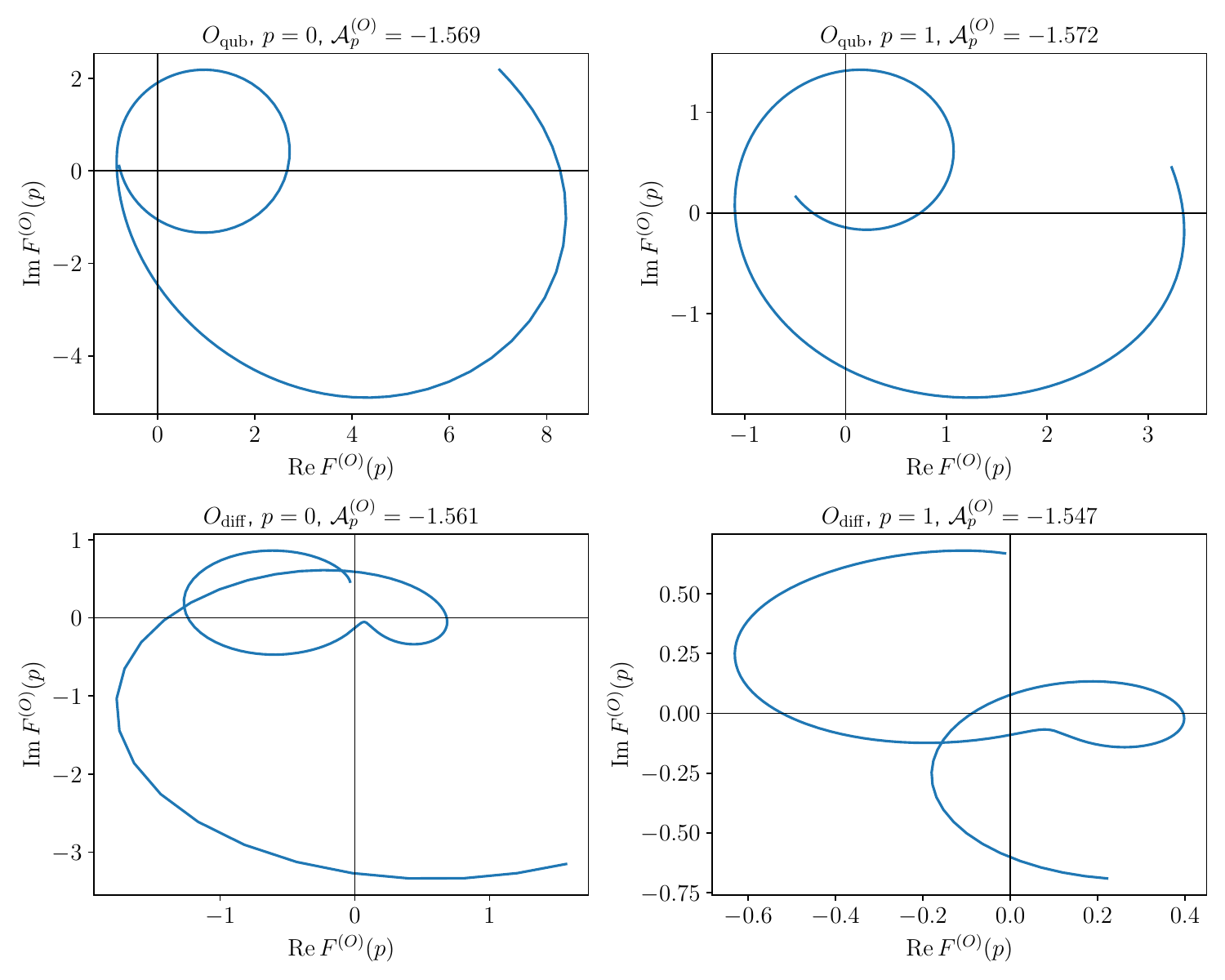}
\caption{Complex-plane trajectories of the ODSD $\mathcal{F}^{(O)}(p)$ for $p=0$ and $p=1$ under the variation of the bare energy $E_1 \in [0.02, 2.02]$. We compare the total qubit population $O_{\rm qub}$ with the population difference $O_{\rm diff}$. The continuous deformation of these loops under drive-induced dressing highlights the onset of severe destructive interference, serving as the geometric precursor to the quantized visibility jumps shown in Fig.~\ref{fig:ex1_phase_mixing}.}
\label{fig:ex1_visibility}
\end{figure}

\begin{figure}[t]
\centering
\includegraphics[width=1.\linewidth]{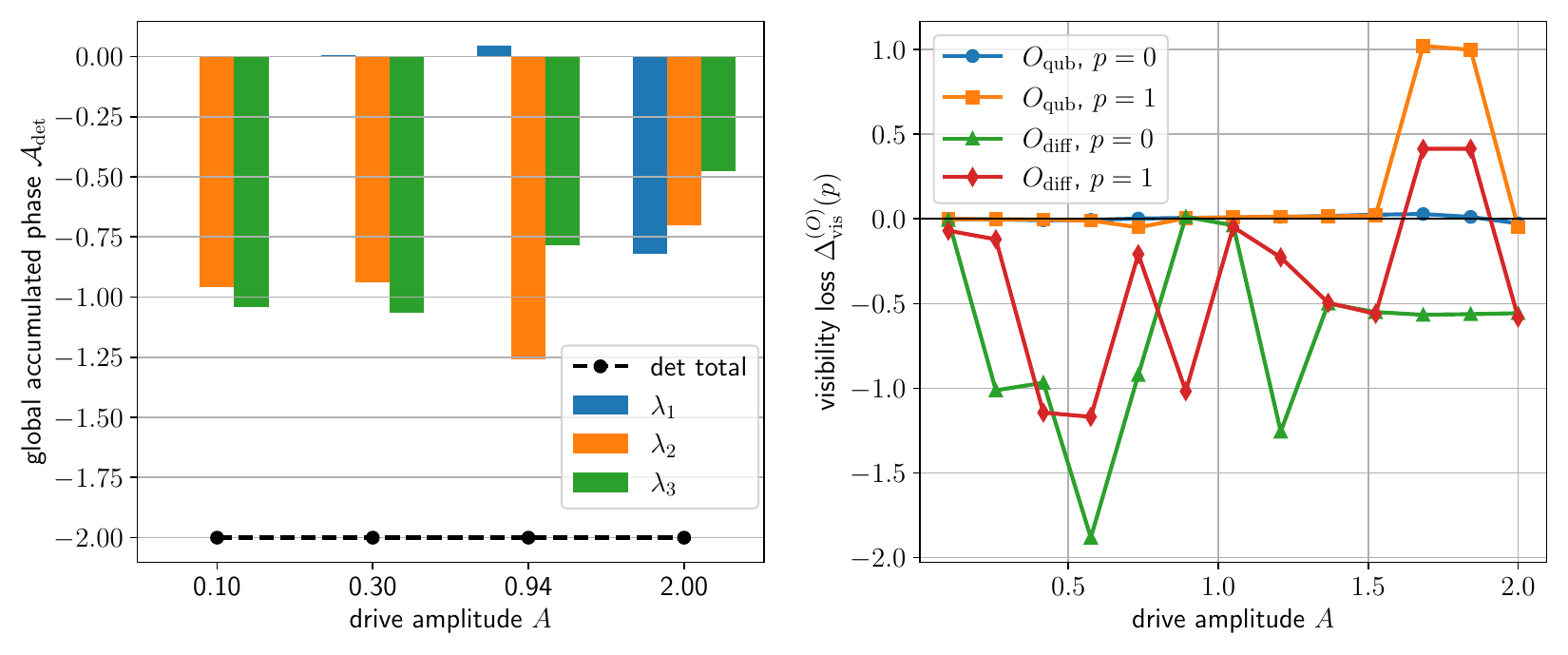}
\caption{(Left) Drive-induced redistribution of the mode-wise phase accumulation observed by $O_{\rm qub}$ as a function of the drive amplitude $A$. (Right) The visibility loss $\Delta_{\rm vis}^{(O)}(p)$ for $p=0,1$ evaluated for the total population $O_{\rm qub}$ and the difference observable $O_{\rm diff}$. The sharp, quantized steps in the visibility loss reflect winding-changing events where the continuous deformation of the dressing coefficients forces the complex trajectory of $\mathcal{F}^{(O)}(p)$ to cross the origin, abruptly shedding a full $2\pi$ phase rotation.}
\label{fig:ex1_phase_mixing}
\end{figure}

Figs.~\ref{fig:ex1_visibility} and \ref{fig:ex1_phase_mixing} show the mismatch between the spectral phase accumulation and the channel-dependent observable response.
Using the same $GL(3,\mathbb{C})$ monodromy matrix, Fig.~\ref{fig:ex1_visibility} demonstrates that the apparent phase trajectories of $\mathcal F^{(O)}(p)$ vary depending on the chosen observable, $O_{\rm qub}$ or $O_{\rm diff}$, and the index $p$.
To unpack this discrepancy, the left panel of Fig.~\ref{fig:ex1_phase_mixing} resolves the global determinant phase into individual eigenmode contributions.
It reveals that the drive amplitude $A$ induces phase mixing and splitting.
At relatively weak drive, the phase accumulation is concentrated in one or two levels, whereas at stronger drive, the levels hybridize and the accumulated phase is split and redistributed more evenly.
Thus, while the total determinant phase remains nearly unchanged, the mode-wise allocations are substantially reshuffled.
The right panel of Fig.~\ref{fig:ex1_phase_mixing} then quantifies this effect through the visibility loss $\Delta_{\rm vis}^{(O)}(p)$.
Strikingly, the visibility loss $\Delta_{\rm vis}^{(O)}(p)$ exhibits sharp, quantized jumps (e.g., dropping to exact integer values such as $-1$) rather than a continuous degradation.
This reflects a winding-changing event in the observable trajectory:
as the drive $A$ alters the dressing coefficients, the complex trajectory of $\mathcal F^{(O)}(p)$ is continuously deformed.
When this deformed loop crosses the origin due to severe destructive interference among the branches, the observable abruptly loses a full $2\pi$ phase rotation.
Particularly, the population difference $O_{\rm diff}$ is shown to suffer from profound visibility losses at moderate to strong drives, whereas the total qubit-population $O_{\rm qub}$ remains relatively robust against the interference.

This mismatch has a clear structural origin.
The unfiltered response ${\mathcal F}^{({\mathbb I}_3)}(p)$ captures the sum of all spectral levels equally.
In contrast, the observable response $\mathcal F^{(O)}(p)$ is a vector sum weighted by the observable coefficients $\{c_j^{(O)}\}_j$, which allows different eigenmodes to interfere and reduce the net rotation.
Furthermore, the polylogarithmic weight suppresses the large-$n$ part of the spectral history when $p$ is increased, making the response progressively more local in time.
Consequently, what the observable sees is not simply the total phase, but the drive-induced phase redistribution filtered through its own dressing-coefficients profile.
Therefore, the apparent phase accumulation and the resulting visibility loss are not accidental, but direct, channel-dependent measures of how much of the underlying spectral geometry is suppressed by the chosen measurement.

From a forward-problem perspective, such severe observable dependence might appear as a fatal loss of global phase information.
However, within the present inverse-problem framework, the algebraic separation of the common spectral skeleton allows us to invert this logic, i.e., the visibility loss itself becomes a quantitative measure.
Rather than being discarded as an error, the mismatch explicitly detects the onset of strong level hybridization and branch competition as the system approaches the DP.

\subsubsection{Tracking the dimensionality of the visible operator space} \label{subsubsec:dtq3_rank_saturation}

\begin{figure}[t]
\centering 
\includegraphics[width=.8\linewidth]{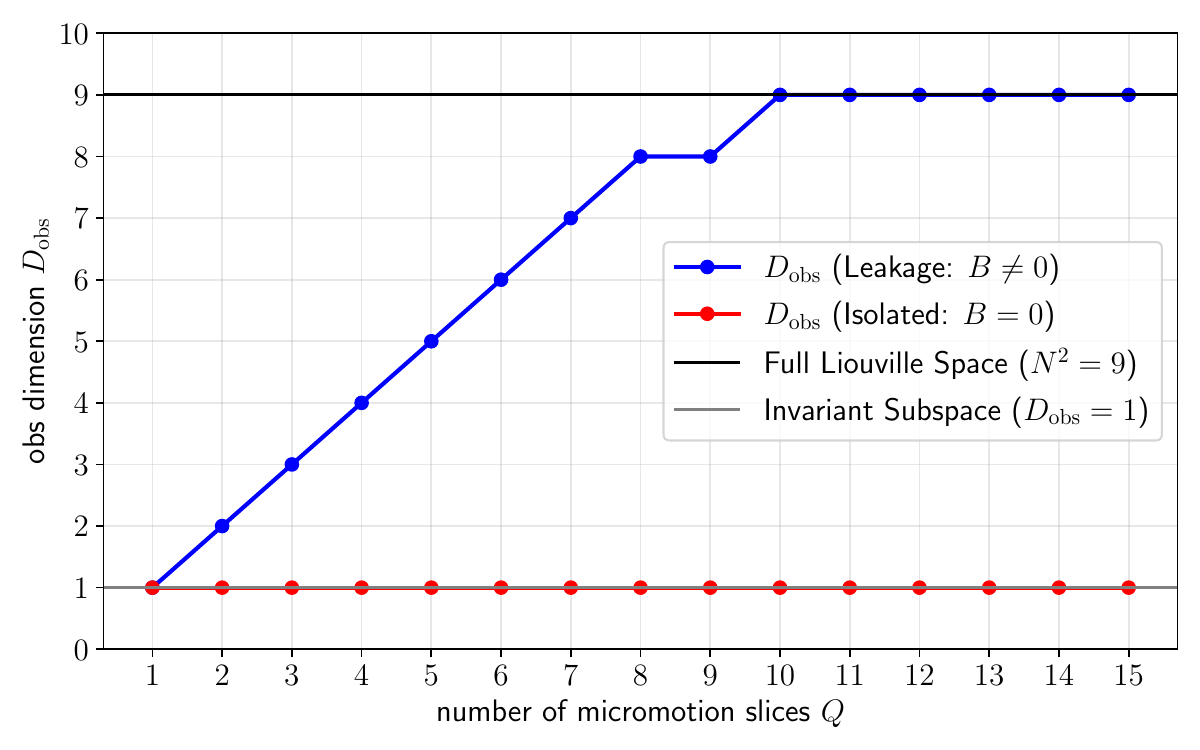}
\caption{Sampled observable dimension $D_{\rm obs} = \mathrm{rank}(\mathcal{K}_0)$ in the DTQ3 model plotted against the number of micromotion slices $Q$. In the isolated qubit limit ($B=0$), the observable remains trapped within a one-dimensional invariant sector ($D_{\rm obs}=1$). In contrast, turning on coherent leakage ($B=0.47$) dynamically accesses the third level, triggering an algebraic expansion of the visible operator space that reaches the full qutrit operator-space dimension $N^2=9$.}
\label{fig:dtq3_rank_saturation}
\end{figure}

We next use the DTQ3 model to illustrate how the observable dimension $D_{\rm obs}$ detects the enlargement of the visible operator space caused by coherent leakage.
Here $D_{\rm obs}$ is not the recurrence order of a single scalar OTS.
Rather, following the definition in Sec.~\ref{sec:reconstruct_Mex}, it is computed as the rank of the sampled observable map ${\cal K}_0$ generated by the time-shifted observable family.
In the present example, this family is generated from the qubit-population observable $O_{\rm qub}$ in Eq.\eqref{eq:ex1_Oqub} through the micromotion orbit given by Eq.\eqref{eq:def_obs_dt_ext} as ${\bf O}_{\rm ext} = \{O_{\rm qub }(t_q) \}_{q=0}^{Q-1}$, i.e.,
\be
O_{\rm qub}(t_q) = U(t_q,0)^{-1}O_{\rm qub}U(t_q,0), \qquad t_q= \frac{q}{Q} T, \qquad q \in \{ 0,\ldots,Q-1\}. \nn
\ee
Numerically, we form the matrix ${\cal K}_0$ by stacking the vectorized operators $O_{\rm qub}(t_q)$ and define the observable dimension $D_{\rm obs}$ by Eq.\eqref{eq:Dobs_K0}.
Thus, $D_{\rm obs}$ measures how many independent operator directions are actually sampled by the micromotion-expanded observable family.

In the isolated qubit limit, $B=0$, the third level is dynamically decoupled from the computational subspace.
The observable $O_{\rm qub}$ then acts as the identity on the closed two-level sector and remains invariant under the micromotion, $O_{\rm qub}(t)  = O_{\rm qub}$.
Consequently, all sampled observables are proportional to the same operator direction, and the sampled observable map has rank $D_{\rm obs}=1$.
This is the sense in which the isolated qubit case remains confined to an invariant visible sector.
This statement should be distinguished from the scalar OTS recurrence order: the effective two-level dynamics can still generate a degree-2 CH recurrence, but the micromotion-expanded observable family itself spans only one operator direction.
Once coherent leakage is turned on, the situation changes.
For the representative choice $B=A/2=0.47$, the third level becomes dynamically accessible, and the conjugated observable $O_{\rm qub}(t)$ no longer remains fixed.
As the number of micromotion slices $Q$ is increased, the sampled operators explore additional directions in the qutrit operator space.
Fig.~\ref{fig:dtq3_rank_saturation} shows that the observable dimension then rises from the isolated value $D_{\rm obs}=1$ and eventually reaches the full observable dimension  $D_{\rm obs}=N^2=9$.
In this sense, coherent leakage is detected algebraically as an expansion of the visible operator space.
This behavior is consistent with the partial-realization viewpoint in Prop.~\ref{prop:Alg_rel_M}.
When $B=0$, the observable algebra generated by $O_{\rm qub}$ and its micromotion orbit is too small to identify the full qutrit monodromy, leaving a nontrivial trace-invisible sector.
When $B\neq0$, the same initial observable is transported by the micromotion into directions involving the leakage level, and the sampled observable map can become full rank.
Thus the transition from $D_{\rm obs}=1$ to $D_{\rm obs}=9$ provides a direct algebraic readout of whether the restricted qubit description remains closed or whether the full qutrit operator space has become visible.

We emphasize that this observable-dimension analysis concerns the span of the micromotion-expanded observable family.
It is therefore complementary to the spectral reconstruction in Sec.~\ref{sec:alg_rec} and to the phase-visibility analysis in Sec.~\ref{sec:phase_vis}.
The former reconstructs the common spectral skeleton of the monodromy matrix, while the present analysis asks how much of the operator space is actually accessible to the chosen observable family.

\section{Example 2: Finite-size non-Hermitian Floquet SSH chain} \label{sec:example_ssh}
To investigate the rich interplay between spatial topology, non-Hermitian spectral geometry, and dynamical visibility limits, we apply our algebraic framework to a periodically driven, non-Hermitian extension of the Su--Schrieffer--Heeger (SSH) model.
While the original SSH formulation serves as a paradigmatic blueprint for topological insulators~\cite{SuSchriefferHeeger1979,AsbothOroszlanyPalyi2016}, its non-Hermitian variants have emerged as a standard platform for exploring non-reciprocal transport, directional amplification, and the non-Hermitian skin effect (NHSE)~\cite{YaoWang2018,Gong2018,Borgnia2020,OkumaKawabataShiozakiSato2020,Manna2022}.
To explicitly capture this non-reciprocal physics within an algebraically exact setup, we employ a finite-size bipartite chain where non-reciprocity is implemented via asymmetric tunneling amplitudes.
The stroboscopic evolution of this network yields a highly non-trivial monodromy matrix in $GL(N, \mathbb{C})$, providing an ideal testing ground for verifying how spatial structures, disorder-induced localization, and underlying physical symmetries dictate the algebraic constraints and sampled observable-dimension growth.

\subsection{Model definition and purposes}
\label{subsec:example_2_ssh}
As a second demonstration, we consider a finite non-Hermitian Floquet SSH (NHFSSH) chain with an even number of sites, $N\in 2{\mathbb N}$.
We consider a two-step procedure in the site basis ordered as
\be
&&   {\bf s} = (0,1, \cdots, N -2, N -1 ) = (A_1,B_1,A_2,B_2,\cdots,A_{\frac{N}{2}},B_{\frac{N}{2}}), \nl
&& A_m:=2m-2,\ \ B_m:=2m-1, \qquad m \in \{1,2,\cdots,\tfrac{N}{2} \},
\ee
where ${\bf s}$ denotes a set of lattice sites on the chain, and $A_m$ and $B_m$ are the even- and odd-labeled sites, respectively.
We define one-period monodromy matrix using two Hamiltonians, $\widehat{H}_{1,2}$, as
\be
M = e^{-i\frac{T}{2}\widehat{H}_2} e^{-i\frac{T}{2}\widehat{H}_1}, \label{eq:ex2_monodromy}
\ee
for each half of the driving period, $T\in{\mathbb R}_{>0}$.
The first- and second-step Hamiltonians are defined by
\be
\widehat{H}_1 &:=& \sum_{m=1}^{N/2} (v + \delta v_m) \left( |A_m\rangle\langle B_m| + |B_m\rangle \langle A_m| \right)
+ \sum_{s=0}^{N-1} \delta \varepsilon_s  |s\rangle\langle s| + i\gamma \Sigma_z, \label{eq:ex2_H1_explicit} \\
\widehat{H}_2 &:=& \sum_{m=1}^{N/2-1} (w + \delta w_m) \left( e^{-h}\,|B_m\rangle\langle A_{m+1}| + e^{h}\,|A_{m+1}\rangle\langle B_m| \right)
+ \sum_{s=0}^{N-1} \delta \varepsilon_s |s\rangle \langle s| + i \gamma \Sigma_z + \widehat{H}_{{\rm bd},\theta}, \label{eq:ex2_H2_explicit}
\ee
where $v, w \in {\mathbb R}_{>0}$ are hopping amplitudes, $h \in {\mathbb R}$ means the nonreciprocity, and $\gamma \in {\mathbb R}$ controls the staggered gain/loss through the operator defined as
\be
\Sigma_z := \sum_{m=1}^{N/2} \left( |A_m\rangle\langle A_m| - |B_m\rangle\langle B_m| \right) = \bigoplus_{m=1}^{N/2} \sigma_z. \label{eq:ex2_sigmaz}
\ee
In addition, the boundary term $\widehat{H}_{{\rm bd},\theta}$ for PBC/OBC is given by
\be
\widehat{H}_{{\rm bd},\theta}  =
\begin{dcases}
  (w + \delta w_{N/2}) \left( e^{-h} e^{i\theta}\,|B_{N/2}\rangle\langle A_1| + e^{h} e^{-i\theta}\,|A_1\rangle\langle B_{N/2}| \right) & \text{for PBC} \\
  0 & \text{for OBC}
\end{dcases},
\label{eq:ex2_bdry_pbc_obc}
\ee
with a twisted parameter, $\theta \in [-\pi,\pi)$.
Furthermore, we introduce the disorder terms as
\be
 \text{Bond disorder ($V$-noise)} &:& \delta v_m = V \xi_{m,v}, \quad \delta w_m = V \xi_{m,w}, \label{eq:ex2_bond_disorder} \\
 \text{On-site disorder ($W$-noise)} &:& \delta \varepsilon_s = W \,\xi_s, \label{eq:ex2_onsite_disorder}
\ee
with $V,W \in {\mathbb R}_{\ge 0}$ and the uniform random variables, $\xi \sim U(-1,1)$.

Depending on the parameters $\gamma$, $V$, and $W$, the system possesses exact symmetries, ${\cal Q}_{{\cal P}^\prime}$ and ${\cal Q}_{\chi}$:
\be
   {\cal Q}_{{\cal P}^\prime} := {\cal P}^\prime, \qquad   {\cal Q}_{\chi} := \Sigma_{z}, \label{eq:QP_Qchi}
\ee
where the former is defined through an imaginary gauge transformation ${\cal S}$ and the standard spatial parity ${\cal P}$ as
\be
&&    {\cal P}^\prime := {\cal S} {\cal P} {\cal S}^{-1}, 
\ee
\be
&& {\cal P} =
   \begin{pmatrix}
     0 & 0 & \cdots & 0 & 1 \\
     0 & 0 & \cdots & 1 & 0 \\
     \vdots & \vdots & \ddots & \vdots & \vdots \\
     0 & 1 & \cdots & 0 & 0 \\
     1 & 0 & \cdots & 0 & 0 \\
   \end{pmatrix} \in {\mathbb R}^{N \times N}, \qquad {\cal S} = {\rm diag} (1,e^{h},e^{2h}, \cdots, e^{(N/2-1)h}) \otimes {\mathbb I}_2 \in {\mathbb R}^{N \times N},
\ee
and the latter is characterized by the staggered chiral operator $\Sigma_z$ given in Eq.\eqref{eq:ex2_sigmaz}.
In the clean limit, i.e., $\gamma = V = W = 0$, these symmetries exactly (anti-)commute with the Hamiltonians \eqref{eq:ex2_H1_explicit}\eqref{eq:ex2_H2_explicit}.
Crucially, however, the parity ${\cal Q}_{{\cal P}^\prime}$ is explicitly broken by turning on the bond disorder $V$, leading to the distinct symmetric conditions as
\be
 [ {\cal Q}_{{\cal P}^\prime}, \widehat{H}_{1,2} ] = 0 \quad \mbox{for} \quad \gamma = V = W = 0, \qquad \{ {\cal Q}_{\chi}, \widehat{H}_{1,2} \} = 0 \quad \mbox{for} \quad \gamma = W = 0. \label{eq:QP_Qchi_comm}
\ee
   
In our numerical experiments, we fix the following parameters:
\be
N=10, \qquad T=2.0, \qquad v=0.5, \qquad w=1.5, \qquad h=0.4. \label{eq:ex2_fixed_base}
\ee
The other numerical fixed parameters and the boundary condition are summarized in Tab.~\ref{tab:ex2_params}.
In our numerical experiments, we employ the following observables:
\be
&& O_{0} := |0 \rangle \langle 0|, \qquad O_{\rm stag} := \frac{1}{N} \Sigma_z, \nl
&& O_{A} := \sum_{m=1}^{N/2} | A_m \rangle \langle A_m|, \qquad O_{\rm break} := O_A + |A_1\rangle\langle B_1| + |B_1\rangle\langle A_1|. \label{eq:ex2_stag_obs}
\ee

\begin{table}[t]
\centering
\begin{tabular}{|p{0.425\linewidth} | p{0.465\linewidth}|}
\hline
Purpose / figure & Parameters \\
\hline
Eigenvalue exchange around the EP (Fig.~\ref{fig:ex2_ep_loops}) & $(V,W)=(0.0,0.2)$ for PBC. See the main text for $(\gamma,\theta)$. \\
Observable response near the EP (Fig.~\ref{fig:ex2_ep_visibility})  & $(V,W)=(0.0,0.2)$ for PBC. See the main text for $(\gamma,\theta)$. \\
Channel-selective winding readout (Fig.~\ref{fig:ex2_winding_readout}) & $(\gamma,V,W)=(0.4,0.0,0.2)$, $\theta \in [0, 2 \pi)$ (361 points) for PBC.\\
  Observable-dimension growth under disorder (Fig.~\ref{fig:ssh_signature_result}) & $\gamma = 0.0, (V,W) = (0.0,0.0), (0.2,0.0), (0.0,0.2)$ for OBC.\\
\hline
\end{tabular}
\caption{Numerical parameter choices used in our numerical experiments of the NHFSSH chain.}
\label{tab:ex2_params}
\end{table}

The purposes of the present example are as follows:
\begin{enumerate}
    \item We use an EP-accessible finite-size NHFSSH chain to test how a branch-sensitive spectral skeleton is filtered into observable-dependent visible readouts, from local EP-neighborhood responses to winding-related responses along a twist cycle.
    \item We use the sampled observable dimension $D_{\rm obs}={\rm rank} ({\cal K}_0)$ to quantify the finite-sampling visibility of the operator space under symmetry, disorder, and probe choice, without identifying this sampled rank with a sharp symmetry-wall saturation.
\end{enumerate}

\subsection{Numerical experiments}
\subsubsection{Local EP response and global winding Readout} \label{sec:obs_topology_visibility}

The central utility of the present framework lies in its ability to distinguish between the common spectral skeleton, or shadow-level structure, and the channel-dependent output (the visible level).
To demonstrate this, we first establish the local spectral geometry underlying the NHFSSH chain and then examine how this structure is manifested or hidden through different observable channels.

We begin by confirming that the system possesses a branch-sensitive singular geometry in its spectral skeleton.
In the translationally invariant and disorder-free limit ($N \to \infty, V=W=0$), the Bloch reduction of the Floquet problem yields a $2\times2$ monodromy matrix $M_k(\gamma)$ for each quasimomentum $k$. For the parameter set fixed in Eq.\eqref{eq:ex2_fixed_base}, this idealized limit exhibits an exact EP at
\be
k=\pi, \qquad \gamma=\gamma_{\ast} \approx 0.5675645, \label{eq:ex2_ep_location}
\ee
where the specific value $\gamma_{\ast}$ is numerically determined by the coalescence condition of the Floquet eigenvalues. In the actual system under consideration, i.e., a finite-size periodic chain ($N=10$) with finite disorder ($W=0.2$), this exact degeneracy is strictly lifted into a finite spectral gap due to finite-size effects and spatial inhomogeneities.
However, the topological monodromy, i.e., the eigenvalue exchange upon a parameter loop enclosing this lifted singularity, is robustly inherited as a signature of the underlying spectral skeleton.
This exact Bloch-sector EP provides the natural starting point for searching for inherited finite-size remnants in the $N=10$ PBC chain and for studying how they appear through observable-dependent channels.
Returning to the finite periodic system, we examine the neighborhood of $(\theta,\gamma)\approx(\pi,\gamma_\ast)$ by considering a closed loop in the parameter space:
\be
\theta( \tau ) = \pi + r_\theta \cos \tau, \qquad \gamma( \tau ) = \gamma_\ast + r_\gamma \sin \tau, \qquad \tau \in [0,2\pi), \label{eq:ex2_loop_def}
\ee
with $r_\theta=0.18$ and $r_\gamma=0.025$. As displayed in Fig.~\ref{fig:ex2_ep_loops}, tracking the eigenvalues continuously along this loop reveals a nontrivial permutation after one cycle, while two cycles return them to their original ordering.
This square-root monodromy is a clear signature of an EP-consistent eigenvalue exchange residing in the shadow level, which remains a common property of the monodromy matrix $M(\theta, \gamma)$ regardless of the observable choice.

\begin{figure}[t]
\centering
\includegraphics[width=.7\linewidth]{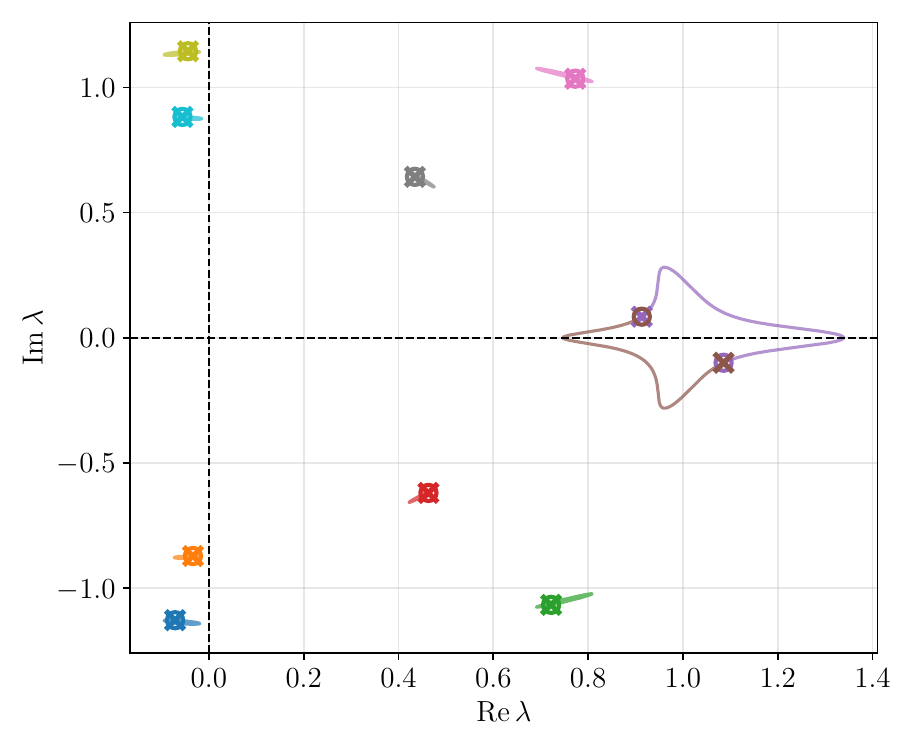}
\caption{Permutation and exchange of Floquet eigenvalues around the EP-accessible neighborhood in the finite-size periodic chain. Each eigenvalue is continuously tracked in the complex plane from $\tau=0$ (circles) to $\tau=2\pi$ (crosses), showing a non-trivial square-root Riemann-sheet topology inherent to the common spectral skeleton.}
\label{fig:ex2_ep_loops}
\end{figure}

The visibility of this shadow-level topology, however, is strongly dependent on the observable channel and the specific probe point.
To isolate this dependence, we fix the probe point at $z_{\rm ref}=0.6$ and compare the OSD response $h^{(O)}(z=z_{\rm ref})$ along the same closed loop defined in Eq.\eqref{eq:ex2_loop_def}.
In Fig.~\ref{fig:ex2_ep_visibility}, we compare the local site projector $O_0$ and $O_{\rm stag}$ in Eq.\eqref{eq:ex2_stag_obs}.
The comparison shows that the resulting trajectories in the complex plane do not produce a nontrivial net winding for either channel on this specific loop.
Even though the underlying eigenvalues are permuting at the shadow level, the visible readout does not produce a nontrivial winding-related response on this specific loop.
This mismatch occurs because the observable dressing can shift or deform the visible image of the spectral singularity relative to $z_{\rm ref}$, rendering the system's topological activity ``blind'' or invisible in these specific observable channels.
Here the point of the EP-accessible extension is not to use Fig.~\ref{fig:ex2_ep_visibility} as a direct proof of an EP mechanism for the later readout, but to show that the framework can cleanly distinguish between a branch-sensitive local spectral geometry and its observable-dependent visible response. This distinction is the key background for the global twist-cycle readout discussed below.

\begin{figure}[t]
\centering
\includegraphics[width=.9\linewidth]{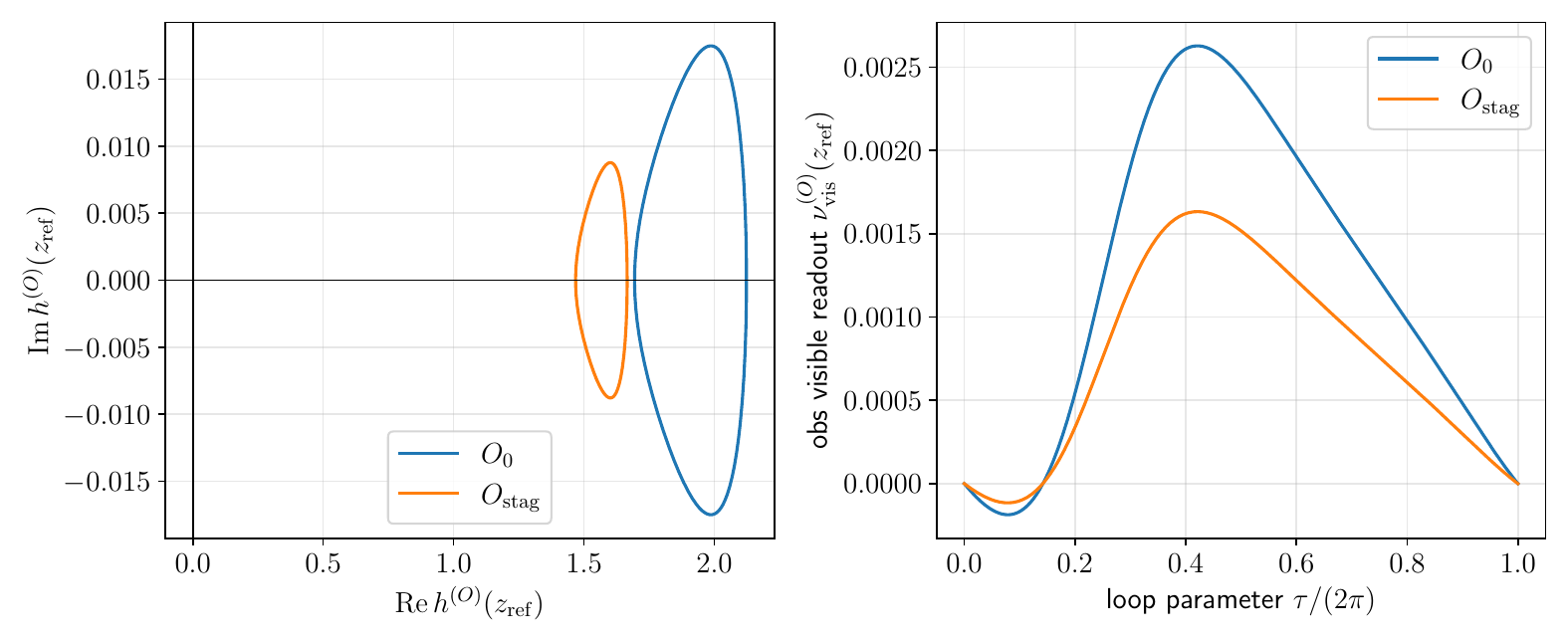}
\caption{Observable response evaluated from the OSD at $z_{\rm ref}=0.6$ near the inherited EP neighborhood. (Left) Trajectories of $h^{(O)}(z_{\rm ref})$ in the complex plane along the parameter loop for the local channel $O_0$ and the staggered channel $O_{\rm stag}$. (Right) The corresponding visible readout $\nu_{\rm vis}^{(O)}(z_{\rm ref})$. Notably, despite the square-root eigenvalue exchange occurring at the common spectral skeleton (shadow level), both channels yield a vanishing net winding, explicitly demonstrating how observable dressing can suppress the local branch-sensitive visible response.}
\label{fig:ex2_ep_visibility}
\end{figure}

\begin{figure}[t]
\centering
\includegraphics[width=.9\linewidth]{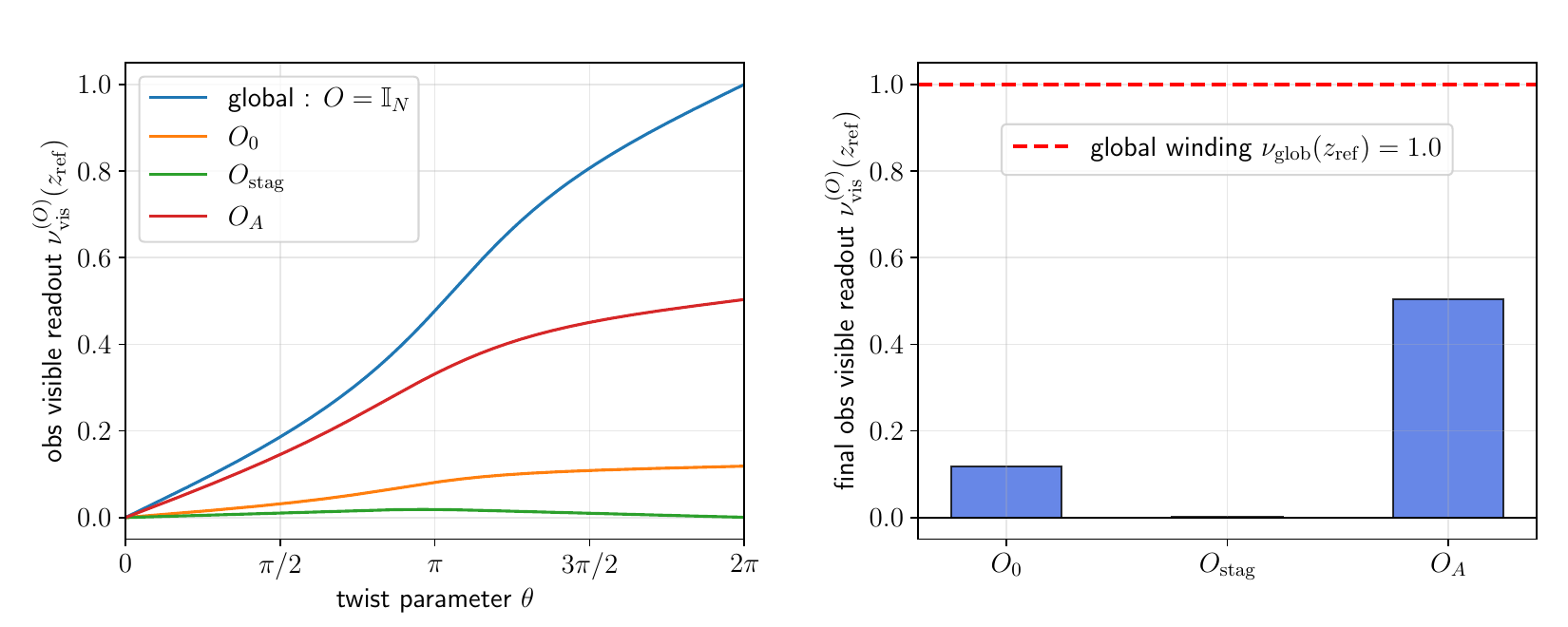}
\caption{Observable-visible readout of winding-related structure along the twist cycle $\theta:0\to2\pi$ under PBC for the point-gap reference $z_{\rm ref}=-0.2+1.0 i$. (Left) Trajectories of the cumulative phase $\frac{1}{2\pi} \Delta_\theta \arg h^{(O)}(z_{\rm ref})$ for the global channel, $O_0$, $O_{\rm stag}$, and $O_A$. (Right) The corresponding final winding readouts $\nu_{\rm vis}^{(O)}(z_{\rm ref})$, contrasting with the global integer invariant $\nu_{\rm glob}=1$ (dashed line). The results exhibit a sharp channel-selectivity, where global winding information is entirely suppressed in the staggered channel via destructive sublattice cancellation, while the collective probe $O_A$ successfully retains sufficient coherent weight to capture approximately half of the global winding.}
\label{fig:ex2_winding_readout}
\end{figure}

While the local response near the EP can be filtered, the global winding-related structure of the monodromy matrix along the full twist cycle $\theta:0\to 2\pi$ provides a more robust target for visibility analysis.
For this purpose, we choose the reference point inside the point-gap component, $z_{\rm ref}=-0.2+1.0i$, rather than the local reference-point choice used for the inherited EP-neighborhood loop.
We define the global winding accumulation $\nu_{\rm glob}(z_{\rm ref})$ and the observable-visible readout $\nu_{\rm vis}^{(O)}(z_{\rm ref})$ as
\be
\nu_{\rm glob}(z_{\rm ref}) :=   \frac{1}{2\pi}\,\Delta_\theta \arg h^{({\mathbb I}_N)}(z_{\rm ref}) \in {\mathbb Z}, \qquad \nu_{\rm vis}^{(O)}(z_{\rm ref}) &:=& \frac{1}{2\pi}\,\Delta_\theta \arg h^{(O)}(z_{\rm ref}) \in {\mathbb R}.
\label{eq:ex2_visible_winding} 
\ee
It is notable that whereas the global accumulation $\nu_{\rm glob}(z_{\rm ref})$ is strictly an integer, the visible readout $\nu_{\rm vis}^{(O)}(z_{\rm ref})$ is generally real-valued.
This fact arises because $\nu_{\rm vis}^{(O)}(z_{\rm ref})$ is an observable-dressed quantity governed by the fluctuating coefficients $c^{(O)}_j(\theta) = \langle \widetilde{u}_j(\theta) | O | u_j(\theta) \rangle$ in Eq.\eqref{eq:OTS_lam}, and thus, representing partial visibility rather than a strict topological invariant.
In Fig.~\ref{fig:ex2_winding_readout}, we compare the three observables, $O_0$, $O_{\rm stag}$, and $O_A$ in Eq.\eqref{eq:ex2_stag_obs}.
The results illustrate a sharp channel-selectivity. While the global readout exhibits a clear winding ($\nu_{\rm glob} = 1.0$), this spectral information is almost entirely suppressed in the staggered channel ($\nu_{\rm vis}^{(O_{\rm stag})} \approx 0.0006$) due to destructive interference between the sign-alternating sublattices. For the single-site channel $O_0$, the readout yields a non-universal fractional value ($\nu_{\rm vis}^{(O_0)} \approx 0.118$) reflecting localized amplitude fluctuations. 
In contrast, the collective channel $O_A$ yields a fractional readout $\nu_{\rm vis}^{(O_A)} \approx 0.503$. Since $O_A$ effectively projects onto half of the system's degrees of freedom, it captures roughly half of the global winding.
This demonstrates how the underlying spectral geometry is partially rendered in the visible object.
The slight deviation from an exact half-integer ($+0.003$) quantitatively reflects the observable dressing induced by the disorder ($W=0.2$).
Here the roles of the two reference points are different. For Fig.~\ref{fig:ex2_ep_visibility}, we used the representative local choice $z_{\rm ref}=0.6$ in order to see how the inherited EP-neighborhood response is filtered in specific observable channels. By contrast, for Fig.~\ref{fig:ex2_winding_readout} we choose $z_{\rm ref}$ inside the point-gap component of the full twist cycle, so that $\nu_{\rm glob}$ is well defined while the observable-visible readout remains channel dependent.

These observations support the following interpretation: in the finite-size NHFSSH chain, the relevant issue is not merely whether a winding-related spectral structure exists, but through which observable channel it becomes visible. Local or sign-alternating observables can undergo strong internal cancellation at the level of observable dressing. Conversely, suitable collective channels can retain sufficient coherent weight to recover a stable visible remnant of the same background spectral structure, sometimes only partially and sometimes more robustly, depending on the probe.
These results also show that the visibility of spectral features depends not only on the probe point $z_{\rm ref}$ but also strongly on the choice of the observable $O$. To quantify this constraint more systematically, we now turn to the algebraic dimension analysis, which clarifies how exact symmetries, finite-size degeneracies, and probe structure together shape the division between visible and invisible sectors.

\subsubsection{Observable-dimension growth under symmetry and disorder}
\label{sec:symmetry_signatureing_disorder}

We next examine how the observable dimension $D_{\rm obs}$ depends on the choice of observable and on the disorder profile in the finite NHFSSH chain.
In this subsection, $D_{\rm obs}$ is understood in the sampled sense introduced in Sec.~\ref{sec:reconstruct_Mex}: it is the rank of the sampled micromotion-expanded observable map, as in Eq.\eqref{eq:Dobs_K0}.
Here ${\cal K}_0$ is constructed by stacking the vectorized time-shifted observables defined in Eq.\eqref{eq:def_obs_dt_ext}.
Thus, $D_{\rm obs}$ should not be interpreted as the dimension of the full bicommutant algebra itself.
Rather, it measures how many independent operator directions are actually resolved by the finite set of sampled micromotion probes.
In general, it obeys the hierarchy discussed in Sec.~\ref{sec:reconstruct_Mex},
\be
D_{\rm obs}(Q) \leq \dim {\rm Span}\{\Xi\} \leq \dim({\frak O}_{\rm ext}^{\prime\prime}),
\nonumber
\ee
and the first inequality can be strict because of finite sampling, temporal aliasing, finite-size degeneracies, or probe-locality effects.

Throughout this subsection we work with the finite open chain (OBC) and set $\gamma=0$.
We compare three observables, $O_0$, $O_A$, and $O_{\rm break}$ in Eq.\eqref{eq:ex2_stag_obs}, and three disorder profiles:
the clean case $(V,W)=(0,0)$, the bond-disordered case $(V,W)=(0.2,0)$, and the on-site-disordered case $(V,W)=(0,0.2)$.
These three cases are chosen to separate three effects.
The clean chain retains strong finite-size structure and spatial regularity.
Bond disorder breaks this regularity while preserving the off-diagonal bipartite structure.
On-site disorder breaks the chiral structure itself.

In the clean finite chain, shown in Fig.~\ref{fig:ssh_signature_result}(a), the growth of $D_{\rm obs}(Q)$ is strongly limited for all three observables.
The local observable $O_0$ remains confined to a small visible subspace, reflecting the restricted spatial reach of a single-site probe.
The collective observable $O_A$ and the symmetry-breaking observable $O_{\rm break}$ access larger operator subspaces, but their growth is still substantially below the full observable dimension $D_{\rm obs}=N^2 = 100$.
This suppression should not be interpreted as a new symmetry law.
Rather, in the clean finite system, spatial regularity, finite-size degeneracies, and residual structure in the Floquet spectrum make the sampled micromotion orbit linearly dependent before it fills the algebraically allowed space.

The bond-disordered case is shown in Fig.~\ref{fig:ssh_signature_result}(b).
Bond disorder preserves the bipartite off-diagonal structure but lifts a substantial part of the accidental finite-size degeneracy and spatial regularity of the clean chain.
Consequently, the micromotion-expanded observable family spans a larger sampled operator space, especially for the collective probes $O_A$ and $O_{\rm break}$.
This behavior is consistent with the hierarchy in Eq.\eqref{eq:Dobs_span_bound}, i.e.
$D_{\rm obs}(Q) \leq \dim {\rm Span}\{\Xi\} \leq \dim({\frak O}_{\rm ext}^{\prime\prime})$.
The increase of $D_{\rm obs}$ should be viewed as disorder-induced lifting of linear dependencies in the sampled observable orbit, not as an automatic saturation of a symmetry-protected algebraic bound.

Fig.~\ref{fig:ssh_signature_result}(c) shows the on-site-disordered case.
Since on-site disorder breaks the chiral anti-commutation structure, one might expect it to remove symmetry-imposed restrictions.
However, the sampled value of $D_{\rm obs}(Q)$ is not determined by symmetry breaking alone.
It also depends on the chosen observable, the finite sampling of micromotion, temporal aliasing among Liouville-frequency components, and the observable dressing along the orbit.
For this reason, the on-site-disordered case need not produce a larger sampled rank than the bond-disordered case.
In the present finite-size setting, bond disorder can produce a larger observable dimension because it efficiently lifts clean-chain linear dependencies while still allowing the micromotion orbit of the chosen collective probes to explore many independent operator directions.

The main lesson of Figs.~\ref{fig:ssh_signature_result}(a)--\ref{fig:ssh_signature_result}(c) is therefore not that $D_{\rm obs}$ directly equals a symmetry-wall dimension, but rather that $D_{\rm obs}(Q)$ serves as a finite-sampling readout of algebraic visibility.
In principle, if one were to choose a completely random dense matrix as the initial observable, i.e., explicitly breaking all spatial and internal symmetries, the observable dimension would generically reach the theoretical maximum $D_{\rm obs} = N^2$ for sufficiently large $Q$. 
However, physically accessible observables in practical experiments, such as local densities ($O_0$) or sublattice polarizations ($O_A$), are highly structured and often partially overlap with the system's underlying symmetries. 
This structural overlap inherently traps the reconstructible operator space into a proper subalgebra, providing a concrete physical visualization of the strictly positive algebraic deficiency, $\Delta_{\rm obs} > 0$, formulated in Sec.~\ref{sec:prel_micro_symm}.
Beyond these exact algebraic upper bounds (Prop.~\ref{prop:symmetry_indicator}), the observed dimension further reflects finite-size degeneracies, probe locality, and disorder. 
Specifically, the persistently small $D_{\rm obs}$ for $O_0$ highlights the restriction imposed by local probing, while the larger values for $O_A$ and $O_{\rm break}$ demonstrate that collective probes can access broader visible sectors. 
Ultimately, in the finite NHFSSH chain, the observable dimension should be interpreted as a practical visibility measure controlled jointly by fundamental symmetries, disorder-induced lifting of linear dependence, and the structural choice of the observable.

\begin{figure}[t]
  \centering
  \begin{minipage}[t]{0.49\linewidth}
    \centering
    \includegraphics[width=\linewidth]{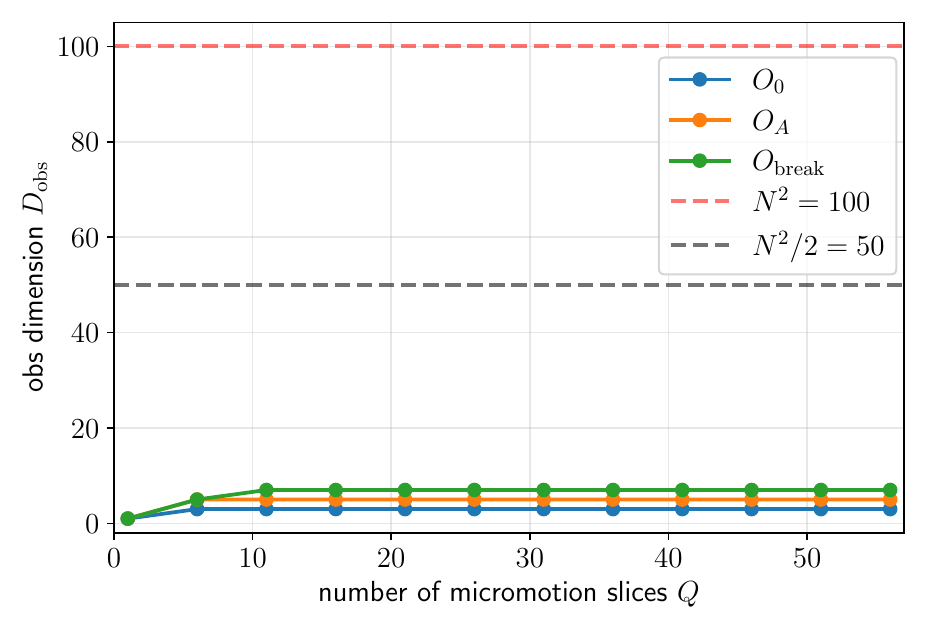}\\[-0.5em]
    \small (a) $V=W=0$
  \end{minipage}
  \hfill
  \begin{minipage}[t]{0.49\linewidth}
    \centering
    \includegraphics[width=\linewidth]{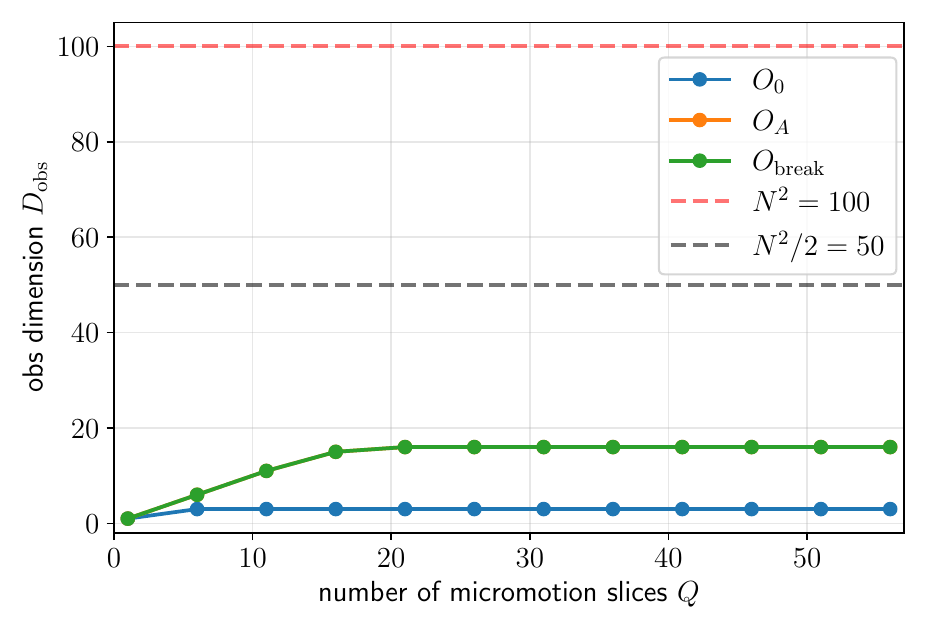}\\[-0.5em]
    \small (b) $V=0.2,\ W=0$
  \end{minipage}
  \vspace{0.8em}
  \begin{minipage}[t]{0.49\linewidth}
    \centering
    \includegraphics[width=\linewidth]{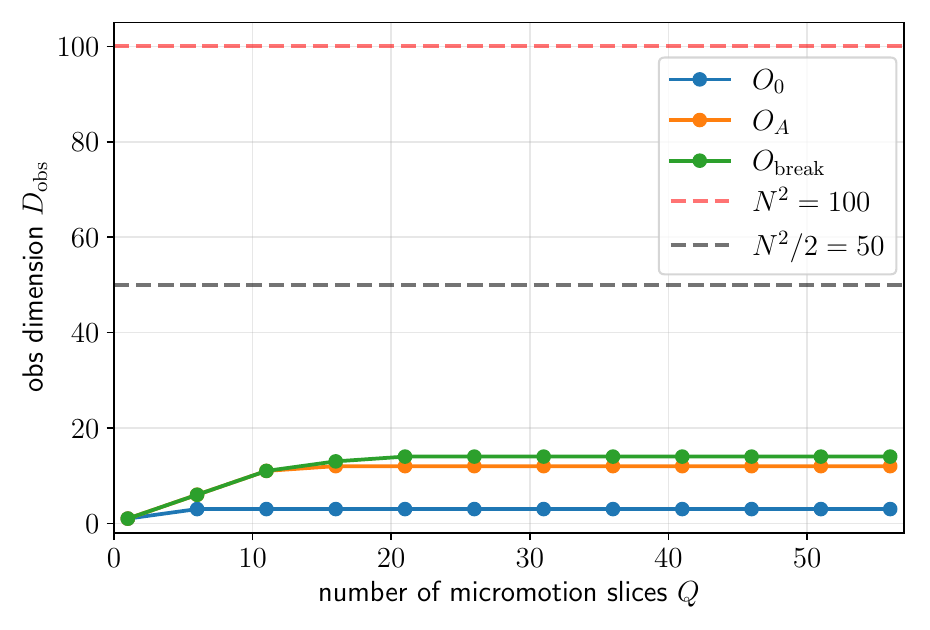}\\[-0.5em]
    \small (c) $V=0,\ W=0.2$
  \end{minipage}
  \caption{Sampled observable-dimension growth $D_{\rm obs}(Q)=\rank({\cal K}_0)$ across different disorder profiles in the open finite NHFSSH chain: (a) the clean system ($V=W=0$), (b) the bond-disordered configuration ($V=0.2, W=0$), and (c) the on-site-disordered configuration ($V=0, W=0.2$). The results highlight that while the single-site probe $O_0$ remains strictly bounded due to spatial locality, collective probes ($O_A, O_{\rm break}$) exploit disorder-induced lifting of accidental linear dependencies to map out significantly larger visible sectors of the operator state space.}
  \label{fig:ssh_signature_result}
\end{figure}

\section{Summary and Outlook} \label{sec:summary}

In this work, we formulated a framework of Floquet algebraic tomography for finite-dimensional monodromy matrices in $GL(N,\mathbb C)$, with an optional reduction to $SL(N,\mathbb C)$.
The starting point is that the observable trace sequence
\be
\zeta_n^{(O)}=\Tr(OM^n), \qquad n\in{\mathbb N}_0, \nn
\ee
is not an arbitrary discrete signal.
It is constrained by the characteristic data of the Floquet monodromy matrix $M$, and therefore the inverse problem is naturally a finite-dimensional algebraic reconstruction problem rather than a generic exponential fit.
The analytic side of the framework, introduced in Sec.~\ref{sec:CH_GLN_const}, is organized by the observable resolvent (ORS), observable spectral determinant (OSD), and observable Dirichlet spectral data (ODSD).
At the ORS level, the decomposition
\be
{\cal Z}^{(O)}(z)=\frac{{\cal N}^{(O)}(z)}{\Delta(z)}, \qquad \Delta(z)=\det({\mathbb I}_N-zM), \nn
\ee
makes the basic skeleton/dressing separation explicit: the denominator $\Delta(z)$ is common to all observables, while the numerator ${\cal N}^{(O)}(z)$ carries the observable-dependent dressing.

The reconstruction side of the framework was developed in Sec.~\ref{sec:GLNstructure}.
The Cayley--Hamilton recurrence and the associated Hankel matrices recover the effective recurrence order, characteristic coefficients, and finite-dimensional spectral skeleton from finite OTS data.
The same characteristic skeleton also constrains nonconsecutive sampling patterns through the arbitrary-stepsize Cayley--Hamilton extension.
In the multi-observable setting, the Hankel construction produces an equivalent realized monodromy matrix $M_{\rm rel}\sim M$, thereby recovering the similarity-invariant spectral skeleton under suitable visibility conditions.
We also introduced Liouville-space extensions, including the fundamental-type OTS $\zeta_n^{(O;\Omega)}=\Tr(OM^n\Omega)$ and the adjoint-type OTS $\zeta_n^{\rm adj(O;\Omega)}=\Tr(OM^n\Omega M^{-n})$, which connect the framework to MIMO realization, boundary/density-operator variation, and ratio-spectrum reconstruction.

A second structural aspect was clarified in Sec.~\ref{sec:prel_micro_symm}.
Even when finite OTS data determine a reconstructed object, the full monodromy matrix need not be identifiable from the accessible observable family.
If the observables generate only a proper algebra ${\frak O}\subsetneq \End({\cal H})$, Prop.~\ref{prop:Alg_rel_M} shows that the trace data determine a canonical visible representative
\be
M_{\rm vis}={\mathbb E}_{{\frak O}^{\prime\prime}}(M)\in{\frak O}^{\prime\prime}, \nn
\ee
together with a trace-invisible remainder.
Micromotion enlarges the observable algebra by replacing $O$ with the time-shifted family $O(t)=U(t,0)^{-1}OU(t,0)$, but Prop.~\ref{prop:symmetry_indicator} shows that exact commuting symmetries can still impose non-removable algebraic deficiencies.
In Sec.~\ref{sec:reconstruct_Mex}, we further described a stronger, calibrated reconstruction problem, in which a sufficiently rich observable family is explicitly known in a fixed physical basis, so that the similarity ambiguity can be fixed.

The two examples illustrate complementary aspects of this structure.
In Example~1, Sec.~\ref{sec:example_transmon}, the driven transmon qutrit provides a minimal $GL(3,\mathbb C)$ setting.
The single-observable reconstruction in Sec.~\ref{sec:alg_rec} verifies that the common spectral skeleton can already be recovered with high precision from the qubit-population OTS.
The ODSD analysis in Sec.~\ref{sec:phase_vis} then shows that the visible phase response is a filtered and cancellation-sensitive projection of the same spectral skeleton, rather than a universal proxy for the determinant-phase accumulation.
Finally, the observable-dimension analysis in Sec.~\ref{subsubsec:dtq3_rank_saturation} shows how coherent leakage is detected algebraically: the micromotion-expanded observable dimension $D_{\rm obs}(Q)=\rank{\cal K}_0(Q)$ remains one-dimensional in the isolated qubit limit, while leakage into the third level expands the sampled visible operator space to the full qutrit value $N^2=9$.
In Example~2, Sec.~\ref{sec:example_ssh}, the finite non-Hermitian Floquet SSH chain provides a more singular finite-size setting.
The EP-accessible extension analyzed in Sec.~\ref{sec:obs_topology_visibility} shows that branch-sensitive spectral geometry can be inherited from the Bloch-sector exceptional point, while the visible OSD response may be filtered or suppressed by observable dressing.
Along a full twist cycle, the same framework separates the global winding-related structure from the observable-visible readouts.
The observable-dimension analysis in Sec.~\ref{sec:symmetry_signatureing_disorder} then studies the sampled rank $D_{\rm obs}(Q)=\rank{\cal K}_0(Q)$ under clean, bond-disordered, and on-site-disordered conditions.
The resulting behavior should not be read as a direct saturation of a sharp symmetry-wall dimension.
Rather, it shows that sampled visibility is controlled jointly by exact symmetries, finite-size degeneracies, probe locality, temporal sampling, disorder-induced lifting of linear dependence, and observable choice.

Taken together, the present results show that Floquet algebraic tomography is not only a method for reconstructing finite-dimensional spectral data.
It also provides a language for organizing the distinction between common structure and observable visibility.
The common spectral skeleton, observable dressing, visible representative selected by the observable algebra, and residual shadow or blind-spot structure all appear naturally within the same framework.
The two examples realize complementary scenarios of restricted observability: leakage can enlarge the sampled visible operator space, while symmetry, locality, and finite-size structure can restrict or reshape the observable readout.

There are several natural directions for future work.
First, the present analysis is finite-dimensional and essentially noiseless at the structural level, so extending the framework to noisy and experimentally limited data will require robust rank selection, regularized reconstruction, and uncertainty quantification.
In practical applications, statistical fluctuations or experimental relaxation will inevitably blur the sharp rank-deficiency boundaries of the Hankel matrix, and a quantitative treatment of these effects is left for future work.
Second, the Liouville-space and adjoint-type formulations should be developed further for experimentally accessible expectation-value data, where the ratio spectrum is often more directly visible than the absolute spectrum.
Third, it would be valuable to apply the same skeleton/dressing and partial-realization viewpoint to larger non-Hermitian Floquet systems, where exceptional-point sensitivity, finite-size topology, and symmetry-constrained observability may coexist.
Finally, the present framework suggests a broader inverse-problem viewpoint on driven systems: rather than asking only whether a spectral or topological structure exists, one can ask through which observable algebra it becomes visible, which parts are reconstructed, and which parts survive only as shadow-level remnants.

\acknowledgments
S.~K is supported by JSPS KAKENHI Grant Nos.~22H05118 and 25K07298.

\appendix

\section{Technical details in the Hankel matrix analysis} \label{app:tech_hankel}
In this appendix, we supplement the algebraic and realization-theoretic methods formulated in Sec.~\ref{sec:GLNstructure} by detailing their reconstruction procedures and formal generalizations. 
First, App.~\ref{app:prony} reviews the classic Prony reconstruction procedure, which is utilized to systematically decouple the exact eigenvalues from the observable-dependent weights. 
Next, App.~\ref{app:arb_stepsize} provides the rigorous derivation of the arbitrary-stepsize Cayley--Hamilton interpolation, exploiting the algebraic structure of Schur polynomials to constrain nonconsecutive sampling patterns. 
Finally, App.~\ref{app:Mrel_proof} presents the definitive state-space factorization proof, establishing that the multi-observable block-Hankel realization yields a matrix $M_{\rm rel}$ that preserves the similarity class of the exact monodromy matrix.

\subsection{Prony reconstruction of the spectrum and observable coefficients} \label{app:prony}

In this part, we briefly explain how to reconstruct the spectrum $\{\lambda_j\}_{j=1}^N$ and the observable coefficients $\{c_j^{(O)}\}_{j=1}^N$ from the OTS once the rank $N$ and the characteristic coefficients $\{e_a\}_{a=1}^N$ have been identified. We restrict ourselves here to the diagonalizable case with simple spectrum.
The DPs case is obtained by grouping degenerate eigenvalues, while the EPs case requires the confluent Prony method and is not discussed here.

In order to recover the spectrum, we use the characteristic polynomial of $M$ in Eq.\eqref{eq:Delta_def}.
Equivalently, the eigenvalues $\lambda_j$ are the roots of the monic polynomial
\be
P(\lambda):=\lambda^N-e_1\lambda^{N-1}+e_2\lambda^{N-2}-\cdots+(-1)^N e_N.
\label{eq:charpoly_app}
\ee
Hence, once the coefficients $\{e_a\}_{a=1}^N$ are obtained from the CH recurrence, the spectrum is recovered by solving
\be
P(\lambda)=0.
\label{eq:charpoly_eq_app}
\ee
In the SS cases, the roots are not degenerated, i.e. $\lambda_{j_1}\neq \lambda_{j_2}$ for all $j_1\neq j_2$, and determine the spectrum of $M$, i.e. $\{ \lambda_j \}_j$, uniquely up to permutation.

Then, we consider recovery of the observable coefficients, $\{ c^{(O)}_j\}_j$, for the simple spectrum using Eq.\eqref{eq:OTS_lam}.
Once the eigenvalues $\{\lambda_j\}_{j=1}^N$ are known, the observable coefficients $\{c_j^{(O)}\}_{j=1}^N$ are determined from the first $N$ data points $\{\zeta_n^{(O)}\}_{n=0}^{N-1}$ by solving the Vandermonde system defined as
\be
\begin{pmatrix}
1 & 1 & \cdots & 1 \\
\lambda_1 & \lambda_2 & \cdots & \lambda_N \\
\lambda_1^2 & \lambda_2^2 & \cdots & \lambda_N^2 \\
\vdots & \vdots & \ddots & \vdots \\
\lambda_1^{N-1} & \lambda_2^{N-1} & \cdots & \lambda_N^{N-1}
\end{pmatrix}
\begin{pmatrix}
c_1^{(O)}\\
c_2^{(O)}\\
\vdots\\
c_N^{(O)}
\end{pmatrix}
=
\begin{pmatrix}
\zeta_0^{(O)}\\
\zeta_1^{(O)}\\
\vdots\\
\zeta_{N-1}^{(O)}
\end{pmatrix}. \label{eq:Vandermonde_app}
\ee
Since the eigenvalues are distinct, the Vandermonde matrix is invertible, and thus, the coefficients are uniquely determined.

In summary, in the SS case, the reconstruction proceeds in two steps:
\begin{enumerate}
\item determine $\{e_a\}_{a=1}^N$ from the CH recurrence,
\item solve Eq.\eqref{eq:charpoly_eq_app} for $\{\lambda_j\}_{j=1}^N$ and then solve Eq.\eqref{eq:Vandermonde_app} for $\{c_j^{(O)}\}_{j=1}^N$.
\end{enumerate}


This is precisely the classical Prony reconstruction in the present context.
The role of the annihilating polynomial is played by the characteristic polynomial $P(\lambda)$ in Eq.\eqref{eq:charpoly_app}, while the amplitudes are the observable coefficients, $\{ c_j^{(O)} \}_j$.
The essential difference from a purely phenomenological use of Prony's method is that, in our framework, the characteristic coefficients, $\{e_a\}_a$, are not introduced as abstract fitting parameters, but are identified with the characteristic data of the monodromy matrix itself.

For completeness, let us briefly comment on the non-SS cases.
For the DP cases, if $M$ is diagonalizable but the spectrum is degenerate, then the OTS still has the form
\be
\zeta_n^{(O)}=\sum_{j=1}^{K} \widetilde{c}_{j}^{(O)} \lambda_j^n, \qquad K < N,
\ee
where the coefficients $\widetilde{c}_{j}^{(O)}$ are sums over the degenerate eigenspaces.
In this case, the same reconstruction applies after replacing $N$ by the number
$K$ of distinct eigenvalues.
For the EP cases, if $M$ is non-diagonalizable, then the OTS contains polynomially dressed exponentials of the form in Eq.\eqref{eq:OTS_lam_eps}.
This structural property requires the confluent Prony method~\cite{WallNeuhauser1995FDM,MandelshtamTaylor1997HI}, or equivalently, a Jordan-block generalization of the multi-channel realization.
While the single-observable probe requires care due to these confluent polynomial structures, the full algebraic handling of the EP configuration—including its block-Hankel factorization and basis-resolved exact reconstruction—will be systematically developed and explicitly closed in Sec.~\ref{sec:multi_obs} and Sec.~\ref{sec:reconstruct_Mex}, respectively.

\subsection{Arbitrary-stepsize Cayley--Hamilton interpolation} \label{app:arb_stepsize}

In this appendix, we give the detailed formulas for the arbitrary-stepsize generalization of the CH recurrence used in Sec.~\ref{sec:MGF_HM}.

The ordinary CH recurrence for the OTS \eqref{eq:CH_bn_recurrence} consists of the one-step sequence
$\{\zeta^{(O)}_{n+N},\zeta^{(O)}_{n+N-1},\cdots,\zeta^{(O)}_{n}\}$.
Since only $N$ points in the OTS are independent, one may generalize Eq.\eqref{eq:CH_bn_recurrence} to arbitrary $(N+1)$ points
$\{\zeta^{(O)}_{n^{(0)}},\zeta^{(O)}_{n^{(1)}},\cdots,\zeta^{(O)}_{n^{(N)}}\}$
with
$n^{(0)}>n^{(1)}>\cdots>n^{(N)}\in{\mathbb N}_0$,
using the Schur polynomials.
The generalized recurrence is expressed by
\be
\zeta^{(O)}_{n^{(0)}} + \sum_{r=1}^N S^{(r)}_{\mu(\widetilde{\bf n})}(\{ e_a\}) \zeta^{(O)}_{n^{(r)}} = 0,
\qquad
S^{(r)}_{\mu(\widetilde{\bf n})}(\{ e_a\}) := (-1)^r \frac{s_{\mu^{(r)}(\widetilde{\bf n})}(\{ e_a\})}{s_{\mu(\widetilde{\bf n})}(\{ e_a\})},
\qquad
\widetilde{\bf n} = (n^{(0)},\cdots,n^{(N)}).
\label{eq:GLN_Schur_interp_app}
\ee
The relevant partitions are
\be
\mu(\widetilde{\bf n}) = (n^{(1)}-N+1,\ n^{(2)}-N+2,\ \cdots,\ n^{(N)}),
\label{eq:young_tab_mu}
\ee
and
\be
\mu^{(r)}(\widetilde{\bf n}) =
(n^{(0)}-N+1,\ n^{(1)}-N+2,\ \cdots,\ n^{(r-1)}-N+r,\ n^{(r+1)}-N+r+1,\ \cdots,\ n^{(N)}).
\label{eq:young_tab_muj}
\ee
The Schur polynomial $s_\mu(\{e_a\})$ is defined by the dual Jacobi--Trudi formula as
\be
s_{\mu(\widetilde{\bf n})}(\{ e_a\}) = \det_{1\le i,j\le \mu_1(\widetilde{\bf n})} \bigl(
e_{\mu_i^\prime(\widetilde{\bf n})-i+j} \bigr), \qquad e_0=1, \quad e_a=0 \ \ \text{for} \ \ a<0 \ \ \text{or} \ \ a>N,
\ee
where $\mu^\prime(\widetilde{\bf n})$ denotes the conjugate partition of $\mu(\widetilde{\bf n})$.
Notice that the $\widetilde{\bf n}$-dependence in $S^{(r)}_{\mu(\widetilde{\bf n})}(\{ e_a\})$ is relevant only to relative differences from $n^{(0)}$, i.e. $(n^{(r)}-n^{(0)})$ for $r\in\{1,\cdots,N\}$.

This gives the generalization in terms of arbitrary stepsize for the row of the Hankel matrix and the basis in Eq.\eqref{eq:def_HNEN} as
\be
\bm{\mathcal{H}}^{(O)}_{\{ \widetilde{\bf n}_j \}, N} :=
\begin{pmatrix}
\zeta^{(O)}_{n^{(0)}_0} & \zeta^{(O)}_{n^{(1)}_0} &\cdots&  \zeta^{(O)}_{n^{(N)}_0} \\
\zeta^{(O)}_{n^{(0)}_1} & \zeta^{(O)}_{n^{(1)}_1} &\cdots&  \zeta^{(O)}_{n^{(N)}_1} \\
\vdots & \vdots &\ddots& \vdots \\
\zeta^{(O)}_{n^{(0)}_N} & \zeta^{(O)}_{n^{(1)}_N} &\cdots&  \zeta^{(O)}_{n^{(N)}_N}
\end{pmatrix}
\in {\mathbb C}^{(N+1)\times (N+1)},
\qquad
{\bf S}_N(\{ e_a\}) :=
\begin{pmatrix}
1 \\
S^{(1)}_{\mu(\widetilde{\bf n})}(\{e_a\}) \\
\vdots \\
S^{(N)}_{\mu(\widetilde{\bf n})}(\{e_a\})
\end{pmatrix}
\in {\mathbb C}^{N+1},
\label{eq:def_HNEN_gen}
\ee
where $\widetilde{\bf n}^{(r)}=(n^{(r)}_0,\cdots,n^{(r)}_N)$ are constrained as
\be
n_j^{(0)}-n_j^{(r)}=k_r \in {\mathbb N},
\qquad
k_1<\cdots<k_N,
\qquad
n_{j_1}^{(0)}\neq n_{j_2}^{(0)}
\quad
\mbox{for all} \quad
j_1\neq j_2\in\{0,\cdots,N\}.
\label{eq:rel_dis_k}
\ee
The condition to find $N$ in Eq.\eqref{eq:cond_GLN} holds by replacing
$({\bf H}^{(O)}_{{\bf n},N},{\bf E}_N)$
with
$(\bm{\mathcal{H}}^{(O)}_{\{ \widetilde{\bf n}_j \},N},{\bf S}_N(\{e_a\}))$.

\subsection{Multi-observable realization and similarity} \label{app:Mrel_proof}

In this appendix, we give the detailed derivation of the multi-observable realization and the fact that the realized monodromy matrix is equivalent to the exact monodromy matrix, $M$, up to similarity transform.
Below, we assume the SS and EP cases.

We define a vector of the OTS with a fixed $n$ measured by $\{O_\ell\}_{\ell=1}^L$ as
\be
\bm{\zeta}^{({\bf O})}_n :=
\begin{pmatrix}
\zeta_n^{(O_1)} \\
\zeta_n^{(O_2)} \\
\vdots \\
\zeta_n^{(O_L)}
\end{pmatrix} \in{\mathbb C}^{L}, \qquad {\bf O}=\{O_1,\cdots,O_L\},
\ee
and generalize the Hankel matrix in Eq.\eqref{eq:def_HNEN} as
\be
\widetilde{\bf H}^{({\bf O})}_{{\bf n},N} := (\bm{\zeta}^{({\bf O})}_{{\bf n}+N},\bm{\zeta}^{({\bf O})}_{{\bf n}+N-1},\cdots,\bm{\zeta}^{({\bf O})}_{{\bf n}}) \in{\mathbb C}^{L(N+1)\times(N+1)}, \label{eq:HankelO_nk}
\ee
with
\be
\bm{\zeta}^{({\bf O})}_{{\bf n}+k} :=
\begin{pmatrix}
\bm{\zeta}^{({\bf O})}_{n_0+k} \\
\bm{\zeta}^{({\bf O})}_{n_1+k} \\
\vdots \\
\bm{\zeta}^{({\bf O})}_{n_N+k}
\end{pmatrix}
\in {\mathbb C}^{L(N+1)}. \label{eq:zetaO_nk}
\ee
This still satisfies
\be
\widetilde{\bf H}^{({\bf O})}_{{\bf n},N} \cdot {\bf E}_N = {\bf 0},
\ee
where ${\bf E}_N$ is the basis of the characteristic coefficients defined in Eq.\eqref{eq:def_HNEN}.
One can write down $\bm{\zeta}^{({\bf O})}_{{\bf n}+k}$ by the spectrum as
\be
\bm{\zeta}^{({\bf O})}_{{\bf n}+k} = \sum_{j=1}^{N}\lambda_j^k\,{\bf c}_j^{({\bf O})}\otimes \bm{\Lambda}_j^{\bf n}, \label{eq:zeta_lam_c_lam}
\ee
where
\be
{\bf c}_j^{({\bf O})} :=
\begin{pmatrix}
c_j^{(O_1)} \\
c_j^{(O_2)} \\
\vdots \\
c_j^{(O_L)}
\end{pmatrix} \in{\mathbb C}^{L}, \qquad \bm{\Lambda}_j^{\bf n} :=
\begin{pmatrix}
\lambda_j^{n_0} \\
\lambda_j^{n_1} \\
\vdots \\
\lambda_j^{n_N}
\end{pmatrix} \in{\mathbb C}^{N+1}.
\ee
It is worth noting the difference between a diagonalizable DP and a non-diagonalizable EP. 
While the expression in Eq.\eqref{eq:zeta_lam_c_lam} is perfectly valid for a DP, the presence of Jordan blocks at an EP introduces confluent Vandermonde matrices, modifying the expansion into a more convoluted sum\footnote{The corresponding form to Eq.\eqref{eq:zeta_lam_c_lam} for an EP is given by
\be
\bm{\zeta}^{({\bf O})}_{{\bf n}+k} = \sum_{j=1}^K \sum_{m=0}^{d_j-1} \sum_{s=0}^{m}
\begin{pmatrix}
  k \\
  s
\end{pmatrix}
\lambda_j^{k-s}\,{\bf c}_{j,m}^{({\bf O})}\otimes \bm{\Lambda}_{j,m-s}^{\bf n},
\ee
where
\be
{\bf c}_{j,m}^{({\bf O})} :=
\begin{pmatrix}
c_{j,m}^{(O_1)} \\
c_{j,m}^{(O_2)} \\
\vdots \\
c_{j,m}^{(O_L)}
\end{pmatrix} \in{\mathbb C}^{L}, \qquad \bm{\Lambda}_{j,r}^{\bf n} :=
\begin{pmatrix}
\binom{n_0}{r}\lambda_j^{n_0-r} \\
\binom{n_1}{r} \lambda_j^{n_1-r} \\
\vdots \\
\binom{n_N}{r} \lambda_j^{n_N-r}
\end{pmatrix} \in{\mathbb C}^{N+1}.
\ee
}.
However, remarkably, the subsequent algebraic framework of the realization theory and the proof of similarity do not rely on such explicit spectral decompositions. 
The realization argument below is most naturally formulated in coordinate-free state-space language and therefore does not require the explicit simple-spectrum expansion.

In order to find a realized monodromy matrix from the Hankel matrix $M_{\rm rel}$, we perform the singular value decomposition (SVD) to the Hankel matrix, which yields
\be
\widetilde{\bf H}^{({\bf O})}_{{\bf n},N}=U\Sigma V^\dagger,
\ee
with a rectangular diagonal matrix $\Sigma\in{\mathbb R}^{L(N+1)\times(N+1)}$ containing singular values,
$\sigma_1\ge \sigma_2\ge \cdots \ge \sigma_{N+1}=0$,
and unitary matrices $U\in{\mathbb C}^{L(N+1)\times L(N+1)}$ and $V\in{\mathbb C}^{(N+1)\times(N+1)}$.
By extracting the first $N$ parts from $\Sigma$, $U$, and $V$, and denoting them by $\Sigma_N={\rm diag}(\sigma_1,\cdots,\sigma_N)\in{\mathbb R}^{N\times N}$, $U_N\in{\mathbb C}^{L(N+1)\times N}$, and $V_N\in{\mathbb C}^{(N+1)\times N}$,
one obtains a shift matrix ${\cal M}\in{\mathbb C}^{N\times N}$ through
\be
\widetilde{\bf H}^{({\bf O})}_{{\bf n}+1,N} = U_N {\cal M} U_N^\dagger \widetilde{\bf H}^{({\bf O})}_{{\bf n},N} = U_N {\cal M} \Sigma_N V_N^\dagger.
\ee
Thus, one finds the realized monodromy matrix $M_{\rm rel}$ as
\be
M_{\rm rel} := \Sigma_N^{-1/2}{\cal M}\Sigma_N^{1/2} = \Sigma_N^{-1/2}U_N^\dagger \widetilde{\bf H}^{({\bf O})}_{{\bf n}+1,N} V_N\Sigma_N^{-1/2}.
\ee

Then, we show that the realized matrix $M_{\rm rel}$ is equivalent to the original monodromy matrix $M$ up to a similarity transformation.
According to linear realization theory\cite{HoKalman1966,Kailath1980}, the discrete sequence $\bm{\zeta}^{({\bf O})}_n$ can be factorized into a state-space representation.
Specifically, there exists a minimal realization given by a system matrix $A \in \mathbb{C}^{N \times N}$, an observation matrix $C \in \mathbb{C}^{L \times N}$, and an input vector $B \in \mathbb{C}^{N}$, such that
\be
\bm{\zeta}^{({\bf O})}_n = C A^n B. \label{eq:zeta_CAB}
\ee
For instance, in the SS case corresponding to Eq.\eqref{eq:zeta_lam_c_lam}, one can explicitly construct this by choosing $A = \mathrm{diag}(\lambda_1, \dots, \lambda_N)$, $B = (1, \dots, 1)^\top$, and $C = ({\bf c}_1^{({\bf O})}, \dots, {\bf c}_N^{({\bf O})})$.
In all cases, the minimal system matrix $A$ shares the exact same spectral skeleton as the original monodromy matrix, ensuring $A = W^{-1} M W$ for some non-singular matrix $W \in GL(N, \mathbb{C})$.

By defining the extended observability matrix ${\cal O}_{\bf n}$ and the controllability matrix ${\cal X}_N$ as
\be
{\cal O}_{\bf n} :=
\begin{pmatrix}
C A^{n_0} \\
C A^{n_1} \\
\vdots \\
C A^{n_N}
\end{pmatrix} \in{\mathbb C}^{L(N+1)\times N}, \qquad {\cal X}_N: = (A^N B, A^{N-1} B,\cdots, B)\in{\mathbb C}^{N\times(N+1)}, \label{eq:On_XN}
\ee
the Hankel matrix is exactly factored as
\be
\widetilde{\bf H}^{({\bf O})}_{{\bf n},N} = {\cal O}_{\bf n}{\cal X}_N \ \ ( = U_N\Sigma_N V_N^\dagger). \label{eq:H_OX}
\ee
Replacing ${\bf n}\to{\bf n}+1$ corresponds to shifting the sequence by one step, which yields
\be
\widetilde{\bf H}^{({\bf O})}_{{\bf n}+1,N} = {\cal O}_{\bf n} A {\cal X}_N.
\ee
Substituting this back into the definition of $M_{\rm rel}$, we obtain
\be
M_{\rm rel} = \Sigma_N^{-1/2}U_N^\dagger {\cal O}_{\bf n} A {\cal X}_N V_N\Sigma_N^{-1/2}.
\ee
By setting $U_N={\cal O}_{\bf n}R^{-1}$ and recalling Eq.\eqref{eq:H_OX}, the transformation matrix $R$ is uniquely determined as
\be
R=\Sigma_N V_N^\dagger {\cal X}_N^+ \in{\mathbb C}^{N\times N},
\ee
where ${\cal X}_N^+$ is the Moore-Penrose pseudo-inverse of ${\cal X}_N$, ${\cal X}_N^+={\cal X}_N^\dagger({\cal X}_N{\cal X}_N^\dagger)^{-1}$.
Then, the realized matrix can be algebraically simplified as
\be
M_{\rm rel} &=& \Sigma_N^{-1/2}({\cal O}_{\bf n}R^{-1})^\dagger {\cal O}_{\bf n} A {\cal X}_N V_N\Sigma_N^{-1/2} \nl
&=& \Sigma_N^{-1/2}(R^{-1})^\dagger R^\dagger R A {\cal X}_N V_N\Sigma_N^{-1/2}\qquad
(U_N^\dagger U_N={\mathbb I}_N \Leftrightarrow {\cal O}_{\bf n}^\dagger{\cal O}_{\bf n}=R^\dagger R) \nl
&=& \Sigma_N^{-1/2}R A R^{-1}\Sigma_N V_N^\dagger V_N\Sigma_N^{-1/2} \qquad \ \ \ \, ({\cal X}_N=R^{-1}\Sigma_N V_N^\dagger) \nl
&=& \Sigma_N^{-1/2}R A R^{-1}\Sigma_N^{1/2} = S^\prime A (S^\prime)^{-1}, \quad \ \  (S^\prime:=\Sigma_N^{-1/2}R) \label{eq:Mrel_SMSinv_app}
\ee
Since $M_{\rm rel}$ is similar to $A$, and $A$ is similar to the exact monodromy matrix $M$ ($A = W^{-1} M W$), it rigorously follows that
\be
M_{\rm rel} = S M S^{-1}, \qquad S := S^\prime W^{-1} \in GL(N, \mathbb{C}).
\ee
Therefore, the multi-observable Hankel realization guarantees that the reconstructed matrix $M_{\rm rel}$ belongs to the exact same similarity class as the original monodromy matrix $M$, preserving its similarity-invariant spectral skeleton.

The analysis described above also works for the arbitrary-stepsize OTS argued in App.~\ref{app:arb_stepsize}.
In this case, Eqs.\eqref{eq:HankelO_nk}\eqref{eq:zetaO_nk} are replaced with
\be
\widetilde{\bm{\mathcal{H}}}^{({\bf O})}_{\{ \widetilde{\bf n}_j \}, N} := (\bm{\zeta}^{({\bf O})}_{{\bf n}^{(0)}},\bm{\zeta}^{({\bf O})}_{{\bf n}^{(1)}},\cdots,\bm{\zeta}^{({\bf O})}_{{\bf n}^{(N)}}) \in {\mathbb C}^{L(N+1)\times(N+1)},
\ee
and
\be
\bm{\zeta}^{({\bf O})}_{{\bf n}^{(r)}} :=
\begin{pmatrix}
\bm{\zeta}^{({\bf O})}_{n_0^{(r)}} \\
\bm{\zeta}^{({\bf O})}_{n_1^{(r)}} \\
\vdots \\
\bm{\zeta}^{({\bf O})}_{n_N^{(r)}}
\end{pmatrix}
= \sum_{j=1}^{N} {\bf c}_j^{({\bf O})}\otimes \bm{\Lambda}_j^{{\bf n}^{(r)}} \in{\mathbb C}^{L(N+1)}. \label{eq:def_HNEN_multiO_gen}
\ee
By this modification, Eq.\eqref{eq:On_XN} changes as
\be
{\cal O}_{\bf n}
:=
\begin{pmatrix}
C^{({\bf O})} M^{n_0^{(0)}} \\
C^{({\bf O})} M^{n_1^{(0)}} \\
\vdots \\
C^{({\bf O})} M^{n_N^{(0)}}
\end{pmatrix}
\in{\mathbb C}^{L(N+1)\times N},
\qquad
{\cal X}_N:=({\bf x}_0,{\bf x}_{-k_1},\cdots,{\bf x}_{-k_N})\in{\mathbb C}^{N\times(N+1)},
\label{eq:On_XN_gen}
\ee
where $k_r$ is the relative distance from $n_j^{(0)}$ in Eq.\eqref{eq:rel_dis_k}, and one can in consequence derive the same result as Eq.\eqref{eq:Mrel_SMSinv_app} with a different specific form of $S$.

The rank-reduction for the DP cases can be immediately seen by constructing Eq.\eqref{eq:zeta_lam_c_lam}.
If $\lambda_d=\lambda_{d+p1}$, the sequence becomes
\be
\bm{\zeta}^{({\bf O})}_{{\bf n}+k} = \sum_{\substack{j=1\\j\neq d,d+1}}^N \lambda_j^k\,{\bf c}_j^{({\bf O})}\otimes \bm{\Lambda}_j^{\bf n} + \lambda_d^k\,({\bf c}_d^{({\bf O})}+{\bf c}_{d+1}^{({\bf O})}) \otimes \bm{\Lambda}_d^{\bf n}.
\ee
Since the eigenvectors are fixed by $M$, the choice of the observables ${\bf O}$ only alters the coefficient vectors. 
Consequently, the degenerate sector collapses into a single tensor product term $({\bf c}_d^{({\bf O})}+{\bf c}_{d+1}^{({\bf O})}) \otimes \bm{\Lambda}_d^{\bf n}$. 
Since it shares the identical basis $\bm{\Lambda}_d^{\bf n}$, this sector contributes exactly rank 1 to the Hankel matrix, regardless of the number of observables $L$. 
Therefore, the resulting realized monodromy matrix $M_{\rm rel}$ is still maximally rank $K$.

Notably, even if one does not find the specific form of the similarity transformation $S$, one can use $M_{\rm rel}$ to perform the standard spectral analysis by beginning with the OSD \eqref{eq:def_hz} with the unit observable, $O_{\rm rel}={\mathbb I}_N$.
In this case, the behavior of the spectrum can be directly read from $M_{\rm rel}$ without loss of information of $M$.

Finally, to complement the exact algebraic realization with a practical numerical diagnostic, we define the rank-truncation residual associated with the singular values $\sigma_k$ of the Hankel matrix. For an assumed effective rank $r$, this indicator is given by
\be
{\cal R}_{r}^{({\bf O})} := \frac{\sum_{k=r+1}^{K_{\max}} \sigma_k^2} {\sum_{k=1}^{K_{\max}} \sigma_k^2}, \label{eq:def_rank_tail_residual}
\ee
where $K_{\max}$ is the total number of singular values. Rather than signifying mere numerical truncation noise, a large residual tail serves as a realization-theoretic measure of the structural visibility limitations formulated in Sec.~\ref{sec:prel_micro_symm}. Comparing ${\cal R}_{r}^{({\bf O})}$ across different observable channels thus provides a compact tool to distinguish whether informational deficiencies stem from unstructured noise or from a genuine, channel-dependent compression of the common spectral skeleton.

\section{Proofs of Propositions} \label{app:proof_props}
In this appendix, we provide the complete proofs of Props.~\ref{prop:Alg_rel_M} and \ref{prop:symmetry_indicator} discussed in Sec.~\ref{sec:prel_micro_symm}. 
Since the physical setups under consideration are entirely finite-dimensional, we work throughout within the rigorous framework of finite-dimensional unital $*$-subalgebras of $\mathrm{End}(\mathcal{H})$ equipped with the Frobenius inner product,
\be
\langle A,B\rangle_F:=\Tr(A^{\dagger}B).
\ee
To make the operator-algebraic mechanics self-contained, we first establish auxiliary lemmas regarding canonical orthogonal projections and conditional expectations before proceeding to the individual proofs. 
Specifically, App.~\ref{app:proof_prop1} details the precise visible/invisible splitting of the monodromy matrix relative to the observable algebra, and App.~\ref{sec:symmetry_signature} completes the proof showing that exact commuting symmetries can enforce a non-vanishing algebraic deficiency under continuous micromotion extension.

\subsection{Proof of Prop.~\ref{prop:Alg_rel_M}} \label{app:proof_prop1}
We begin with two lemmas that make explicit the visible/invisible decomposition associated with the observable algebra.

\begin{lemma}[Finite-dimensional visible decomposition] \label{lem:finite_vis}
  Let ${\frak O} \subset {\rm End}({\cal H})$ be a finite-dimensional unital $*$-subalgebra.
  Then, there exists an orthogonal decomposition given by
\be
{\cal H} = \bigoplus_{j=1}^{J} \left( {\cal H}_j \otimes {\cal M}_j \right)
\ee
such that
\be
{\frak O} \cong \bigoplus_{j=1}^{J} \left( {\rm End}({\cal H}_j) \otimes \mathbb I_{m_j} \right),
\qquad
{\frak O}^{\prime} \cong \bigoplus_{j=1}^{J} \left( \mathbb I_{d_j} \otimes {\rm End}({\cal M}_j) \right),
\ee
where $d_j:=\dim {\cal H}_j$ and $m_j:=\dim {\cal M}_j$.
In particular, only the block-diagonal components contribute to traces against ${\frak O}$, and for any $X \in {\rm End}(\cal H)$ and any $\sigma\in {\frak O}$ one finds
\be
\Tr(\sigma X)=\sum_{j=1}^{J} \Tr_{{\cal H}_j} \left(\sigma_j  \widetilde{X}_j \right), \qquad \widetilde{X}_j:=\Tr_{{\cal M}_j}(\Pi_j X \Pi_j), \label{eq:appE_trace_reduction}
\ee
where $\Pi_j$ is the orthogonal projection onto ${\cal H}_j\otimes {\cal M}_j$.
\end{lemma}
\begin{proof}
  This is the standard finite-dimensional structure theorem for unital $*$-subalgebras (Artin--Wedderburn form, see, e.g., Ref.~\cite{BratteliRobinson}).
  In this decomposition, all $\sigma\in {\frak O}$ take the following form:
\be
\sigma = \bigoplus_{j=1}^{J} \left(\sigma_j \otimes {\mathbb I}_{m_j}\right).
\ee
By writing $X=\sum_{j,k=1}^J \Pi_j X \Pi_k$, one has
\be
\Tr(\sigma X)=\sum_{j,k=1}^J \Tr(\sigma \Pi_j X \Pi_k)=\sum_{j=1}^{J} \Tr(\sigma \Pi_j X \Pi_j),
\ee
because the cross-terms vanish by orthogonality.
Taking the partial trace over the multiplicity space gives Eq.\eqref{eq:appE_trace_reduction}.
\end{proof}

\begin{remark}
  Notice that the number of blocks $J$ and their dimensions $d_j$ and $m_j$ are entirely and a priori determined by the algebraic structure of ${\frak O}$.
  They do not depend on the dynamics of the system or the specific monodromy matrix $M$. Instead, they strictly represent how the chosen set of observables ${\bf O}$ algebraically partitions the Hilbert space.
  Concretely, in the basis adapted to this decomposition, the projection $\Pi_j$ and any observable $\sigma \in {\frak O}$ take the following explicit block-diagonal forms:
\be
\Pi_j = {\rm diag}\bigl( 0, \dots, 0, {\mathbb I}_{{\cal H}_j \otimes {\cal M}_j}, 0, \dots, 0 \bigr), \qquad \sigma = {\rm diag}\bigl(\sigma_1 \otimes {\mathbb I}_{m_1}, \dots, \sigma_J \otimes {\mathbb I}_{m_J}\bigr).
\ee
Thus, $\Pi_j$ acts as a block-selector, algebraically masking out all sectors except the $j$-th composite space.
\end{remark}

\begin{lemma}[Canonical bicommutant projection] \label{lem:can_bicom}
In the notation of Lem.~\ref{lem:finite_vis}, define
\be
   {\mathbb E}_{{\frak O}^{\prime\prime}}(X) := \bigoplus_{j=1}^{J} \left( 
   \Tr_{{\cal M}_j}(\Pi_j X \Pi_j) \otimes \frac{{\mathbb I}_{m_j}}{m_j} \right). \label{eq:appE_conditional_expectation}
\ee
Then, ${\mathbb E}_{{\frak O}^{\prime\prime}} : {\rm End}({\cal H})\to {\frak O}^{\prime\prime}$ is the Frobenius-orthogonal projection onto ${\frak O}^{\prime\prime}$.
In particular,
\be
\langle X - {\mathbb E}_{{\frak O}^{\prime\prime}}(X), Y \rangle_F = 0, \qquad \forall Y\in {\frak O}^{\prime\prime}. \label{eq:appE_projection_orthogonality}
\ee
\end{lemma}
\begin{proof}
Any $Y \in {\frak O}^{\prime\prime}$ can be expressed by
\be
Y = \bigoplus_{j=1}^{J} \left(Y_j \otimes {\mathbb I}_{m_j}\right).
\ee
By using the partial-trace identity, one finds
\be
\Tr \left(Y^{\dagger} X \right) =\sum_{j=1}^{J} \Tr_{{\cal H}_j} \left(Y_j^{\dagger}  \Tr_{{\cal M}_j}(\Pi_j X \Pi_j)\right).
\ee
On the other hand, Eq.\eqref{eq:appE_conditional_expectation} gives
\be
\Tr \left(Y^{\dagger} {\mathbb E}_{{\frak O}^{\prime\prime}}(X)\right) = \sum_{j=1}^{J} \Tr_{{\cal H}_j} \left(Y_j^{\dagger}  \Tr_{{\cal M}_j}(\Pi_j X \Pi_j)\right),
\ee
because $\Tr_{{\cal M}_j}(\mathbb I_{m_j})=m_j$.
Therefore,
\be
\Tr \left(Y^{\dagger}[X-{\mathbb E}_{{\frak O}^{\prime\prime}}(X)]\right) = 0, \qquad \forall Y\in {\frak O}^{\prime\prime},
\ee
which proves Eq.\eqref{eq:appE_projection_orthogonality}. Since $\mathbb E_{{\frak O}^{\prime\prime}}(X)\in {\frak O}^{\prime\prime}$ by construction, it is the Frobenius-orthogonal projection.
\end{proof}
\, \\
\noindent
Then, we prove Prop.~\ref{prop:Alg_rel_M}.
\begin{proof}[Proof of \ref{prop:p1-1}]
  By assumption, ${\frak O}\subsetneq {\rm End}({\cal H})$ is a proper unital $*$-subalgebra. By Lem.~\ref{lem:finite_vis}, the Hilbert space decomposes as
\be
{\cal H} = \bigoplus_{j=1}^{J} \left( {\cal H}_j \otimes {\cal M}_j \right),
\ee
with either multiple blocks ($J>1$) or a nontrivial multiplicity space ($m_j>1$ for some $j$).
The family of restricted OTSs depends only on the reduced block data appearing in Eq.\eqref{eq:appE_trace_reduction}, i.e., in general it cannot distinguish arbitrary operators acting inside multiplicity spaces or mixing sectors invisible to ${\frak O}$.
Hence, the full monodromy matrix $M$ is not identifiable from the family of restricted OTSs in general. 
\end{proof}
\begin{proof}[Proof of \ref{prop:p1-2}]
Applying Eq.\eqref{eq:appE_trace_reduction} to $X=M$, we obtain reduced block matrices as
\be
\widetilde{M}_j:=\Tr_{{\cal M}_j}(\Pi_j M \Pi_j),
\ee
such that
\be
\Tr(\sigma M)=\sum_{j=1}^{J} \Tr_{{\cal H}_j}(\sigma_j \widetilde{M}_j).
\label{eq:appE_M_trace_reduction}
\ee
Since $\sigma_j$ ranges over all of ${\rm End}({\cal H}_j)$, the non-degeneracy of the Frobenius pairing on ${\rm End}({\cal H}_j)$ implies that each reduced matrix $\widetilde{M}_j$ is uniquely determined by the family of restricted OTSs.
By defining
\be
M_{\rm vis}:= \bigoplus_{j=1}^{J} \left( \widetilde{M}_j \otimes \frac{{\mathbb I}_{m_j}}{m_j} \right), \label{eq:appE_Mvis_definition}
\ee
and comparison with Eq.\eqref{eq:appE_conditional_expectation}, one finds 
\be
M_{\rm vis} = {\mathbb E}_{{\frak O}^{\prime\prime}}(M).
\ee
Hence, $M_{\rm vis}\in {\frak O}^{\prime\prime}$ is the Frobenius-orthogonal projection of $M$ onto ${\frak O}^{\prime\prime}$ and is uniquely determined by the family of restricted OTSs.
\end{proof}
\begin{proof}[Proof of \ref{prop:p1-3}]
Set
\be
\Delta M:=M-M_{\rm vis}.
\ee
By Lem.~\ref{lem:can_bicom}, one has $\Delta M \perp_F {\frak O}^{\prime\prime}$.
Since ${\frak O}\subset {\frak O}^{\prime\prime}$, this implies
\be
\Tr(\sigma^{\dagger} \Delta M)=0, \qquad \forall \sigma \in {\frak O}.
\ee
Since ${\frak O}$ is a $*$-algebra, the map $\sigma \mapsto \sigma^\dagger$ is a bijection on ${\frak O}$.
Thus, one equivalently obtains
\be
\Tr(\sigma\, \Delta M)=0,  \qquad \forall \sigma \in {\frak O}.
\ee
In particular, applying this to the initial observables $O \in {\bf O}$, one finds $\Tr(O \Delta M)=0$. Thus, $\Delta M$ is strictly trace-invisible relative to the observable algebra. The same construction can be applied to each power $M^n$ separately, yielding the corresponding visible representative ${\mathbb E}_{{\frak O}^{\prime\prime}}(M^n)$ and invisible remainder $M^n-{\mathbb E}_{{\frak O}^{\prime\prime}}(M^n)$.
\end{proof}

\subsection{Proof of Prop.~\ref{prop:symmetry_indicator}} \label{sec:symmetry_signature}

We next prove Prop.~\ref{prop:symmetry_indicator}. The key point is that micromotion enlarges the observable algebra, but exact commuting symmetries survive this enlargement.
In order to make a proof for the proposition, let us consider the following lemma:
\begin{lemma}[Micromotion preserves commuting symmetries] \label{lem:micro_preserve}
Let
\be
O(t):=U(t,0)^{-1} O\, U(t,0).
\ee
If an operator ${\cal Q}$ satisfies
\be
[\mathcal Q,U(t,0)]=0 \quad \text{for all } t\in [0,T),
\qquad
[\mathcal Q,O]=0,
\ee
then,
\be
[\mathcal Q,O(t)]=0 \qquad \text{for all } t\in [0,T).
\label{eq:appE_micromotion_commutation}
\ee
Consequently, ${\cal Q}$ commutes with all elements in ${\frak O}_{\rm ext}$ and therefore belongs to ${\frak O}_{\rm ext}^{\prime}$.
\end{lemma}
\begin{proof}
Using the commutation of ${\cal Q}$ with $U(t,0)$ and $O$, we compute
\be
\mathcal Q O(t) =\mathcal Q U(t,0)^{-1} O\, U(t,0) =U(t,0)^{-1} \mathcal Q O\, U(t,0) =U(t,0)^{-1} O\, \mathcal Q U(t,0) =O(t)\mathcal Q.
\ee
Thus, Eq.\eqref{eq:appE_micromotion_commutation} holds. 
\end{proof}
\, \\ \noindent
Then, we prove Prop.~\ref{prop:symmetry_indicator}.
\begin{proof}[Proof of \ref{prop:p2-1}]
  By construction, the extended family of OTSs can depend only on the component of the dynamics visible inside the algebra generated by the time-shifted observables.
Therefore, the number of linearly independent reconstructible components cannot exceed the dimension of the corresponding bicommutant,
\be
D_{\rm obs}\le \dim({\frak O}_{\rm ext}^{\prime\prime})\le N^2,
\ee
which proves the bound.
\end{proof}
\begin{proof}[Proof of \ref{prop:p2-2}]
  This corresponds to Lem.~\ref{lem:micro_preserve}.
  In particular, any nontrivial symmetry ${\cal Q}$ commuting with both the dynamics and the initial observable, $O$, survives the micromotion shift and belongs to ${\frak O}_{\rm ext}^{\prime}$.
\end{proof}
\begin{proof}[Proof of \ref{prop:p2-3}]
  Since ${\cal Q}\in {\frak O}_{\rm ext}^{\prime}$ and ${\cal Q}$ is non-scalar, the commutant ${\frak O}_{\rm ext}^{\prime}$ is strictly larger than the scalar center $\{c\,\mathbb I_N\}_{c\in\mathbb C}$.
  In finite dimensional cases, the full matrix algebra ${\rm End}(\mathcal H)$ has only scalar commutant.
  Therefore, ${\frak O}_{\rm ext}^{\prime\prime}$ cannot coincide with $\mathrm{End}(\mathcal H)$ and has to be a proper subspace as
\be
{\frak O}_{\rm ext}^{\prime\prime}\subsetneq {\rm End}(\mathcal H).
\ee
Hence,
\be
\Delta_{\rm obs}:=N^2-\dim({\frak O}_{\rm ext}^{\prime\prime})>0.
\ee
Therefore, an exact commuting symmetry enforces a positive algebraic deficiency even after micromotion extension.
\end{proof}

\bibliographystyle{utphys}
\bibliography{floquet_tomography} 

\end{document}